\documentclass[10pt,twocolumn,twoside] {IEEEtran}
\makeatletter
\def\ps@headings{%
\def\@oddhead{\mbox{}\scriptsize\rightmark \hfil \thepage}%
\def\@evenhead{\scriptsize\thepage \hfil \leftmark\mbox{}}%
\def\@oddfoot{}%
\def\@evenfoot{}}

\makeatother 

\IEEEoverridecommandlockouts
\usepackage{bbm}
\usepackage{amsfonts}
\usepackage[dvips]{graphicx}
\usepackage{times}
\usepackage{cite}
\usepackage{amsmath}
\usepackage{array}
\usepackage{amssymb}

\usepackage{stfloats}
\usepackage{slashbox}
\usepackage{graphicx}
\usepackage{footnote}
\usepackage{booktabs}
\usepackage{array}
\usepackage{algorithm}
\usepackage{subeqnarray}
\usepackage{cases}
\usepackage{threeparttable}
\usepackage{color}
\usepackage{hyperref}
\usepackage{epstopdf}
\usepackage{algpseudocode}
\usepackage{bm}
\usepackage{subcaption}
\usepackage{multirow}
\usepackage{adjustbox}
\usepackage[labelformat=simple]{subcaption}
\newcommand{\tabincell}[2]{\begin{tabular}{@{}#1@{}}#2\end{tabular}}

\newtheorem{theorem}{Theorem}
\newtheorem{remark}{Remark}
\newtheorem{proposition}{Proposition}
\newtheorem{lemma}{Lemma}
\newtheorem{definition}{Definition}
\newtheorem{problem}{Problem}

\usepackage{xcolor}

\newenvironment{answer}{%
   \color{black}
}{}
\begin{document}

\title{Optimal Dynamic Multicast Scheduling for Cache-Enabled Content-Centric  Wireless Networks}

\author{Bo~Zhou, Ying~Cui,~\IEEEmembership{Member,~IEEE}, and Meixia~Tao,~\IEEEmembership{Senior~Member,~IEEE}
\thanks{This paper has been presented in part at IEEE ISIT 2015.

B.~Zhou, Y.~Cui and M.~Tao are with the Department of
Electronic Engineering at Shanghai Jiao Tong University, Shanghai,
200240, P. R. China. Email: \{b.zhou, cuiying, mxtao\}@sjtu.edu.cn.}
}

\maketitle

%\vspace{-1.5cm}
\begin{abstract}
Caching and multicasting at base stations are two promising approaches to support massive content delivery over wireless networks.
However, existing scheduling designs do not make full use of the advantages of the two approaches.
In this paper, we consider the optimal dynamic multicast scheduling to jointly minimize the average delay, power, and  fetching costs for cache-enabled content-centric wireless networks. We formulate this stochastic optimization problem as an infinite horizon average cost Markov decision process (MDP). It is well-known to be a difficult problem \textcolor{black}{due to the curse of dimensionality,} and there generally only exist numerical solutions. By using \emph{relative value iteration algorithm} and the special structures of the request queue dynamics, we analyze the properties of the value function and the state-action cost function of the MDP for both the uniform and nonuniform channel cases. Based on these properties, we show that the optimal policy, which is adaptive to the request queue state, has a switch structure in the uniform case and a partial switch structure in the nonuniform case. Moreover, in the uniform case with two contents, we show that the switch curve is monotonically non-decreasing.
Then, by exploiting these structural properties of the optimal policy, we propose two low-complexity optimal algorithms.
\textcolor{black}{Motivated by the switch structures of the optimal policy,} to further reduce the complexity, we also propose a low-complexity suboptimal policy, which possesses \textcolor{black}{similar structural properties to} the optimal policy, and  develop a low-complexity algorithm to compute this policy.
% \textcolor{black}{Our analytical approach is also applicable to the Markov-modulated request arrival model.}
%The optimality properties obtained in this paper can also provide design insights for practical networks.
\end{abstract}
\section{Introduction}\label{sec:introduction}
The demand for wireless communication services has been shifting from connection-centric communications such as, traditional voice telephony and messaging to content-centric communications such as video streaming, social networking, and content sharing.
Moreover, the wireless data traffic is expected to grow at a compound annual growth rate of 57 percent from 2014 to 2019, reaching 24.3 exabytes per month by 2019\cite{Cisco}.
These phenomena propel the development of content-centric wireless networks\cite{liuhui}.

Recently, to support the dramatic growth of the wireless data traffic, caching at base stations (BSs) has been proposed as a promising approach for massive content delivery and extensively studied in the literature\cite{femto,TassiulasTCOM,mimo,Bastug2015}.
Specifically, \textcolor{black}{in \cite{femto}, the authors introduce the concept of \emph{Femtocaching} and study content placement at the small BSs to minimize the average content access delay}.
In \cite{TassiulasTCOM}, the authors consider joint request routing and caching in small-cell networks, and propose approximate algorithms to maximize the requests served by small BSs.
% References \cite{caching_LeiYing} and \cite{vip} consider content distribution networks and named data networks, respectively, and propose throughput-optimal joint request routing and caching algorithms to support the maximum number of requests.
In \cite{mimo}, the authors study the joint optimization of cache control and playback buffer management for video streaming in multi-cell multi-user MIMO cellular networks.
\textcolor{black}{Reference\cite{Bastug2015} analyzes the performance (e.g., the outage probability and the average delivery rate) of cache-enabled small-cell networks for given caching strategies.}
\textcolor{black}{However, in most existing literature \cite{femto,TassiulasTCOM,mimo,Bastug2015}, point-to-point unicast transmission is considered, which can only help to reduce the backhaul burden without effectively relieving the ``on air'' congestion.
The inherent broadcast nature of wireless medium is not fully exploited, which is the major distinction of  wireless communications from wired communications.}

On the other hand, enabling multicast service at BSs is an efficient way to deliver contents to multiple requesters simultaneously by effectively utilizing the inherent broadcast nature of  wireless medium \cite{embms}.
References \cite{HouTON} and \cite{multicast_MIT} consider scheduling problems for multicasting inelastic flows (with strict deadlines) in wireless networks.
 In \cite{multicast_capacity}, the authors study the asymptotic capacity of delay-constrained multicast in large scale mobile ad hoc networks.

% Since the popular contents cached in the BSs are more likely to be requested by multiple users, multicast could better exploit the potential of caching.

\textcolor{black}{
In view of the benefits of caching and multicasting, the joint design of the two promising techniques is expected to achieve superior performance for massive content delivery in wireless networks\cite{7249208,TWC16,ton}.
% Many content-centric applications, such as video streaming, are delay-sensitive, and it is critical to consider delay performance in cache-enabled content-centric wireless networks
In specific, the authors in \cite{7249208} study coded multicasting for inelastic services under a given coded caching scheme in a single-cell network.}
\textcolor{black}{In \cite{TWC16}}, the authors consider multicasting for inelastic services in cache-enabled small-cell networks.
\textcolor{black}{An approximate caching algorithm with performance guarantee and a heuristic caching algorithm are proposed to reduce the service cost of a fixed multicast transmission strategy.}
% A heuristic caching algorithm is proposed to reduce the service cost.
In \cite{ton}, the authors consider  multicasting for inelastic services in cache-enabled multi-cell networks.
A joint throughput-optimal caching and scheduling algorithm is proposed to maximize the service rates of inelastic services.
However, \cite{7249208,TWC16,ton} assume that the users have uniform channel conditions, and hence all the users can be served simultaneously by a single multicast transmission.
It remains unclear how to design multicast scheduling \textcolor{black}{for given cache placement} to make full use of the broadcast nature of the wireless medium when users have nonuniform channel conditions.
%There are delay-aware services without strict deadlines (elastic services). For these elastic services,
Moreover, for delay-sensitive services without strict deadlines (i.e., elastic services), it is unknown how to design optimal multicast scheduling \textcolor{black}{for given cache placement} by exploiting the tradeoff between the delay cost and the service cost.

\textcolor{black}{For cache-enabled content-centric wireless networks, there are two important phases, i.e., content placement and content delivery, and the two phases in general happen on different timescales \cite{femto}. In the existing literature on the joint design of caching and multicasting \cite{7249208,TWC16,ton}, the authors either focus on the optimization of one phase for a fixed strategy of the other phase\cite{7249208,TWC16} or consider that content placement and multicast transmission are in the same timescale\cite{ton}. To the best of our knowledge, the optimal design for the two timescale cache placement and multicast scheduling problem is still unknown.
Therefore, as a first and necessary step for the joint two timescale design, in this work, we focus on the optimal multicast scheduling for given cache placement.
Based on the  small timescale problem considered here, we would like to consider the joint two timescale design in future work.}
% \textcolor{black}{In this paper, we focus on the optimal multicast scheduling to satisfy dynamic user demands for given cache placement. Note that, for cache-enabled content-centric wireless networks, the two important phases, i.e., the content placement and delivery happen on different timescales \cite{femto}. Thus, by studying the optimal multicast scheduling policy (in the small timescale), we can gain intuition into the behavior of the optimal content delivery policy, which may guide us to the design of the joint optimal content placement and delivery policy. }
% In this paper, we shall address the above issues.

In this paper, we consider a cache-enabled content-centric wireless network with one BS, $K$ users (with possibly different channel conditions) and $M$ contents (with possibly different content sizes).
The BS stores a certain number of contents in its cache and can fetch any uncached content from the core network through a backhaul link, with a fetching cost depending on the content size. In each slot, the BS schedules one content for multicasting to serve the users' pending requests, with a power cost depending on both the content size and the channel conditions of the users being served.
%  \textcolor{black}{We focus on the optimal multicast scheduling for given cache placement. For cache-enabled content-centric wireless networks,  the content placement and content delivery phases happen on different timescales \cite{femto}.
% In order to achieve the optimal network performance, generally, we need to first determine the optimal content delivery strategy (in the small timescale) to satisfy dynamic user demands for given cache placement, and then design the optimal cache placement (in the large timescale) based on the obtained optimal content delivery strategy for each cache placement.
% As a first and necessary step of the joint design, it is important for us to consider optimal multicast scheduling to satisfy dynamic user demands for given cache status.
% Specifically,}
We consider the optimal dynamic multicast scheduling to jointly minimize the average delay, power, and fetching costs. We formulate the stochastic optimization problem as an infinite horizon average cost Markov decision process (MDP)\cite{bertsekas}. There are several technical challenges.

  $\bullet$ \textbf{Optimality analysis:} The infinite horizon average cost MDP is well-known to be a difficult problem \textcolor{black}{due to the curse of dimensionality \cite{bertsekas}}. While dynamic programming represents a systematic approach for MDPs, there generally exist only numerical solutions,  which do not typically offer many design insights, and are usually not practical due to the curse of dimensionality.
 Therefore, it is desirable to analyze the structures of the optimal policies.
 % However, existing structural analysis is mainly for simple queueing systems, e.g., two-queue systems\cite{cdc,switch} or queueing systems with symmetric arrivals\cite{batch}.  In this work, we consider a multiple-queue system with general request arrivals, channel conditions and content sizes. Therefore, the structural analysis in our problem is more challenging.
Specifically, the considered problem in this work can be treated as the problem of scheduling a broadcast server to parallel queues  with general arrivals and switching costs. Several existing works have studied the related problems\cite{cdc,batch,switch}. In particular, \cite{cdc} and \cite{batch} consider the problems of scheduling a broadcast server to a two-queue system with general arrivals and a multiple-queue system with symmetric arrivals, respectively.
  Reference \cite{switch} studies the problem of scheduling a single server (without broadcast capability) to two queues with switching costs.
  Note that, the switching costs, which relate to the fetching costs in our problem, are not considered in \cite{batch,cdc}, and the  broadcast capability is not considered for the server in \cite{switch}.
  To the best of our knowledge, the structural properties of the optimal scheduling of a broadcast server to parallel queues with general arrivals and switching costs remains unknown and is highly nontrivial.

  $\bullet$ \textbf{Algorithm design:} Standard brute-force algorithms such as value iteration and policy iteration \cite{bertsekas} to MDPs are usually impractical for implementation due to the curse of dimensionality, and cannot exploit the structural properties of the optimal policy. To reduce the complexity, several existing works propose structured optimal algorithms which incorporate the structural properties into the standard algorithms\cite{MPIA,OR}. However, these structured optimal algorithms still suffer from the curse of dimensionality, \textcolor{black}{which is embedded in the optimal control designs for MDPs and generally cannot be broken without any loss of optimality. On the other hand, the structural properties of the optimal policy may be one key reason for its good performance. Therefore, it is highly desirable to  develop low-complexity suboptimal solutions, which can relieve the curse of dimensionality, while maintaining similar structural properties to optimal policies.
  However, for most existing approximate approaches\cite{powell2007approximate,factoredMDP}, there is (in general) no guarantee that the obtained suboptimal policies have similar structural properties to the optimal policies.
  To the best of our knowledge, the design of low complexity suboptimal solutions of similar structural properties to the optimal policies is unknown.}

In this paper, we consider the uniform and nonuniform channel cases.
By using \emph{relative value iteration algorithm} (RVIA)\cite{bertsekas} and the special structures of the request queue dynamics, as well as the power and fetching costs, we analyze the properties of the value function and the state-action cost function of the MDP for both the uniform and nonuniform cases.
Based on these properties, for the uniform case, we show that the optimal policy has a switch structure. In particular, the request queue state space is divided into $M$ regions corresponding to the $M$ contents.
The optimal policy schedules a content for multicasting when the request queue state falls in the region corresponding to the content.
For the uniform case with two contents, we further show that the switch curve is monotonically non-decreasing.
%This indicates that, the threshold to schedule one content (on its request queue length) increases with the increase of the request queue length of the other content.
Next, for the nonuniform case, we show that the optimal policy has a partial switch structure, which is similar to the switch structure in the uniform case. The difference reflects the channel asymmetry among the users.
Then, we propose two low-complexity optimal algorithms by exploiting these structural properties of the optimal policy.
\textcolor{black}{Note that, although the switch structures may look intuitive, it is challenging to prove these structures rigorously.}
\textcolor{black}{Motivated by the switch structures of the optimal policy,} to further reduce the complexity, we also propose a low-complexity suboptimal solution using approximate dynamic programming\cite{bertsekas}.
\textcolor{black}{Different from suboptimal solutions obtained using existing approximate approaches, the proposed suboptimal solution possesses similar structural properties to the optimal policy. Then, we} develop a low-complexity algorithm to compute \textcolor{black}{the suboptimal policy}.
\textcolor{black}{These analytical results hold for both  i.i.d. request arrival and  Markov-modulated request arrival models.}
Numerical results verify the theoretical analysis and demonstrate the performance of the proposed optimal and suboptimal solutions.
 The important notations used in this paper are summarized in Table~\ref{tablenotation}.

\begin{table}[!htbp]
\small
\begin{tabular}{|c|c|}
\hline
$\mathcal{K}$ & set of users\\
\hline
$\mathcal{M}$ & set of all contents\\
\hline
$\mathcal{C}$ & set of contents cached in the BS\\
\hline
$k,m$ & user, content index\\
\hline
$t$&slot index \\
\hline
$p(m,k)$ & \tabincell{c}{minimum transmission power required for\\ delivering content $m$ to user $k$ within a slot} \\
\hline
$f(m)$ & fetching cost of content $m$\\
\hline
$\mathbf{A}=(A_{m,k})$ & request queue matrix\\
\hline
$\mathbf{Q}=(Q_m)$ & request queue state vector for uniform case\\
\hline
$\mathbf{Q}=(Q_{m,k})$ & request queue state matrix for nonuniform case\\
\hline
$d(\mathbf{Q})$ & sum request queue length\\
\hline
$\mu$ & stationary multicast scheduling policy\\
\hline
$u$ & multicast scheduling action\\
\hline
$V(\mathbf{Q})$ & value function\\
\hline
$J(\mathbf{Q},u)$ & state-action cost function\\
\hline
\end{tabular}
\centering
\caption{\small{List of important notations}}\label{tablenotation}
\end{table}

\section{Network Model}\label{sec:model}
As illustrated in Fig.~\ref{fig:systemmodel}, we consider a cache-enabled content-centric wireless network with one BS, $K$ users and $M$ contents.
Let $\mathcal{K}=\{1,2,\cdots,K\}$ denote the set of users.
\textcolor{black}{In our model, each user represents a group of users in the same location.}
 % where each user $k\in\mathcal{K}$ can represent a group of \textcolor{black}{(individual)} users in the same location.
Let $\mathcal{M}=\{1,2,\cdots,M\}$ denote the set of contents, where content $m\in\mathcal{M}$ has the size of $l_m$ (in bits).
Consider time slots of unit length (without loss of generality), and indexed by $t=1,2,\cdots$.\textcolor{black}{\footnote{\textcolor{black}{We consider an abstract model to capture the main features of cache-enabled content-centric networks. The contents could be short videos, soundtracks, E-publications, etc, and the duration of a slot can be several seconds or minutes depending on the specific type of contents considered in this model.}}}
In each slot, each user submits content requests to the BS according to a general distribution.
The BS maintains request queues for different contents, which are implemented using \emph{counters}.
The BS is equipped with a cache storing a certain number of contents, depending on the cache size and the sizes of the cached contents.
We assume the contents stored in the cache are given.
Notice that, caching is in a much larger timescale and in this work, we consider multicast scheduling in a smaller timescale for a given caching design.
% \textcolor{black}{Notice that, the content placement and delivery phases happen on different timescales: in general, content placement is in a much larger timescale (e.g., on a weekly or monthly basis) while content delivery is in a smaller timescale \cite{femto}. In this work, we focus on multicast scheduling (in the small timescale) to satisfy dynamic user demands for given cache placement. This line of approach is also adopted in the literature \cite{7249208,Bastug2015}.}
Let $\mathcal{C}\subset\mathcal{M}$ denote the set of cached contents.
The BS can fetch any uncached content from the core network through a backhaul link, with a fetching cost depending on the content size.
In each slot, the BS schedules one content for multicasting to serve the users' pending requests, with a power cost depending on both the content size and the channel conditions of the users being served. In the following, we elaborate on the physical layer model, the service model, and the request model.
\begin{figure}[!ht]
\begin{centering}
%%\vspace{-0.1cm}
%[scale=.54][height=3.3cm, width=7cm]
\includegraphics[scale=.6]{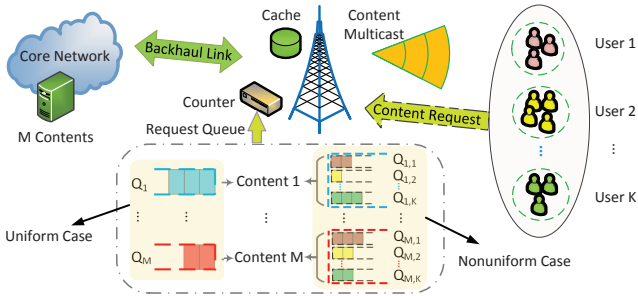}
%%\vspace{-0.05cm}
 \caption{\small{Cache-enabled content-centric wireless network.}}\label{fig:systemmodel}
\end{centering}
%%\vspace{-0.2cm}
\end{figure}
\subsection{Physical Layer Model}
We assume that the duration of the scheduling slot is long enough to average the small-scale channel fading process, and hence the ergodic capacity can be achieved using channel coding.\footnote{Note that, this assumption is also used in \cite{femto} and \cite{Neely}.}
Let $h_k$ denote the average channel gain between user $k$ and the BS.
Assume that only one content is delivered in each slot.
Let $p(m,k)$ denote the minimum transmission power required for delivering content $m$ to user $k$ within a scheduling slot.
Assume $p(m,k)$ satsifies $p(m,k)=y(h_k,l_m)$, where $y(h,l)$ is monotonically non-increasing with $h$ for all $l\geq 0$.
Without loss of generality, we assume that $h_1\geq h_2\geq\cdots\geq h_K$, which implies $p(m,1)\leq p(m,2)\leq\cdots \leq p(m,K)$ for all $m$. In this paper, we consider the uniform and nonuniform channel cases.  In the uniform case, the channel gains of different users are the same, and hence, we have $p(m,1)=p(m,2)=\cdots=p(m,K)\triangleq p(m)$ for each $m$. In the nonuniform case, the channel gains of different users can be different, and hence for each $m$, $p(m,k)$ can be different for different users.

\subsection{Service Model}
We consider multicast service for content delivery.
\textcolor{black}{For clarity, we assume that in each slot, the BS schedules one content for multicasting to serve the users' pending requests. The analytical framework and results can be extended to the general case in which the BS can transmit multiple contents in each slot.}
Let $\mathcal{K}(m,t)\in\mathcal{K}$ denote the set of users who have pending requests for content $m$ at slot $t$.
Let $u(t)\in\mathcal{M}$ denote the content scheduled for multicasting at slot $t$.
If content $u(t)$ is cached (i.e., $u(t)\in\mathcal{C}$), the BS transmits it to all the users in $\mathcal{K}(u(t),t)$ directly; otherwise, the BS first downloads $u(t)$ from the core network through the backhaul link, then multicasts it to the users in $\mathcal{K}(u(t),t)$ and finally discards it after the transmission. \textcolor{black}{Note that, we consider fixed content placement and there is no extra cache storage to hold a new fetched content.}

Next, we illustrate the fetching and power costs. Let $c(m)$ denote the cost for fetching content $m$ via the backhaul link, depending on the content size.
Then, the fetching cost is given by
\begin{equation}
f(m)\triangleq\mathbf{1}(m\not\in\mathcal{C})c(m),\label{eqn:fetchcost}
\end{equation}
where $\mathbf{1}(\cdot)$ denotes the indicator function.
Let $k^*(m,t)\in\mathcal{K}(m,t)$ denote the user who requires the highest transmission power among the users in $\mathcal{K}(m,t)$, i.e., $k^*(m,t)\triangleq\max \mathcal{K}(m,t)$.
Then, to deliver content $m$ to all the users in $\mathcal{K}(m,t)$ within a slot, the power cost $p(m,t)$ is given by
\begin{equation}
p(m,t)\triangleq p\left(m,k^*(m,t)\right)=\max_{k\in\mathcal{K}(m,t)}p(m,k).\label{eqn:powerdefi}
\end{equation}

\subsection{Request Model}
In each slot, each  user submits content requests to the BS.
\textcolor{black}{Notice that each user (representing a group of users in the same location) can submit multiple requests for each content in each slot.}
Let $A_{m,k}(t)\textcolor{black}{\in\mathcal{A}_{m,k}\triangleq\{0,1,\cdots,A_{m,k}^{max}\}}$ denote the number of the new request arrivals for content $m$ from user $k$ at the end of slot $t$, where $m\in\mathcal{M}$ and $k\in\mathcal{K}$.
Let $\mathbf{A}(t)=(A_{m,k}(t))_{m\in\mathcal{M}, k\in\mathcal{K}}\textcolor{black}{\in\bm{\mathcal{A}}\triangleq\prod_{m,k}\mathcal{A}_{m,k}}$ denote the request arrival matrix at slot $t$.
We assume that $A_{m,k}(t)$ is i.i.d. over slots and independent w.r.t. $m$ according to a general distribution.
\textcolor{black}{For ease of illustration, we assume that the request arrival process is i.i.d. according to the Independent Reference Model (IRM), which is a standard approach adopted in the literature \cite{femto,ton}. The IRM is reasonable as each user in our model represents a group of users in the same location \cite{ITC}.
In Section~\ref{sec:markov}, we shall extend the analysis for the i.i.d. request arrival model to a Markov-modulated request arrival model.}
The BS maintains request queues for different contents. The request queues are implemented using \emph{counters} and no data is stored in these request queues.
In the following, we introduce two request queue models for the uniform and nonuniform cases, respectively.
\subsubsection{Uniform Case}  In the uniform case, once content $m$ is multicasted using transmission power $p(m)$, all the users can receive content $m$. Therefore, we do not differentiate the requests for each content at the user level.
Specifically, the BS maintains a separate request queue for each content $m\in\mathcal{M}$.
Let $Q_m(t)\in\mathcal{Q}_m\triangleq\{0,1,\cdots,N_m\}$ denote the request queue length for content $m$ at the beginning of slot $t$, where $N_m$ is assumed to be finite (can be sufficiently large) for technical tractability.
As illustrated in Section II-B, if content $m$ is scheduled for transmission at slot $t$ (i.e., $u(t)=m$), all the pending requests for content $m$ are satisfied, i.e., the request queue for content $m$ is emptied.
Thus, the request queue dynamics for content $m$ is as follows:
\begin{equation}
  Q_{m}(t+1)=\min\{\mathbf{1}(u(t)\neq m)Q_m(t)+A_m(t),N_m\},\label{eqn:queue-ho}
\end{equation}
where $A_m(t)\triangleq\sum_kA_{m,k}(t)$ denotes the total number of the request arrivals for content $m$ at the end of slot $t$.
Let $\mathbf{Q}(t)\triangleq(Q_m(t))_{m\in\mathcal{M}}\in\bm{\mathcal{Q}}$ denote the request queue state vector at the beginning of slot $t$ in the uniform case, where $\bm{\mathcal{Q}}\triangleq\prod_{m\in\mathcal{M}}\mathcal{Q}_m$ denotes the request queue state space in the uniform case.
% \textcolor{black}{In Fig.~\ref{fig:tradeoff}, we provide an example to illustrate the service dynamics and the tradeoff between the average delay and the average service cost (including the power and fetching costs) for a certain content in the uniform case. This content will be scheduled for multicast if its queue length is larger than $Q_{th}$. We can see that, when the BS waits longer to collect the requests (i.e., $Q_{th}$ increases), the average service cost decreases while the average delay cost increases.  This reveals the tradeoff between the delay cost and service cost.}
\subsubsection{Nonuniform Case}  In the nonuniform case, different transmission powers are required to deliver a content to different users, as illustrated in Section II-B.
Therefore, we  differentiate the requests for each content at the user level.
Specifically, the BS maintains a separate request queue for each content-user pair $(m,k)\in\mathcal{M}\times\mathcal{K}$.
Let $Q_{m,k}(t)\in\mathcal{Q}_{m,k}\triangleq\{0,1,\cdots,N_{m,k}\}$ denote the request queue length for content-user pair $(m,k)$ at the beginning of slot $t$,  where $N_{m,k}$ is assumed to be finite (can be sufficiently large) for technical tractability.
Therefore, $\mathcal{K}(m,t)$ can  be expressed in terms of the request queue state, i.e., $\mathcal{K}(m,t)=\{k|Q_{m,k}(t)>0\}$. The request queue dynamics for content-user pair $(m,k)$ is as follows:
\begin{equation}
  Q_{m,k}(t+1)=\min\{\mathbf{1}(u(t)\neq m)Q_{m,k}(t)+A_{m,k}(t),N_{m,k}\}.\label{eqn:queue-he}
\end{equation}
Let $\mathbf{Q}_m(t)\triangleq(Q_{m,k}(t))_{k\in\mathcal{K}}\in\bm{\mathcal{Q}}_m$ denote the request queue state vector for content $m$ at the beginning of slot $t$ in the nonuniform case, where $\bm{\mathcal{Q}_m}\triangleq\prod_{k\in\mathcal{K}}\mathcal{Q}_{m,k}$ denotes the request queue state space for content $m$ in the nonuniform case.
Let $\mathbf{Q}(t)\triangleq(\mathbf{Q}_m(t))_{m\in\mathcal{M}}\in\bm{\mathcal{Q}}$ denote the request queue state matrix at the beginning of slot $t$ in the nonuniform case, where $\bm{\mathcal{Q}}\triangleq\prod_{m\in\mathcal{M}}\bm{\mathcal{Q}}_m$ denotes the request queue state space in the nonuniform case.

\textcolor{black}{Note that, in \eqref{eqn:queue-ho} and \eqref{eqn:queue-he}, once a content is scheduled, the corresponding request queue (in the uniform case) or request queues (in the nonuniform case) are emptied. This special form of queue departure  reflects the multicast gain. Our framework  holds for any number of users and any profile of request arrivals.}

\section{Problem Formulation and Optimality Equation}\label{sec:formulation}
%\textcolor{red}{In this section, we first formulate the stochastic optimization problem. Then, we obtain the optimality equation.}

\subsection{Problem Formulation}
 Given an observed request queue state, the multicast scheduling action $u$ is determined according to a stationary policy defined below.
\begin{definition}[\text{Stationary Multicast Scheduling Policy}]
  A stationary multicast scheduling policy $\mu$ is a mapping from the request queue state $\mathbf{Q}\in\bm{\mathcal{Q}}$ to the multicast scheduling action $u\in\mathcal{M}$, where $\mu(\mathbf{Q})=u$.
\label{definition:definition1}
\end{definition}

 By the queue dynamics in \eqref{eqn:queue-ho} or \eqref{eqn:queue-he}, the induced random process $\{\mathbf{Q}(t)\}$ under policy $\mu$ is a controlled Markov chain.
We restrict our attention to stationary unichain policies\footnote{A unichain policy is a policy, under which the induced Markov chain has a single recurrent class (and possibly some transient states)\cite{bertsekas}}. For a given stationary unichain policy $\mu$, the average delay cost is defined as
%For a given stationary unichain policy $\mu$, by Little's law, the average delay cost is defined as
\begin{equation}
  \bar{d}(\mu)\triangleq\limsup_{T\to\infty}\frac{1}{T}\sum_{t=1}^T \mathbb{E}\left[d\left(\mathbf{Q}(t)\right)\right],\label{eqn:delay}
\end{equation}
where the expectation is taken w.r.t. the measure induced by \textcolor{black}{the random request arrivals and} the policy $\mu$, $d\left(\mathbf{Q}(t)\right)\triangleq\sum_mQ_m(t)$ in the uniform case and $d\left(\mathbf{Q}(t)\right)\triangleq\sum_{m,k}Q_{m,k}(t)$ in the nonuniform case.
\textcolor{black}{By Little's law, $\bar{d}(\mu)$ reflects the average waiting time in the network under policy $\mu$.}
By \eqref{eqn:fetchcost} and \eqref{eqn:powerdefi}, the average fetching and power costs are given by
 \begin{align}
  \bar{f}(\mu)&\triangleq\limsup_{T\to\infty}\frac{1}{T}\sum_{t=1}^T \mathbb{E}\left[f\left(u(t)\right)\right],\label{eqn:fetching}\\
  \bar{p}(\mu)&\triangleq\limsup_{T\to\infty}\frac{1}{T}\sum_{t=1}^T \mathbb{E}\left[p(\mathbf{Q}(t),u(t))\right].\label{eqn:power}
\end{align}
Here, with abuse of notation, we also use $p(\mathbf{Q}(t),u(t))$ to represent $p(u(t),t)$ given in \eqref{eqn:powerdefi}, as  $\mathcal{K}(u(t),t)=\{k|Q_{u(t),k}(t)>0\}$.
 \textcolor{black}{Please note that,  there is an inherent tradeoff between the delay cost and the service cost (including the power and fetching costs) in our model.
As a simple example, we consider a multicast scheduling policy for the uniform case, where content $m$ is scheduled only if $Q_m\geq Q_{th}$.
As illustrated in Fig.~\ref{fig:tradeoff}, we can see that for content $m$, when $Q_{th}$ increases, the average service cost decreases while the average delay cost increases.
This is because that when $Q_{th}$ increases, for each of the $Q_{th}$ requests, its waiting time increases while its service cost decreases.}
\begin{figure}[!h]
\begin{centering}
%%\vspace{-0.1cm}
%[scale=.54][height=3.3cm, width=7cm]
\includegraphics[scale=.45]{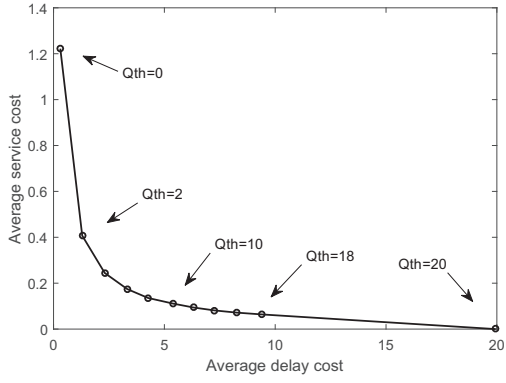}
%%\vspace{-0.05cm}
 \caption{\textcolor{black}{\small{Tradeoff between the average delay cost and the average service cost for a certain content.}}}\label{fig:tradeoff}
\end{centering}
%%\vspace{-0.2cm}
\end{figure}

 \textcolor{black}{Therefore, to capture this tradeoff,  we define the average system cost (weighted sum cost) under a given stationary unichain policy $\mu$ as}
 \begin{align}
  \bar{g}(\mu)&\triangleq\bar{d}(\mu)+w_f\bar{f}(\mu)+w_p\bar{p}(\mu)\nonumber\\
  &=\limsup_{T\to\infty}\frac{1}{T}\sum_{t=1}^T \mathbb{E}\left[g(\mathbf{Q}(t),u(t))\right],\label{eqn:cost}
\end{align}
where $w_f$ and $w_p$  are the associated weights for the fetching and power costs, respectively, \textcolor{black}{which reflect the tradeoff}, and $g(\mathbf{Q},u)\triangleq d\left(\mathbf{Q}\right)+w_ff(u)+w_pp(\mathbf{Q},u)$ is the per-stage cost.

We wish to find an optimal multicast scheduling policy to minimize the average system cost $\bar{g}(\mu)$ in \eqref{eqn:cost}.
\begin{problem}[System Cost Minimization Problem]
  \begin{equation}
  \bar{g}^*\triangleq\min_{\mu}\limsup_{T\to\infty}\frac{1}{T}\sum_{t=1}^T \mathbb{E}\left[g(\mathbf{Q}(t),u(t))\right],\label{problem:originalproblem}
\end{equation}
\end{problem}
where $\mu$ is a stationary unchain multicast scheduling policy and $\bar{g}^*$ denotes the minimum average system cost achieved by the optimal policy $\mu^*$.

 Problem~\ref{problem:originalproblem} is an infinite horizon average cost MDP, which is well-known to be a difficult problem \textcolor{black}{due to the curse of dimensionality}.
 \textcolor{black}{According to \cite[Theorem 8.4.5]{puterman}, for unichain infinite horizon average cost MDPs with finite state and action spaces, there always exists a deterministic stationary policy that is optimal. Note that, these requirements are satisfied by the MDP considered in our work. Therefore, it is sufficient to focus on the deterministic stationary policy space.}

\subsection{Optimality Equation}
The optimal multicast scheduling policy $\mu^*$ can be obtained by solving the following Bellman equation.%\footnote{All the proofs are omitted due to page limit and interested readers can refer to \cite{dropbox} for details.}
\begin{lemma}[Bellman Equation]\label{lemma:bellman}
There exist a scalar $\theta$ and a value function $V(\cdot)$ satisfying
  % Suppose a scalar $\theta$ and a real-valued function $V(\cdot)$ satisfy:%\footnote{Proposition 4.2.1, Proposition 4.2.3 and Proposition 4.2.5 in\cite{bertsekas} provide the conditions for the existence of $\theta$ and $V(\cdot)$.}
\begin{equation}
  \theta+V(\mathbf{Q})=\min_{u\in\mathcal{M}}\left\{g(\mathbf{Q},u)+\mathbb{E}\left[V(\mathbf{Q}')\right]\right\},~\forall \mathbf{Q}\in\bm{\mathcal{Q}},\label{eqn:bellman}
\end{equation}
where \textcolor{black}{the expectation is taken over the distribution of the request arrival $\mathbf{A}$}, and $\mathbf{Q}'=(Q'_m)_{m\in\mathcal{M}}$ with $Q'_m=\min\{\mathbf{1}(u\neq m)Q_m+A_m,N_m\}$ in the uniform case; $\mathbf{Q}'=(Q'_{m,k})_{m\in\mathcal{M}, k\in\mathcal{K}}$ with $ Q'_{m,k}=\min\{\mathbf{1}(u\neq m)Q_{m,k}+A_{m,k},N_{m,k}\}$ in the nonuniform case.
Then, $\theta=\bar{g}^*$ is the optimal value to Problem 1 for all initial state $\mathbf{Q}(1)\in\bm{\mathcal{Q}}$,
and the optimal policy $\mu^*$ achieving $\bar{g}^*$ is given by
 % and $V(\cdot)$ is called the value function.
%Furthermore, if $\mu^*(\mathbf{Q})$ attains the minimum in the R.H.S. of \eqref{eqn:bellman} for each $\mathbf{Q}$, then the stationary unichain policy $\mu^*$ is the optimal policy achieving the optimal value $\bar{g}^*$.
% Furthermore, if% then the stationary unichain policy $\mu^*$ is the optimal policy achieving the optimal value $\bar{g}^*$.
\begin{align}\label{eqn:mu}
  \mu^*(\mathbf{Q})=\arg\min_{u\in\mathcal{M}}\left\{g(\mathbf{Q},u)+\mathbb{E}\left[V(\mathbf{Q}')\right]\right\},~\forall \mathbf{Q}\in\bm{\mathcal{Q}}.
\end{align}
\end{lemma}
\begin{proof}
  Please see Appendix A.
\end{proof}

From the Bellman equation in \eqref{eqn:bellman}, we can see that $\mu^*$ depends on the state $\mathbf{Q}$ through the value function $V(\cdot)$. Obtaining $V(\cdot)$ involves solving the Bellman equation for all $\mathbf{Q}$, for which there is no closed-form solution in general\cite{bertsekas}.  Brute-force numerical solutions such as value iteration and policy iteration do not typically offer many design insights, and are usually impractical for implementation in practical systems due to the curse of dimensionality\cite{bertsekas}.  Therefore, it is desirable to study the structure of $\mu^*$.

To analyze the structure of $\mu^*$, we also introduce the state-action cost function:
\begin{equation}
  J(\mathbf{Q},u)\triangleq g(\mathbf{Q},u)+\mathbb{E}\left[V(\mathbf{Q}')\right].\label{eqn:state_action_func}
\end{equation}
Note that $J(\mathbf{Q},u)$ is related to the R.H.S. of the Bellman equation in \eqref{eqn:bellman}.
In particular, based on Lemma~\ref{lemma:bellman}, the optimal policy $\mu^*$ can be expressed in terms of $J(\mathbf{Q},u)$, i.e.,
\begin{equation}
  \mu^*(\mathbf{Q})=\arg\min_{u\in\mathcal{M}} J(\mathbf{Q},u),~\forall \mathbf{Q}\in\bm{\mathcal{Q}}.
\end{equation}
In Sections IV and V, we shall analyze the structures of the optimal policies for the uniform and nonuniform cases, respectively, based on the properties of the value function $V(\mathbf{Q})$ and the state-action cost function $J(\mathbf{Q},u)$.
\section{Optimality Properties in Uniform Case}\label{sec:uniform}
In this section, we consider the uniform case.  We first show that the optimal policy has a switch structure. Then, we show that the switch curve is monotonically non-decreasing for the uniform case with two contents.

\subsection{Structure of Optimal Policy}
Problem 1 can be treated as the problem of scheduling a broadcast server to parallel queues with general random arrivals, channel conditions, and content sizes.
Therefore, the structural analysis is more challenging than the existing structural analysis for simple queueing systems (see Section I for the detailed discussion).
First, by RVIA and the special structures of the request queue dynamics, as well as the power and fetching costs, we have the following property of $V(\mathbf{Q})$.
\begin{lemma}[Monotonicity of Value Function]
%$V(\mathbf{Q}^1)\leq V(\mathbf{Q}^2)$,
In the uniform case, for any $\mathbf{Q}^1$, $\mathbf{Q}^2\in\bm{\mathcal{Q}}$ such that $\mathbf{Q}^2\succeq\mathbf{Q}^1$, we have $V(\mathbf{Q}^2)\geq V(\mathbf{Q}^1)$.\footnote{The notation $\succeq$ indicates component-wise $\geq$.}
%The value function $V(Q)$ is monotonically non-decreasing in $\mathbf{Q}$, i.e.
\label{lemma:propertyV1}
\end{lemma}
\begin{proof}
  Please see Appendix B.
\end{proof}
%Next, define the state-action cost function as follows.
%\begin{equation}
%  J(\mathbf{Q},u)=d(\mathbf{Q})+f(u)+\mathbb{E}_{\mathbf{A}}\left[\tilde{V}(\mathbf{Q},\mathbf{A},u)\right],\label{eqn:state_action_func}
%\end{equation}
%where $\mathbf{Q}=(Q_m)_{m\in\mathcal{M}}$, $d(\mathbf{Q})=\sum_mQ_m$ and $\tilde{V}(\mathbf{Q},\mathbf{A},u)=V(\mathbf{Q}')$  with $\mathbf{Q}'=(Q'_m)_{m\in\mathcal{M}}$ and $Q'_m=\{\mathbf{1}(u\neq m)Q_m+A_m,N_m\}$.
%Note that $J(\mathbf{Q},u)$ is related to the R.H.S. of the Bellman equation in \eqref{eqn:bellman1}.

Then, based on Lemma~\ref{lemma:propertyV1} and  the special properties of multicasting, we have the following property of $J(\mathbf{Q},u)$.
\begin{lemma}[Monotonicity of State-Action Cost Function]
In the uniform case, for any $u,v\in\mathcal{M}$ and $v\neq u$,  $J(\mathbf{Q},u)-J(\mathbf{Q},v)$ is monotonically non-increasing with $Q_u$, i.e.,
\begin{equation}
  J(\mathbf{Q}+\mathbf{e}_u,u)-J(\mathbf{Q}+\mathbf{e}_u,v)\leq J(\mathbf{Q},u)-J(\mathbf{Q},v),\label{eqn:propertyJ1}
\end{equation}
where $\mathbf{e}_u$ denotes the $1\times M$ vector with all entries 0 except for a 1 in its $u$-th entry.
\label{lemma:propertyJ1}
\end{lemma}
\begin{proof}
  Please see Appendix C.
\end{proof}

%Note that, the property of $J(\mathbf{Q},u)$ in Lemma~\ref{lemma:propertyJ1} is different from the \emph{supermodularity} used in the existing structural analysis.
Note that, the property of $J(\mathbf{Q},u)$ in Lemma~\ref{lemma:propertyJ1} is similar to the diminishing-return property of submodular functions used in the existing structural analysis\cite{Koole}.
Lemma~\ref{lemma:propertyJ1} comes from the special structure introduced by multicasting and is key to analyze the optimality properties.
Lemma~\ref{lemma:propertyJ1} indicates that, if it is better to multicast content $u$ than content $v$ for some state $\mathbf{Q}$ (i.e., $J(\mathbf{Q},u)\leq J(\mathbf{Q},v)$), then it is also better to multicast content $u$ than $v$ for state $\mathbf{Q}+\mathbf{e}_u$ (i.e., $J(\mathbf{Q}+\mathbf{e}_u,u)\leq J(\mathbf{Q}+\mathbf{e}_u,v)$). This leads to the following switch structure of the optimal policy $\mu^*$.
\begin{theorem}[Switch Structure of Optimal Policy]
The optimal policy $\mu^*$ in the uniform case has a switch structure, i.e., for all $u\in\mathcal{M}$, we have
\begin{equation}\label{eqn:switch}
  \mu^*(\mathbf{Q})=u, \text{if}~Q_u\geq s_u(\mathbf{Q}_{-u}),
\end{equation}
where the switch curve for content $u$ is given by
%\begin{equation}\label{eqn:curve}
% s(\mathbf{Q}_{-u})\triangleq\left\{
%                      \begin{array}{ll}
%                        \min\mathcal{S}(\mathbf{Q}_{-u}), & \hbox{if $\mathcal{S}(\mathbf{Q}_{-u})\neq\emptyset$} \\
%                        \infty, & \hbox{otherwise}
%                      \end{array}
%                    \right.,
%\end{equation}
\begin{equation*}%\label{eqn:curve}
s_u(\mathbf{Q}_{-u})\triangleq\begin{cases}\min\mathcal{S}_u(\mathbf{Q}_{-u}),  & \text{if}~\mathcal{S}_u(\mathbf{Q}_{-u})\neq\emptyset \\
            \infty,  &\text{otherwise}
  \end{cases}
  \end{equation*}
with $\mathcal{S}_u(\mathbf{Q}_{-u})\triangleq\{Q_u| J(\mathbf{Q},u)\leq J(\mathbf{Q},v)~\forall v\in\mathcal{M}, v\neq u\}$.
 Here, $\mathbf{Q}_{-u}\triangleq(Q_m)_{m\in\mathcal{M},m\neq u}$ denotes the request queue state vector corresponding to all other contents except content $u$.
\label{theorem:theorem1}
\end{theorem}
\begin{proof}
  Please see Appendix D.
\end{proof}
\begin{remark}
Theorem~\ref{theorem:theorem1} indicates that, the request queue state space is divided into $M$ regions corresponding to the $M$ contents,
and the optimal policy schedules a content for multicasting when the request queue state falls in the region corresponding to the content, as illustrated in Fig.~\ref{fig:symStruc3q}.
In addition, given $\mathbf{Q}_{-u}$, the scheduling for content $u$ is of the threshold type, as illustrated in Fig.~\ref{fig:symStruc2q}.
\textcolor{black}{Specifically, if $Q_u\geq s_u(\mathbf{Q}_{-u})$, the BS schedules content $u$ for multicasting and the request queue for content $u$ is emptied; if $Q_u< s_u(\mathbf{Q}_{-u})$, the BS keeps on waiting to gather more requests for content $u$ and the request queue for content $u$ keeps on increasing.}
This indicates that,  when $Q_u$ is small (i.e., the delay cost is small), it is not efficient to schedule content $u$, as a higher power cost (and a higher fetching cost if $u\not\in\mathcal{C}$) is consumed per request for content $u$; \textcolor{black}{when $Q_u$ is large (i.e., the delay cost is large), it is more efficient to schedule content $u$, as the requests for content $u$  is more urgent.} This reveals the tradeoff between the delay cost and the power cost (and the fetching cost if $u\not\in\mathcal{C}$) for content $u$.
\label{remark:remarkofTheorem1}
\end{remark}
\begin{answer}
\begin{remark}
From Theorem~\ref{theorem:theorem1}, we can see that cache placement does not affect the structural properties of the optimal policy. That is, the switch structure holds for any cache placement strategies. However, cache placement does affect the values of the switch curves of the optimal policy. The reason is that  cache placement affects the tradeoff among the delay, power and fetching costs through affecting the fetching costs, and the switch curves of the optimal policy are adaptive to this tradeoff. The impacts of the fetching costs on the switch curves can be observed from Fig.~\ref{fig:uniform_fetch}.
\label{remark:switchcurve}
\end{remark}
\end{answer}

\textcolor{black}{Note that, although the exact values of the switch curves rely on the exact value of $V(\mathbf{Q})$, the switch structural property only relies on the monotonicity properties of $V(\mathbf{Q})$ and $J(\mathbf{Q},u)$.
These structural properties can be used to reduce the computational complexity in obtaining the optimal policy, without knowing the exact value of the switch curves.
Specifically,}
from Theorem~~\ref{theorem:theorem1}, we know that, for all $\mathbf{Q}\in\bm{\mathcal{Q}}$,
\begin{equation}\label{eqn:indicate_uni}
  \mu^*(\mathbf{Q})=u~\Rightarrow~\mu^*(\mathbf{Q}+\mathbf{e}_u)=u.
\end{equation}
Therefore, computing the optimal policy $\mu^*$ requires conducting the minimization in the R.H.S. of \eqref{eqn:mu} for some $\mathbf{Q}$ only (instead of all $\mathbf{Q}\in\bm{\mathcal{Q}}$), which significantly reduces the computational complexity. Later, in Section VI, we shall design low complexity optimal algorithms based on \textcolor{black}{\eqref{eqn:indicate_uni}}.

%\begin{figure}[h]
%\begin{centering}
%%\vspace{-0.1cm}
%\includegraphics[scale=.25]{3q.eps}
%%\vspace{-0.1cm}
% \caption{\small{Switch structure for three contents.}}\label{fig:3q}
%\end{centering}
%%\vspace{-0.1cm}
%\end{figure}
%\begin{figure}[h]
%\begin{centering}
%%\vspace{-0.1cm}
%\includegraphics[scale=.25]{2q.eps}
%%\vspace{-0.1cm}
% \caption{\small{Switch structure for two contents.}}\label{fig:2q}
%\end{centering}
%%\vspace{-0.1cm}
%\end{figure}
\begin{figure}[t]
\begin{minipage}[t]{.5\linewidth}
\centering
%%\vspace{-0.1cm}
        % \includegraphics[height=4cm, width=4.4cm]{symStruc3q.eps}
        \includegraphics[scale=.33]{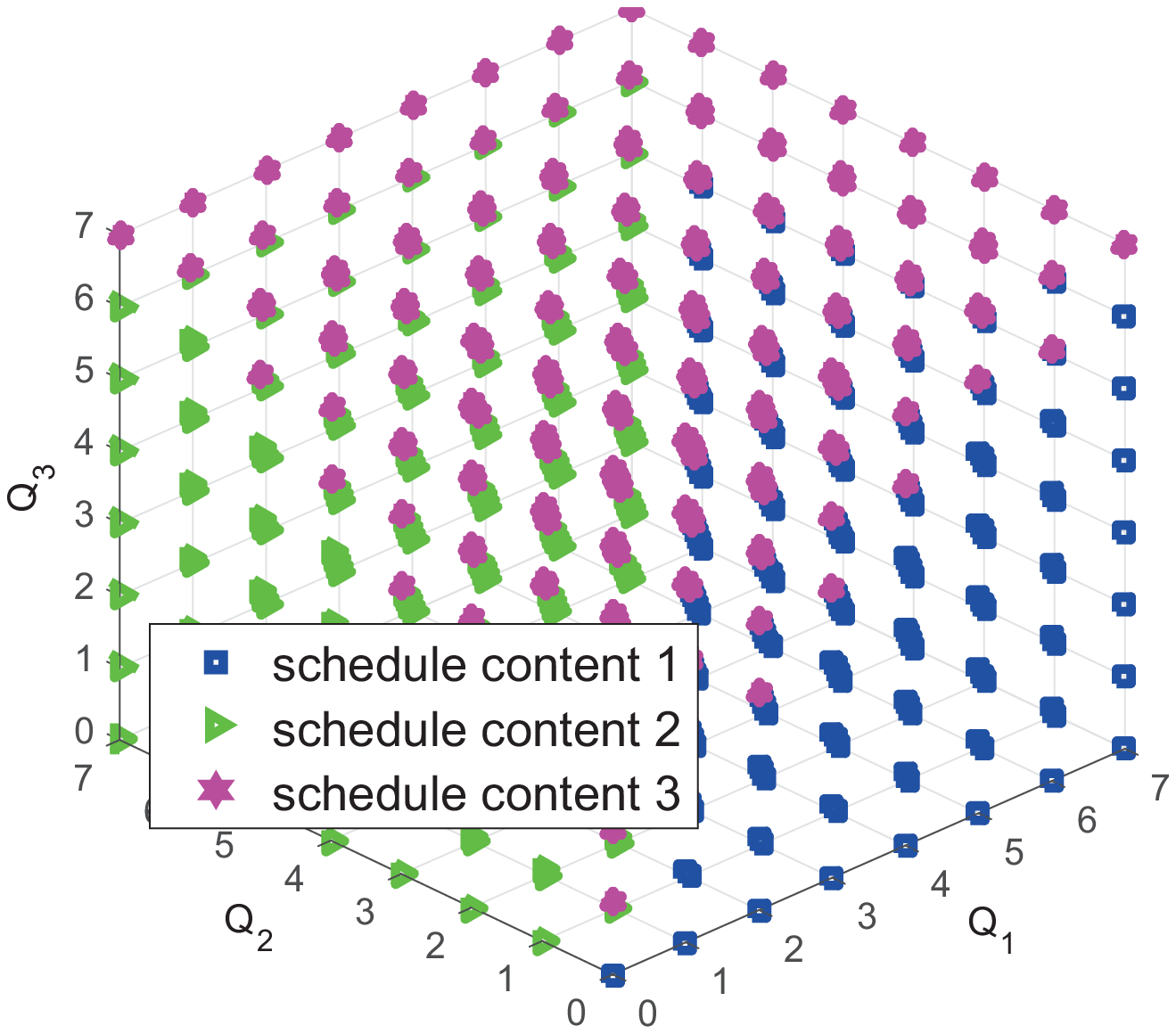}
        %\includegraphics[scale=.33]{3q.eps}
       %\vspace{-0.4cm}
\subcaption{\small{Three-content case.}}\label{fig:symStruc3q}
%%\vspace{-0.05cm}
\end{minipage}%
\begin{minipage}[t]{.5\linewidth}
\centering
%%\vspace{-0.1cm}
        % \includegraphics[height=4cm, width=4.4cm]{symStruc2q.eps}
        \includegraphics[scale=.33]{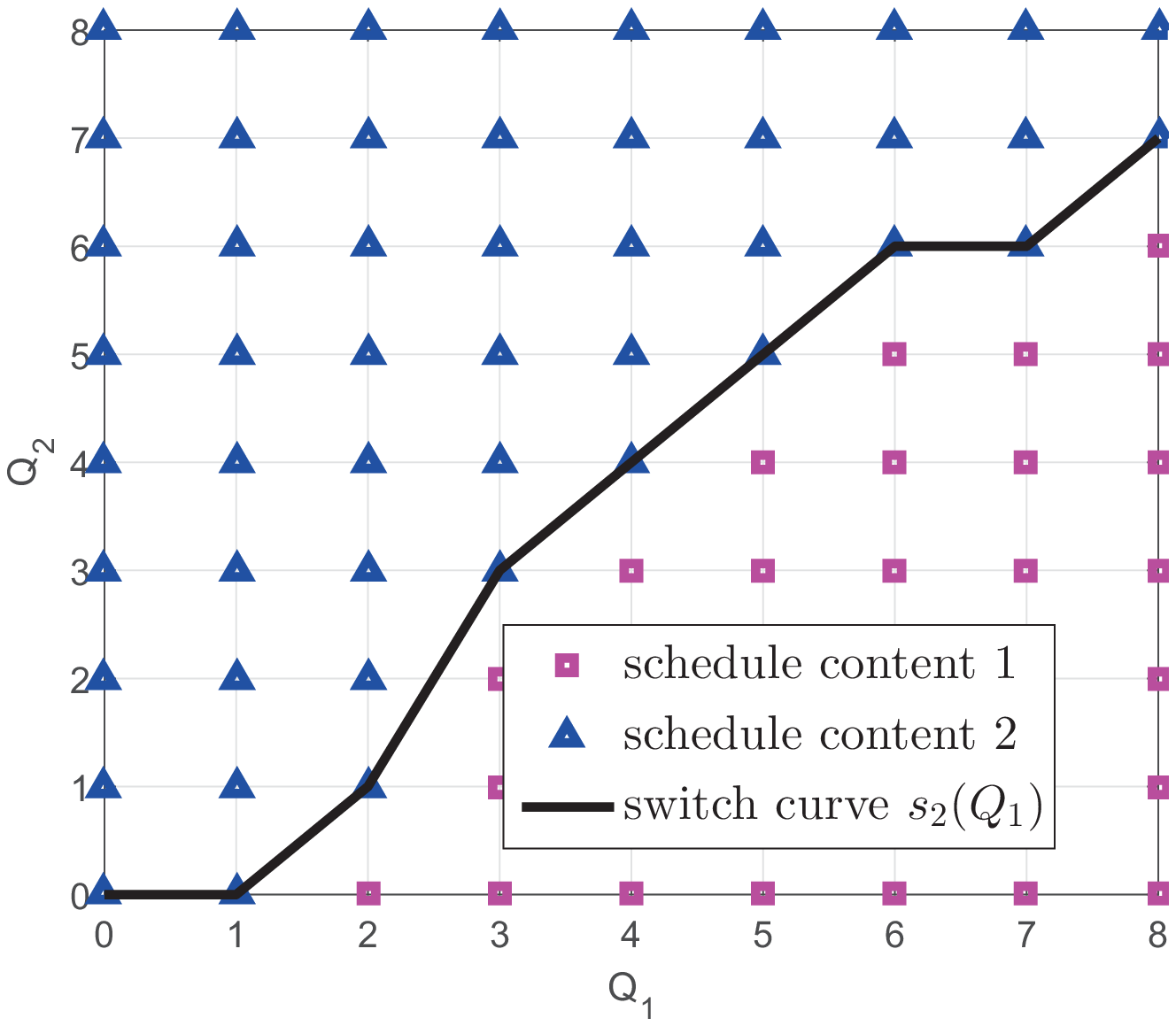}
        %\includegraphics[scale=0.33]{2q.eps}
       %\vspace{-0.4cm}
\subcaption{\small{Two-content case.}}\label{fig:symStruc2q}
%%\vspace{-0.05cm}
\end{minipage}
\caption{\small{Switch structure of optimal scheduling in the uniform case.}}\label{fig:uniform}
%\vspace{-0.2cm}
\end{figure}
\subsection{Special Case: Two Contents}
%Consider the special uniform case with two contents, i.e., $M=2$. By Theorem~\ref{theorem:theorem1}, we have the following corollary.
%\begin{corollary}
%  In the uniform case with two contents, the optimal policy is:
%\begin{equation}\label{eqn:specialcase}
%  \mu^*(\mathbf{Q})=\left\{
%                      \begin{array}{ll}
%                        2, & \hbox{if $Q_2\geq s(Q_1)$} \\
%                        1, & \hbox{otherwise}
%                      \end{array}
%                    \right.,
%\end{equation}
%where  $\mathbf{Q}=(Q_1,Q_2)$ and the switch cure is given by $s(Q_1)\triangleq\min\{Q_2|J(\mathbf{Q},2)\leq J(\mathbf{Q},1)\}$.
%\label{corollary:special}
%\end{corollary}
\begin{figure}[t]
\begin{minipage}[t]{.5\linewidth}
\centering
%%\vspace{-0.1cm}
        % \includegraphics[height=4cm, width=4.4cm]{symStruc3q.eps}
        \includegraphics[scale=.33]{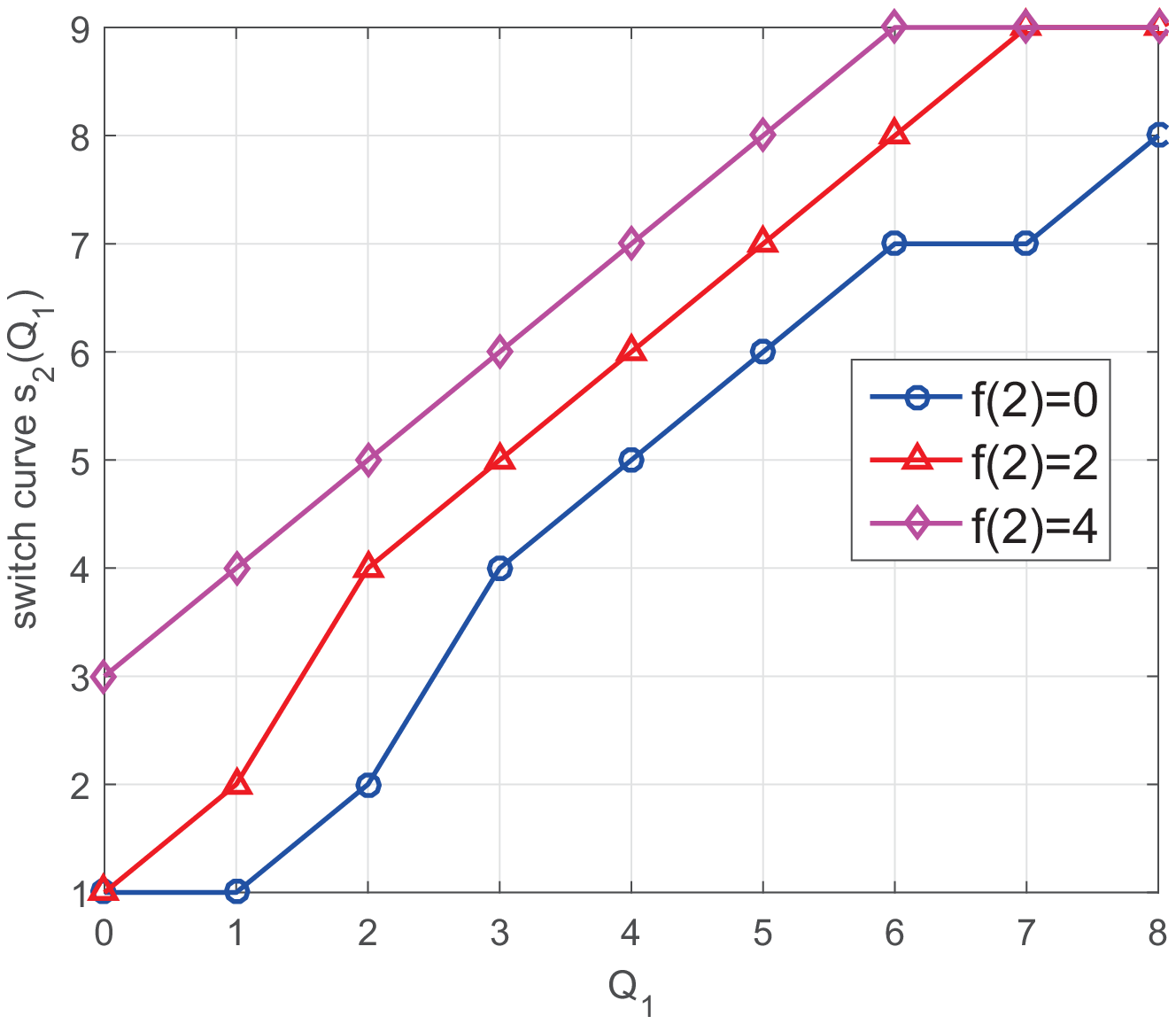}
        %\includegraphics[scale=.33]{3q.eps}
       %\vspace{-0.4cm}
\subcaption{\small{Uniform case: two-content.}}\label{fig:uniform_fetch}
%%\vspace{-0.05cm}
\end{minipage}%
\begin{minipage}[t]{.5\linewidth}
\centering
%%\vspace{-0.1cm}
        % \includegraphics[height=4cm, width=4.4cm]{symStruc2q.eps}
        \includegraphics[scale=.33]{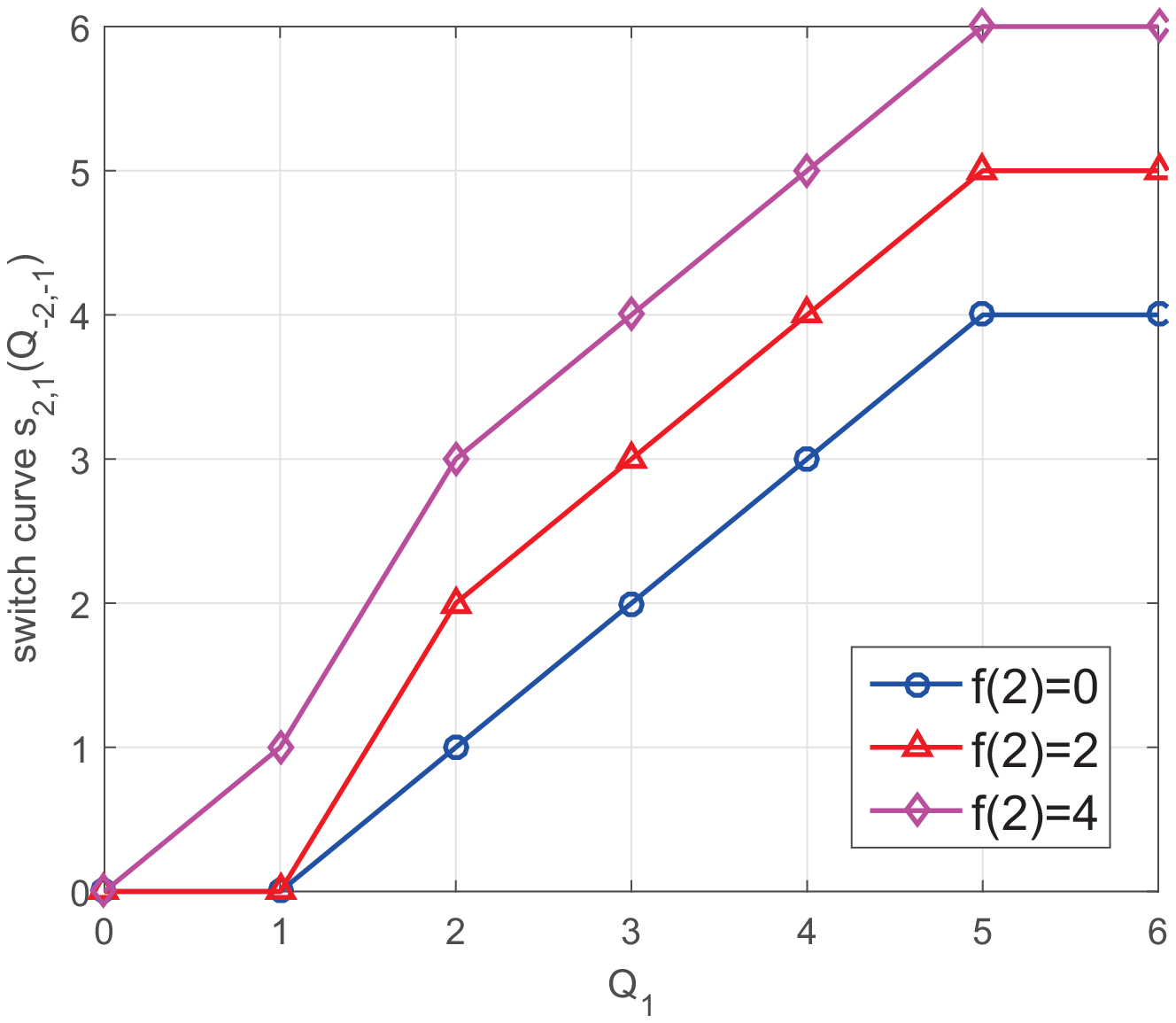}
        %\includegraphics[scale=0.33]{2q.eps}
       %\vspace{-0.4cm}
\subcaption{\small{Nonuniform case: two-content two-user.}}\label{fig:nonuniform_fetch}
%%\vspace{-0.05cm}
\end{minipage}
\caption{\textcolor{black}{\small{Impacts of the fetching costs on switch curves in the uniform and nonuniform cases.}}}\label{fig:fetch}
%\vspace{-0.2cm}
\end{figure}
Now, consider the special uniform case with two contents, i.e., $M=2$. By Theorem~\ref{theorem:theorem1}, we can see that, for $M=2$, either one of the two switch curves, i.e., $s_1(Q_2)$ and $s_2(Q_1)$, is sufficient to characterize the optimal policy.
Moreover, by Lemma \ref{lemma:propertyV1} and Lemma \ref{lemma:propertyJ1}, $s_1(Q_2)$ and $s_2(Q_1)$ have the following property.
 \begin{lemma}[Monotonicity of Switch Curve] For the uniform case with two contents, the switch curves $s_1(Q_2)$ and $s_2(Q_1)$ of the optimal policy are monotonically non-decreasing in $Q_2$ and $Q_1$, respectively.
\label{lemma:propertyofswitch}
\end{lemma}
\begin{proof}
  Please see Appendix E.
\end{proof}
\begin{figure}[!ht]
\begin{centering}
%%\vspace{-0.1cm}
%[scale=.54][height=3.3cm, width=7cm]
\includegraphics[scale=.45]{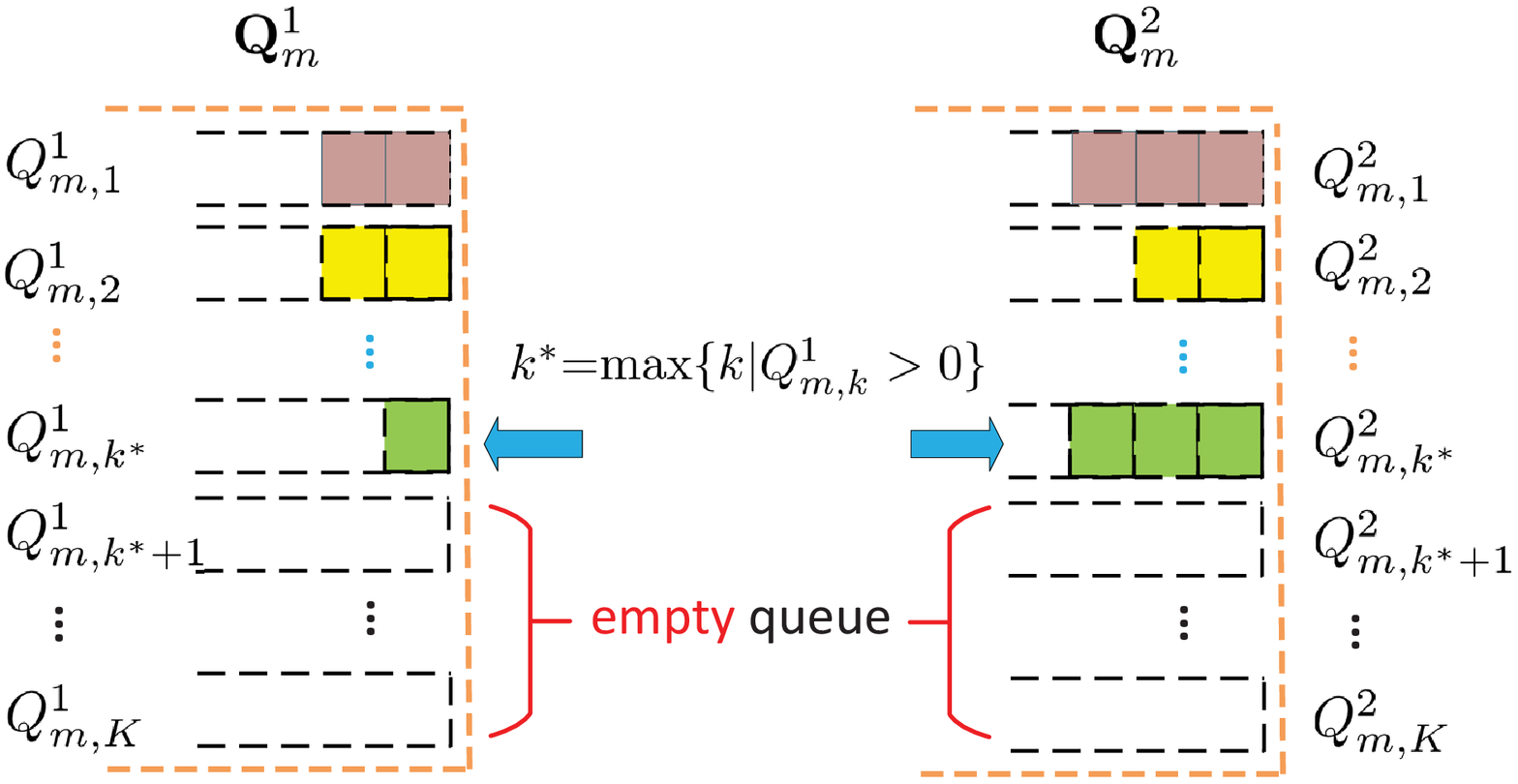}
 \caption{\small{Illustration of $\mathbf{Q}_m^2\trianglerighteq\mathbf{Q}_m^1$.}}\label{fig:notation}
\end{centering}
\end{figure}
\begin{figure*}[t]
%\vspace{-0.05cm}
\begin{minipage}[t]{0.33\linewidth}
\centering
%%\vspace{-0.1cm}
        % \includegraphics[height=4.7cm, width=5.1cm]{asymStruc3D.eps}
        \includegraphics[scale=.42]{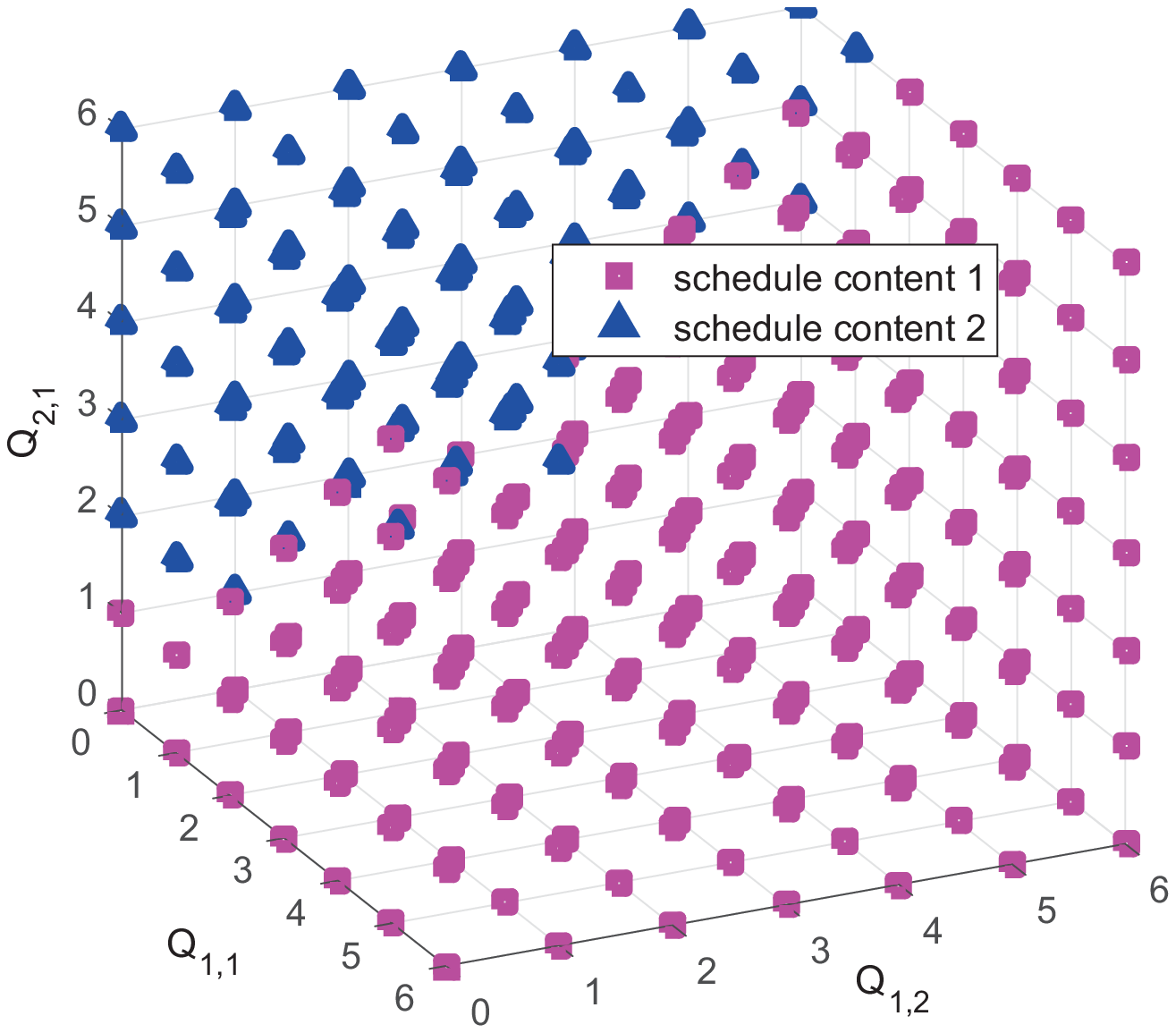}
%\includegraphics[scale=.3]{asymStruc3D.eps}
        %\vspace{-0.1cm}
\subcaption{\small{Whole space.}}\label{fig:asymStruc3D}
%\vspace{-0.1cm}
\end{minipage}%
\begin{minipage}[t]{.33\linewidth}
\centering
%%\vspace{-0.1cm}
        % \includegraphics[height=4.7cm, width=5.1cm]{asymStruc2D1.eps}
        \includegraphics[scale=.42]{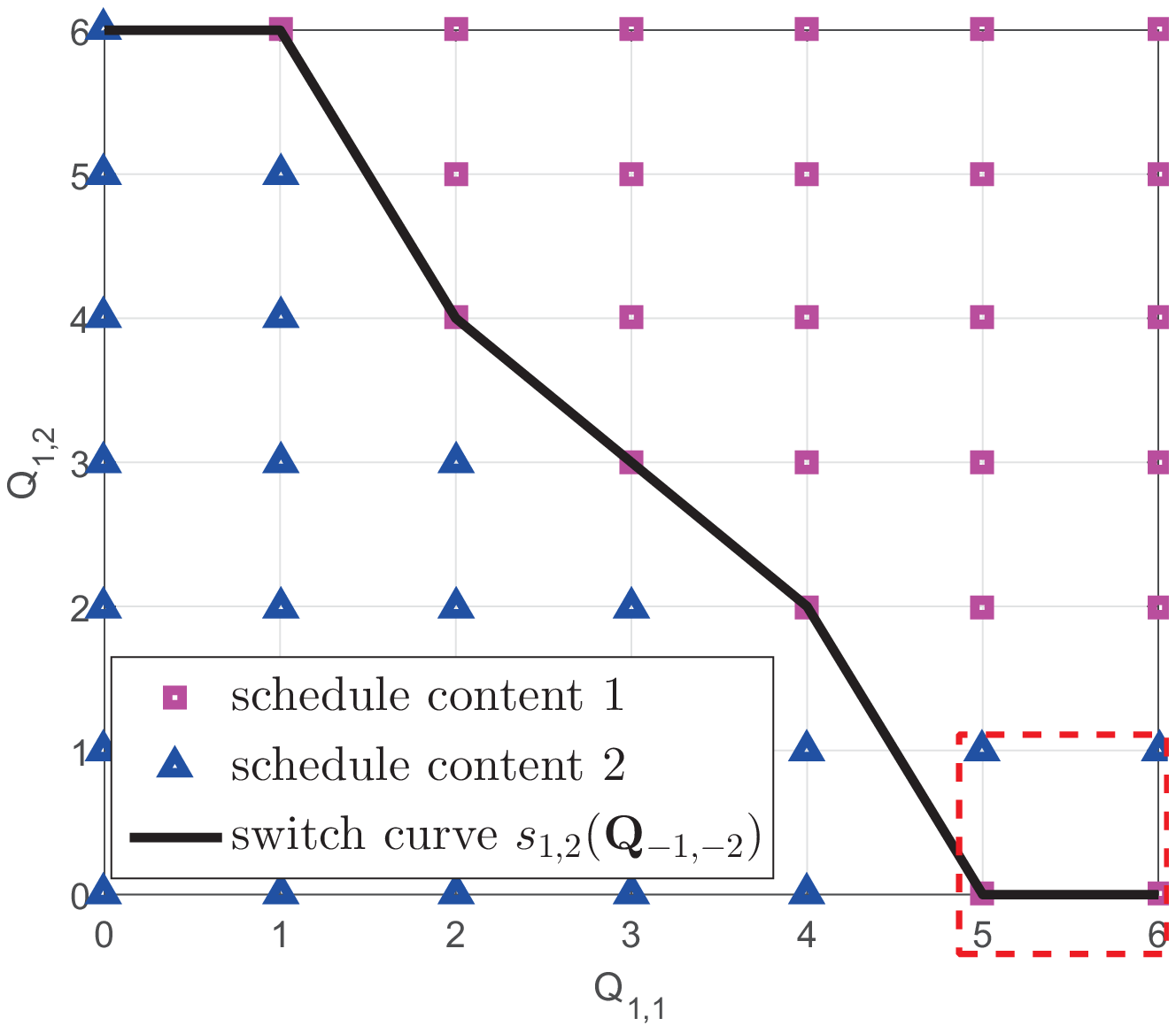}
        %\vspace{-0.1cm}
\subcaption{\small{Fixed $Q_{2,1}$}.}\label{fig:asymStruc2D1}
%\vspace{-0.1cm}
\end{minipage}
\begin{minipage}[t]{.33\linewidth}
\centering
%%\vspace{-0.1cm}
        % \includegraphics[height=4.7cm, width=5.1cm]{asymStruc2D2.eps}
        \includegraphics[scale=.42]{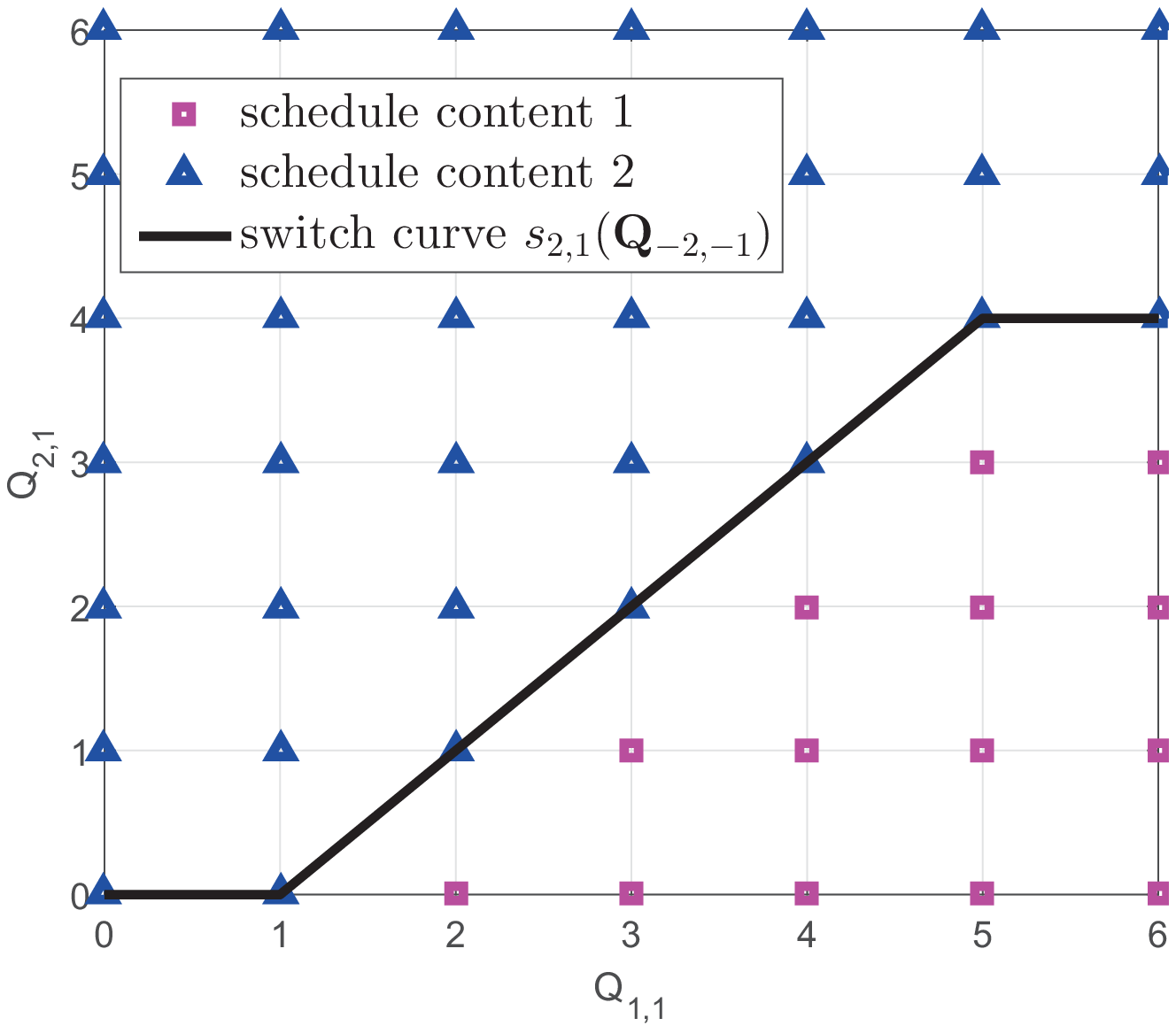}
       %\vspace{-0.1cm}
\subcaption{\small{Fixed $Q_{1,2}$}.}\label{fig:asymStruc2D2}
%\vspace{-0.1cm}
\end{minipage}
\caption{\small{Partial switch structure of optimal scheduling in the nonuniform case. Two-content, two-user case with $A_{2,2}(t)=0, \forall t.$}}\label{fig:nonuniform}
%\vspace{-0.3cm}
\end{figure*}
Fig.~\ref{fig:symStruc2q} illustrates the monotonicity of the switch curve.
We characterize the number of policies with monotonically non-decreasing switch curves in the following proposition.
\begin{proposition}
\label{proposition:num}
  For the uniform case with two contents, the number of the policies with monotonically non-decreasing switch curves is $N_1+N_2+2 \choose N_1+1$.
\end{proposition}
\begin{proof}
  Please see Appendix F.
\end{proof}

Table~\ref{reduction} shows that the space of possible optimal policies in the uniform case with two contents can be substantially reduced based on Lemma~\ref{lemma:propertyofswitch}.
\begin{table}[!h]
\renewcommand{\arraystretch}{0.9}
\small
%\small
%\multirow{3}{*}{\tabincell{c}{\\Queue\\Size}}
%%\vspace{-0.05cm}
% \begin{adjustbox}{max width=0.5\textwidth}
\begin{tabular}{|c|c|c|}
\hline
\tabincell{c}{Queue \\Size}&\tabincell{c}{Policy in \\ Definition 1}& \tabincell{c}{Policy with Monotonically \\ Non-decreasing Switch Curve} \\
\hline
$(N_1,N_2)$&$2^{(N_1+1)(N_2+1)}$ & $N_1+N_2+2 \choose N_1+1$\\
\hline
$(4,4)$&$3.36\times 10^7$&$252$\\
\hline
$(8,8)$&$2.42\times 10^{24}$&$48620$\\
\hline
\end{tabular}
% \end{adjustbox}
\centering
%%\vspace{-0.05cm}
\caption{\small{Policy space size in the uniform case at $M=2$.}}\label{reduction}
%\vspace{-0.2cm}
\end{table}

\section{Optimality Properties in Nonuniform Case}\label{sec:nonuniform}
In this section, we characterize the structure of the optimal policy for the nonuniform case.
Note that, different from the uniform case, the power cost $p(\mathbf{Q},u)$ in the nonuniform case also depends on the request queue state  $\mathbf{Q}$.
Therefore, due to the coupling among the request queues, the structural analysis for the nonuniform case is more challenging than that for the uniform case.

To analyze the structure of the optimal policy, we first introduce a new notation (see Fig.~\ref{fig:notation} for an example).
For each $m$, define $\mathbf{Q}_m^2\trianglerighteq\mathbf{Q}_m^1$ if and only if,
\begin{equation*}
  \left\{
    \begin{array}{ll}
      Q^2_{m,k}\geq Q^1_{m,k}, & \hbox{if $k\leq \max\{k|Q^1_{m,k}>0\}$;} \\
      Q^2_{m,k}=Q^1_{m,k}, & \hbox{otherwise.}
    \end{array}
  \right.\forall k.
\end{equation*}
Define $\mathbf{Q}^2\trianglerighteq\mathbf{Q}^1$ if and only if $\mathbf{Q}_m^2\trianglerighteq\mathbf{Q}_m^1$ for all $m$.

By RVIA and the special structures of the request queue dynamics, as well as the power and fetching costs, we can show the following property of $V(\mathbf{Q})$.
\begin{lemma}[Partial Monotonicity of Value Function]
%$V(\mathbf{Q}^1)\leq V(\mathbf{Q}^2)$,
In the nonuniform case, for any $\mathbf{Q}^1$, $\mathbf{Q}^2\in\bm{\mathcal{Q}}$ such that $\mathbf{Q}^2\trianglerighteq\mathbf{Q}^1$, we have $V(\mathbf{Q}^2)\geq V(\mathbf{Q}^1)$.
%The value function $V(Q)$ is monotonically non-decreasing in $\mathbf{Q}$, i.e.
\label{lemma:propertyV2}
\end{lemma}
\begin{proof}
  Please see Appendix G.
\end{proof}

Then, based on Lemma \ref{lemma:propertyV2} and the special properties of multicasting in the nonuniform channel case, we have the following property of $J(\mathbf{Q},u)$.

\begin{lemma}[Partial Monotonicity of State-Action Cost Function]
In the nonuniform case, for any $u,v\in\mathcal{M}$, $v\neq u$ and $\mathbf{Q}+\mathbf{E}_{u,k}\trianglerighteq\mathbf{Q}$,  we have
\begin{equation}
  J(\mathbf{Q}+\mathbf{E}_{u,k},u)-J(\mathbf{Q}+\mathbf{E}_{u,k},v)\leq J(\mathbf{Q},u)-J(\mathbf{Q},v), \label{eqn:propertyJ2}
\end{equation}
where $\mathbf{E}_{u,k}$ denotes the $M\times K$ matrix with all entries 0 except for a 1 in its $(u,k)$-th entry.
\label{lemma:propertyJ2}
\end{lemma}
\begin{proof}
  Please see Appendix H.
\end{proof}

Lemma~\ref{lemma:propertyJ2} indicates that if it is better to multicast content $u$ than content $v$ for some state $\mathbf{Q}$ (i.e., $J(\mathbf{Q},u)\leq J(\mathbf{Q},v)$) and $\mathbf{Q}+\mathbf{E}_{u,k}\trianglerighteq\mathbf{Q}$, then it is also better to multicast content $u$ than $v$ for state $\mathbf{Q}+\mathbf{E}_{u,k}$ (i.e., $J(\mathbf{Q}+\mathbf{E}_{u,k},u)\leq J(\mathbf{Q}+\mathbf{E}_{u,k},v)$). Thus, we have the following theorem.

\begin{theorem}[Partial Switch Structure of Optimal Policy]
The optimal policy $\mu^*$ in the nonuniform case has a partial switch structure, i.e., for all $u\in\mathcal{M}$ and $k\in\mathcal{K}$, we have
\begin{equation}\label{eqn:switch2}
~~~~~~~~ ~~ \mu^*(\mathbf{Q})=u,~ \hbox{\parbox[t]{.25\textwidth}{if $Q_{u,k}\geq s_{u,k}(\mathbf{Q}_{-u,-k})$ and$~~~~$   ~~~~~~~~~~~~~~~~condition (a) or (b) holds,}}
\end{equation}
where condition (a) is $k<k^\dag(k,\mathbf{Q}_u)$, condition (b) is $k> k^\dag(k,\mathbf{Q}_u)$ and $s_{u,k}(\mathbf{Q}_{-u,-k})>0$, and the switch curve for content-user pair $(u,k)$ is given by
\begin{equation*}
s_{u,k}(\mathbf{Q}_{-u,-k})\triangleq\begin{cases}\min\mathcal{S}_{u,k}(\mathbf{Q}_{-u,-k}),  & \text{if}~\mathcal{S}_{u,k}(\mathbf{Q}_{-u,-k})\neq\emptyset \\
            \infty,  &\text{otherwise}
  \end{cases}
  \end{equation*}
with $\mathcal{S}_{u,k}(\mathbf{Q}_{-u,-k})\triangleq\{Q_{u,k}| J(\mathbf{Q},u)\leq J(\mathbf{Q},v)~\forall v\in\mathcal{M}, v\neq u\}$.
Here, $\mathbf{Q}_{-u,-k}\triangleq(Q_{m,i})_{m\in\mathcal{M}, i\in\mathcal{K}, (m,i)\neq (u,k)}$ denotes the request queue state matrix corresponding to all the other content-user pairs except the content-user pair $(u,k)$, and $k^\dag(k,\mathbf{Q}_u)\triangleq\max\{i|Q_{u,i}>0,i\neq k\}$.
% \textcolor{black}{Moreover, $s_{u,k}(\mathbf{Q}_{-u,-k})$ is non-increasing with $Q_{u,i}$ for all $i<k$.}
\label{theorem:theorem2}
\end{theorem}
\begin{proof}
  Please see Appendix I.
\end{proof}

\begin{remark}
From Theorem~\ref{theorem:theorem2}, we can see that, the structure of the optimal policy in the nonuniform case is very similar to the one in the uniform case, as illustrated in Fig.~\ref{fig:nonuniform}. The only difference is that, the structural property for $k> k^\dag(k,\mathbf{Q}_u)$ and $s(\mathbf{Q}_{-u,-k})=0$ depends on the specific channel asymmetry among the users and is still not known in general, as illustrated in the dashed box of Fig.~\ref{fig:asymStruc2D1}.
% Insights from Theorem~\ref{theorem:theorem2} are similar to those from Theorem~\ref{theorem:theorem1}.
\textcolor{black}{Similar arguments on the tradeoff and the impacts of cache placement for the uniform case also hold for the nonuniform case.}
  % In addition, similar to the uniform case, we can observe the impacts of the fetching costs on the switch curves in Fig.~\ref{fig:nonuniform_fetch}.}
\end{remark}

\textcolor{black}{Similarly, note that, the partial switch structural property only relies on the partial monotonicity properties of $V(\mathbf{Q})$ and $J(\mathbf{Q},u)$.
These structural properties can also be used to reduce the computational complexity in obtaining the optimal policy for the nonuniform case, without
knowing the exact value of the switch curves.}
\textcolor{black}{Specifically,} from Theorem~~\ref{theorem:theorem2}, we know that, for all $\mathbf{Q}\in\bm{\mathcal{Q}}$,
\begin{equation}\label{eqn:indicate_nonuni}
  \mu^*(\mathbf{Q})=u~\text{and}~\mathbf{Q}+\mathbf{E}_{u,k}\trianglerighteq\mathbf{Q}~\Rightarrow~\mu^*(\mathbf{Q}+\mathbf{E}_{u,k})=u.
\end{equation}
Therefore, computing the optimal policy $\mu^*$ requires conducting the minimization in the R.H.S. of \eqref{eqn:mu} for some $\mathbf{Q}$ only (instead of all $\mathbf{Q}$), which significantly reduces the computational complexity.  Based on \textcolor{black}{\eqref{eqn:indicate_nonuni}}, we shall design low complexity optimal algorithms in Section VI.

\section{Low Complexity Optimal Algorithms}\label{sec:optimal_alg}
In this section, we propose two low complexity optimal algorithms for both the uniform and nonuniform cases, by exploiting the structural properties of the optimal policy in Theorems~\ref{theorem:theorem1} and \ref{theorem:theorem2}.
\subsection{Structured Relative Value Iteration Algorithm}
The optimal policy $\mu^*$ in \eqref{eqn:mu} can be computed using RVIA, which is a commonly used numerical method for solving infinite horizon average cost MDPs based on the Bellman equation\cite[Chapter 4.3]{bertsekas}  and is detailed in Appendix B.  We first summarize the standard RVIA for computing $\mu^*$ in Algorithm~\ref{alg:RVIA}.
It is shown in \cite[Proposition 4.3.2]{bertsekas} that under Algorithm~\ref{alg:RVIA}, for any $\{V_0(\mathbf{Q})\}$, we have $V_n(\mathbf{Q})\to V(\mathbf{Q})$ for all $\mathbf{Q}\in\bm{\mathcal{Q}}$, as $n\to\infty$, where $\{V(\mathbf{Q}\})$ satisfies the Bellman equation in \eqref{eqn:bellman}. Given $\{V(\mathbf{Q}\})$, we can obtain the optimal policy $\mu^*$ by \eqref{eqn:mu}.
Note that, in Step~\ref{code:rvia_eva} of the standard RVIA in Algorithm~\ref{alg:RVIA}, a brute-force minimization over $M$ actions needs to be computed for each $\mathbf{Q}\in\bm{\mathcal{Q}}$, which can be computationally expensive when $M$ is large.

By exploiting the structural properties of the optimal policy in Theorems~\ref{theorem:theorem1} and \ref{theorem:theorem2}, we modify \eqref{eqn:valueupdate} in Step~\ref{code:rvia_eva} (value update) of Algorithm~\ref{alg:RVIA} to reduce computational complexity. The modified step is given by Algorithm~\ref{alg:SRVIA} (structured value update).
% Note that, in Algorithm~\ref{alg:SRVIA}, we do not need to perform the minimization over $M$ actions in \eqref{eqn:valueupdate}, when the condition is satisfied (which is the case for most of $\mathbf{Q}\in\bm{\mathcal{Q}}$), and hence the computational saving  is evident.
Replacing \eqref{eqn:valueupdate} in Step~\ref{code:rvia_eva} of Algorithm~\ref{alg:RVIA} with Algorithm~\ref{alg:SRVIA}, we obtain a low complexity modified RVIA, which is referred to as the structured relative value iteration algorithm (SRVIA). From the proofs of Theorems~\ref{theorem:theorem1} and \ref{theorem:theorem2}, we can easily see that under SRVIA, for any $\{V_0(\mathbf{Q})\}$, $V_n(\mathbf{Q})\to V(\mathbf{Q})$ for  all $\mathbf{Q}\in\bm{\mathcal{Q}}$, as $n\to\infty$.
Similarly, given $V(\mathbf{Q})$, we can obtain $\mu^*$ in \eqref{eqn:mu}.
In other words, SRVIA is an optimal algorithm.

\begin{algorithm}[!t]
% \small
\caption{Relative Value Iteration Algorithm}
\label{alg:RVIA}
\begin{algorithmic}[1]
\State Set $V_0(\mathbf{Q})=0$ for all $\mathbf{Q}\in\bm{\mathcal{Q}}$, select reference state $\mathbf{Q}^\S$ and set $n=0$.
\State (Value Update) For each state $\mathbf{Q}\in\bm{\mathcal{Q}}$, compute $V_{n+1}(\mathbf{Q})$:\label{code:rvia_eva}
  %\Statex\hspace{0.4cm}$V_{n+1}(\mathbf{Q})=\min_{u\in\mathcal{M}}\left\{g(\mathbf{Q},u)+\mathbb{E}\left[V_n(\mathbf{Q}')\right]\right\}$,
  %\Statex\hspace{0.4cm}where $\mathbf{Q}'$ is defined in Lemma~\ref{lemma:bellman}.\label{code:rvia_eva}
  \begin{equation}\label{eqn:valueupdate}
    V_{n+1}(\mathbf{Q})=\min_{u\in\mathcal{M}}\left\{g(\mathbf{Q},u)+\mathbb{E}\left[V_n(\mathbf{Q}')\right]\right\},
  \end{equation}
  \Statex where $\mathbf{Q}'$ is defined in Lemma~\ref{lemma:bellman}.
\State For each state $\mathbf{Q}\in\bm{\mathcal{Q}}$, normalize $V_{n+1}(\mathbf{Q})$:\label{code:rvia_normalize}
  %\Statex\hspace{0.5cm}$V_{n+1}(\mathbf{Q})\leftarrow V_{n+1}(\mathbf{Q})-V_{n+1}(\mathbf{Q}^\S).$
  \begin{equation*}
    V_{n+1}(\mathbf{Q})\leftarrow V_{n+1}(\mathbf{Q})-V_{n+1}(\mathbf{Q}^\S).
  \end{equation*}
\State Go to Step~\ref{code:rvia_eva}.
% \State \textbf{if} $\max_{\mathbf{Q}\in\bm{\mathcal{Q}}}|V_{n+1}(\mathbf{Q})-V_{n}(\mathbf{Q})|<\epsilon$,  \textbf{then}
%             %\Statex\hspace{0.4cm} $\mu^*(\mathbf{Q})=\arg\min_{u\in\mathcal{M}}\left\{g(\mathbf{Q},u)+\mathbb{E}\left[V_{n}(\mathbf{Q}')\right]\right\},~\forall \mathbf{Q}$.
%           \begin{equation*}
%                  \mu^*(\mathbf{Q})=\arg\min_{u\in\mathcal{M}}\left\{g(\mathbf{Q},u)+\mathbb{E}\left[V_{n}(\mathbf{Q}')\right]\right\},~\forall \mathbf{Q}.
%           \end{equation*}
%         \Statex\textbf{else}
%             \Statex\hspace{10mm} increment $n$ by 1 and return to Step~\ref{code:rvia_eva}.
%         \Statex\textbf{end if}
\end{algorithmic}
\end{algorithm}

\begin{algorithm}[!t]
% \small
\caption{Structured Value Update}
\label{alg:SRVIA}
\begin{algorithmic}[1]
\If{$\exists u\in\mathcal{M}$, such that $\mu^*_n(\mathbf{Q}-\mathbf{e}_u)=u$ (uniform case),\Statex or
$\exists u\in\mathcal{M}$ and $k\in\mathcal{K}$, such that $\mu^*_n(\mathbf{Q}-\mathbf{E}_{u,k})=u$ and \Statex$\mathbf{Q}\trianglerighteq\mathbf{Q}-\mathbf{E}_{u,k}$ (nonuniform case),}
	    \begin{align*}
          &\mu^*_n(\mathbf{Q})=u,\\
          &V_{n+1}(\mathbf{Q})=g(\mathbf{Q},u)+\mathbb{E}\left[V_n(\mathbf{Q}')\right].
        \end{align*}
\Else
     \begin{align*}
          &V_{n+1}(\mathbf{Q})=\min_{u\in\mathcal{M}}\left\{g(\mathbf{Q},u)+\mathbb{E}\left[V_n(\mathbf{Q}')\right]\right\},\\
          &\mu^*_n(\mathbf{Q})=\arg\min_{u\in\mathcal{M}}\left\{g(\mathbf{Q},u)+\mathbb{E}\left[V_n(\mathbf{Q}')\right]\right\}.
        \end{align*}
\EndIf
\end{algorithmic}
\end{algorithm}

\subsection{Structured Policy Iteration Algorithm}
\begin{algorithm}[!t]
% \small
\caption{Policy Iteration Algorithm}
\label{alg:PIA}
\begin{algorithmic}[1]
\State Set $\mu^*_0(\mathbf{Q})=1$ for all $\mathbf{Q}\in\bm{\mathcal{Q}}$, select reference state $\mathbf{Q}^\S$, and set $n=0$.
\State (Policy Evaluation) Given policy $\mu_n^*$, compute the corresponding average cost $\theta_n$ and value function $V_n(\mathbf{Q})$ from the linear system of equations\footnotemark
%Note that, the linear system of equations \eqref{eqn:pia_eva} of the policy evaluation step in Algorithm~\ref{alg:PIA} (Step~\ref{code:pia_eva})
\begin{equation}
\begin{cases}
     \theta_n+V_n(\mathbf{Q})=g(\mathbf{Q},\mu_n^*(\mathbf{Q}))+\mathbb{E}\left[V_n(\mathbf{Q}')\right],~\forall \mathbf{Q}\in\bm{\mathcal{Q}}\\
  V_n(\mathbf{Q}^\S)=0
\end{cases}
\label{eqn:pia_eva}
\end{equation}
\Statex where $\mathbf{Q}'$ is defined in Lemma~\ref{lemma:bellman}.\label{code:pia_eva}
\State (Policy Update) Obtain a new policy $\mu_{n+1}^*$, where for each $\mathbf{Q}\in\bm{\mathcal{Q}}$, $\mu_{n+1}^*(\mathbf{Q})$ is such that:
  %\Statex\hspace{0.4cm}$\mu_{n+1}(\mathbf{Q})=\arg\min_{u\in\mathcal{M}}\left\{g(\mathbf{Q},u)+\mathbb{E}\left[V_n(\mathbf{Q}')\right]\right\}.$
    \begin{equation}
      \mu_{n+1}^*(\mathbf{Q})=\arg\min_{u\in\mathcal{M}}\left\{g(\mathbf{Q},u)+\mathbb{E}\left[V_n(\mathbf{Q}')\right]\right\}.\label{eqn:pia_imp}
    \end{equation}\label{code:pia_imp}
% \State \textbf{if} $\mu_{n+1}=\mu_n$,  \textbf{then}
%             %\Statex\hspace{0.4cm} $\mu^*=\mu_n$.
%     \begin{equation*}
%       \mu^*=\mu_n.
%     \end{equation*}
%         \Statex\textbf{else}
%             \Statex\hspace{10mm} increment $l$ by 1 and return to Step~\ref{code:pia_eva}.
%         \Statex\textbf{end if}
\State Go to Step \ref{code:pia_eva} until $\mu_{n+1}^*=\mu_{n}^*$.
\end{algorithmic}
\end{algorithm}
\footnotetext{The solution to \eqref{eqn:pia_eva} can be obtained directly using Gaussian elimination or iteratively using the relative value iteration method\cite{bertsekas}.}
%%\vspace{-0.4cm}
\begin{algorithm}[!t]
% \small
\caption{Structured Policy Update}
\label{alg:SPIA}
\begin{algorithmic}[1]
\If{$\exists u\in\mathcal{M}$, such that $\mu_{n+1}^*(\mathbf{Q}-\mathbf{e}_u)=u$ (uniform case),\Statex or $\exists u\in\mathcal{M}$ and $k\in\mathcal{K}$, such that $\mu_{n+1}^*(\mathbf{Q}-\mathbf{E}_{u,k})=u$ and\Statex$\mathbf{Q}\trianglerighteq\mathbf{Q}-\mathbf{E}_{u,k}$ (nonuniform case),}
        \begin{equation*}
          \mu_{n+1}^*(\mathbf{Q})=u.
        \end{equation*}
\Else
        \begin{equation*}
            \mu_{n+1}^*(\mathbf{Q})=\arg\min_{u\in\mathcal{M}}\left\{g(\mathbf{Q},u)+\mathbb{E}\left[V_n(\mathbf{Q}')\right]\right\}.
        \end{equation*}
\EndIf			
\end{algorithmic}
\end{algorithm}
The optimal policy $\mu^*$ in \eqref{eqn:mu} can also be computed using the policy iteration algorithm (PIA), which is another widely used method for solving infinite horizon average cost MDPs\cite[Chapter 8.6]{puterman}.
We summarize PIA for computing $\mu^*$ in Algorithm~\ref{alg:PIA}.
According to \cite[Theorem 8.6.6]{puterman}, Algorithm~\ref{alg:PIA} converges in a finite number of iterations to the optimal policy $\mu^*$ in \eqref{eqn:mu}.
In other words, there exists a finite $\bar{n}$ such that $\mu_n^*=\mu^*$ for all $n\geq\bar{n}$.
Note that, in Step~\ref{code:pia_imp} of the standard PIA in Algorithm~\ref{alg:PIA}, a brute-force minimization over $M$ actions needs to be computed for each $\mathbf{Q}\in\bm{\mathcal{Q}}$, which can be computationally complex when $M$ is large.

By exploiting the structural properties of the optimal policy in Theorems~\ref{theorem:theorem1} and \ref{theorem:theorem2}, we modify \eqref{eqn:pia_imp} in Step~\ref{code:pia_imp} (policy update) of Algorithm~\ref{alg:PIA} to reduce its computational complexity.
The modified step is given by Algorithm~\ref{alg:SPIA} (structured policy update).
% Note that, in Algorithm~\ref{alg:SPIA}, we do not need to perform the minimization over $M$ actions in \eqref{eqn:pia_imp}, when the condition is satisfied  (which is the case for most of $\mathbf{Q}\in\bm{\mathcal{Q}}$).
% Similar to Algorithm~\ref{alg:SRVIA}, the computational saving of Algorithm~\ref{alg:SPIA} is also obvious.
Replacing \eqref{eqn:pia_imp} in Step~\ref{code:pia_imp} of Algorithm~\ref{alg:PIA} with Algorithm~\ref{alg:SPIA}, we obtain a low complexity PIA, which is referred to as the structured policy iteration algorithm (SPIA). From \cite[Chapter 8.11.2]{puterman}, we can see that SPIA also converges in a finite number of iterations to the optimal policy $\mu^*$ in \eqref{eqn:mu} and hence is an optimal algorithm.

\begin{answer}
\subsection{Complexity Comparison}

We compare the computational complexity of the proposed structured optimal algorithms (SRVIA and SPIA) with the standard optimal algorithms (RVIA and PIA) for each iteration, as illustrated in Table~\ref{table:theoretic_complexity}. Specifically, in the structured value update step (Algorithm~\ref{alg:SRVIA}) and the structured policy update step (Algorithm~\ref{alg:SPIA}), if the condition is satisfied for a certain queue state, then we do not need to perform the corresponding minimization over $M$ actions. This leads to a computational saving of $O(M|\bm{\mathcal{Q}}|)$ \cite{complexity}. There are in total $|\bm{\mathcal{Q}}|$ states. Thus, for each iteration, the computational complexity reduction of the structured value/policy update is  $O(M|\bm{\mathcal{Q}}|^2)$. From Table~\ref{table:theoretic_complexity}, we can see that, although the proposed structured optimal algorithms suffer from the exponential growth of the state space, the computational complexity reduction also grows exponentially with the state space.
Therefore, the computational complexity reduction of the proposed structured optimal algorithms is remarkable, considering that the optimality is not sacrificed.
\begin{table}[!thbp]
\color{black}
\begin{adjustbox}{max width=0.49\textwidth}
\begin{tabular}{|c|c|c|c|}
\hline
& \tabincell{c}{Complexity of \\Standard Algs.}  & \tabincell{c}{Complexity of Proposed\\Structured Algs.}  & \tabincell{c}{Complexity \\Reduction}\\
\hline
Value Update & $O(M|\bm{\mathcal{Q}}|^2)$ & $O(M|\bm{\mathcal{Q}}|^2)$ & $O(M|\bm{\mathcal{Q}}|^2)$ \\
\hline
Policy Update & $O(M|\bm{\mathcal{Q}}|^2)$ & $O(M|\bm{\mathcal{Q}}|^2)$ & $O(M|\bm{\mathcal{Q}}|^2)$ \\
\hline
\end{tabular}
\end{adjustbox}
\centering
\caption{\textcolor{black}{Complexity comparison between the proposed and standard optimal algorithms in each iteration.}}
\label{table:theoretic_complexity}
\end{table}

Note that the two proposed low-complexity optimal algorithms still suffer from the curse of dimensionality, i.e., the exponential dependence of the state space \cite{bertsekas}.
This curse of dimensionality comes from the complex coupling structure of the request queue model, and is embedded in the optimal control design for the considered MDP. To the best of our knowledge, unless for very special cases, it is not possible to break the curse of dimensionality without any loss of optimality.
% In addition, the structural properties of the optimal policy may be one key factor that leads to good performance. Therefore, we conjecture that good suboptimal solutions may possess the switch structures and we shall develop a low complexity suboptimal solution which can maintain similar structural properties to the optimal policy.
\end{answer}
\section{Low Complexity Suboptimal Solution}\label{sec:suboptimal_alg}
\textcolor{black}{To further reduce the complexity of the proposed structured optimal algorithms and relieve the curse of dimensionality, we would like to develop low-complexity suboptimal solutions.
Note that the structural properties of the optimal policy may be one key factor that leads to good performance. Thus, in this section, we focus on the design of suboptimal solutions which can maintain the switch structures. Specifically,}
based on a randomized base policy, we first propose a low complexity suboptimal deterministic policy using approximate dynamic programming\cite{bertsekas}, which has better performance than the randomized base policy and possesses \textcolor{black}{similar structural properties to} the optimal policy. Then, based on these structural properties, we develop a low complexity structured algorithm to compute the proposed policy.
Note that, with abuse of notation, in this section, we also use $Q_m$ and $\mathcal{Q}_m$ to represent $\mathbf{Q}_m$ and $\bm{\mathcal{Q}}_m$ in the nonuniform case.

\subsection{Low Complexity Suboptimal Policy}
\textcolor{black}{The switch structural properties of the optimal policy stem from the monotonicity property of the value function.
Therefore, to maintain the switch structures in designing a suboptimal solution, we consider a value function decomposition method that can preserve the structural properties of the value function. Based on this decomposition, we propose a low complexity suboptimal deterministic policy, which will be shown to possess similar structural properties to the optimal policy.}
% Specifically, we introduce a randomized base policy and prove that its value function has an additive separable structure.
% Based on this randomization, we propose a low complexity suboptimal deterministic policy by approximating the value function as the summation of the per-content value functions.
% We further show that the proposed suboptimal deterministic policy possesses similar structural properties to the optimal policy.
We first introduce a randomized base policy.
\begin{definition}[Randomized Base Policy]\label{definition:definition2}
A randomized base policy for the multicast scheduling control $\hat{\mu}$ is given by a distribution on the multicast scheduling action space $\mathcal{M}$.
\end{definition}

We restrict our attention to randomized unichain base policies.
Denote $\hat{\theta}$ and $\{\hat{V}(\mathbf{Q})\}$ as the average cost and the value function under a randomized unichain based  policy $\hat{\mu}$, respectively.
By \cite[Proposition 4.2.2]{bertsekas} and following the proof of Lemma~\ref{lemma:bellman}, there exists  $(\hat{\theta},\{\hat{V}(\mathbf{Q})\})$ satisfying:
  \begin{align}
    &\hat{\theta}+\hat{V}(\mathbf{Q})=\mathbb{E}^{\hat{\mu}}\left[g(\mathbf{Q},u)\right]\nonumber\\
&\hspace{13mm}+\sum_{\mathbf{Q}'\in\bm{\mathcal{Q}}}\mathbb{E}^{\hat{\mu}}\left[\Pr[\mathbf{Q}'|\mathbf{Q},u]\right]\hat{V}(\mathbf{Q}'),~\forall \mathbf{Q}\in\bm{\mathcal{Q}},\label{eqn:randombellman}
  \end{align}
where $g(\mathbf{Q},u)$ and $\Pr[\mathbf{Q}'|\mathbf{Q},u]$ are given in \eqref{eqn:cost} and \eqref{eqn:tranb}, respectively.	
In the following lemma, we show that $\hat{V}(\mathbf{Q})$ has a separable structure.
\begin{lemma}[Separable Structure of $\hat{V}(\mathbf{Q})$]
  Given any randomized unichain base policy $\hat{\mu}$, the value function $\{\hat{V}(\mathbf{Q})\}$ in \eqref{eqn:randombellman} can be expressed as $\hat{V}(\mathbf{Q})=\sum_{m\in\mathcal{M}}\hat{V}_m(Q_m)$, where $\{\hat{V}_m(Q_m)\}$ satisfies:
%    \begin{align}
%    &\hat{\theta}_m+\hat{V}_m(Q_m)=\hat{g}_m(Q_m)+\sum_{m\in\mathcal{M}}\hat{\Pr}[Q_m'|Q_m]\hat{V}_m(Q_m')\nonumber\\
%     &\hspace{60mm}\forall Q_m\in\mathcal{Q}_m.
%  \end{align}
%  Here, $\hat{g}_m(Q_m)\triangleq \mathbb{E}^{\hat{\mu}}\left[g_m(Q_m,u)\right]$ and $\hat{\Pr}[Q_m'|Q_m]\triangleq \mathbb{E}^{\hat{\mu}}\left[\Pr[Q_m'|Q_m,u]\right]$, where
  \begin{align}
    &\hat{\theta}_m+\hat{V}_m(Q_m)=\mathbb{E}^{\hat{\mu}}\left[g_m(Q_m,u)\right]\nonumber\\
    &+\sum_{Q_m'\in\mathcal{Q}_m}\mathbb{E}^{\hat{\mu}}\left[\Pr[Q_m'|Q_m,u]\right]\hat{V}_m(Q_m'),~\forall Q_m\in\mathcal{Q}_m,
\label{eqn:perrandombellman}
  \end{align}
for all $m\in\mathcal{M}$.  Here, $\hat{\theta}_m$ and $\hat{V}_m(Q_m)$ denote the per-content average cost and value function under $\hat{\mu}$, respectively,   $g_m(Q_m,u)\triangleq Q_m+w_ff(u)+w_p\mathbf{1}(u=m)p(m)$ in the uniform case, $g_m(Q_m,u)\triangleq \sum_{k\in\mathcal{K}}Q_{m,k}+w_ff(u)+w_p\mathbf{1}(u=m)p(m,k^\ddag(Q_m,u))$ with $k^\ddag(Q_m,u)\triangleq\max\{k|Q_{m,k}>0\}$ in the nonuniform case, and $\Pr[Q_m'|Q_m,u]\triangleq\Pr[Q_m(t+1)=Q_m'|Q_m(t)=Q_m,u(t)=u]$.
\label{lemma:decomposition}
\end{lemma}
\begin{proof}
  Please see Appendix J.
\end{proof}

To alleviate the curse of dimensionality, we approximate the value function $V(\mathbf{Q})$ in \eqref{eqn:mu} by $\hat{V}(\mathbf{Q})$, i.e.,
\begin{equation}
  V(\mathbf{Q})\thickapprox\hat{V}(\mathbf{Q})=\sum_{m\in\mathcal{M}}\hat{V}_m(Q_m),\label{eqn:approx}
\end{equation}
where $\{\hat{V}_m(Q_m)\}$ is given by the per-content fixed point equation in \eqref{eqn:perrandombellman}.
Then, we obtain the following low complexity deterministic policy $\hat{\mu^*}$:
\begin{align}
  &\hat{\mu}^*(\mathbf{Q})=\arg\min_{u\in\mathcal{M}}\left\{g(\mathbf{Q},u)+\sum_{m\in\mathcal{M}}\mathbb{E}\left[\hat{V}_m(Q_m'))\right]\right\},\nonumber\\
  &\hspace{65mm}\forall \mathbf{Q}\in\bm{\mathcal{Q}}.\label{eqn:hatmu}
\end{align}

\begin{remark}
To obtain $\hat{\mu}^*$ in \eqref{eqn:hatmu},  we only need to compute $\{\hat{V}_m(Q_m)\}$ (a total of $O(\sum_{m}|\mathcal{Q}_m|)$ values)  via solving \eqref{eqn:perrandombellman} for all $m$. The computational complexity is much lower than computing $\{V(\mathbf{Q})\}$ (a total of $O(\Pi_{m}|\mathcal{Q}_m|)$ values) via solving \eqref{eqn:bellman} in obtaining $\mu^*$ in \eqref{eqn:mu}.
\label{remark:sub}
\end{remark}

\begin{remark}
\textcolor{black}{
Note that, the value function decomposition method adopted here is different from most existing approximate approaches \cite{powell2007approximate,factoredMDP}. Our approach does not rely on  choices of specific basis functions. Moreover, our approach can maintain similar switch structural properties to the optimal policy, which will be shown in Theorem~\ref{theorem:strhatmu}.
}
\end{remark}

\subsection{Properties of Suboptimal Policy}
\subsubsection{Performance comparison}
The proposed deterministic policy $\hat{\mu}^*$ always achieves better performance than the randomized unichain base policy $\hat{\mu}$, as summarized in the following theorem.
\begin{theorem}[Performance Improvement]
  If $\Pr[\mathbf{Q}'|\mathbf{Q},u]\neq\Pr[\mathbf{Q}'|\mathbf{Q},u']$ for any $u\neq u'$ and $\mathbf{Q}\in\bm{\mathcal{Q}}$, then we have $\hat{\theta}^*(\mathbf{Q})\textcolor{black}{<}\hat{\theta}$ for all $\mathbf{Q}\in\bm{\mathcal{Q}}$, where $\hat{\theta}^*(\mathbf{Q})$ is the average system cost under the proposed solution starting from $\mathbf{Q}$ and $\hat{\theta}$ is the average system cost under any randomized base policy, respectively.
\end{theorem}
\begin{proof}
  This result follows directly from \cite{harvest}.
\end{proof}
\subsubsection{Structural properties}
We first introduce the state-action cost function for $\hat{\mu}^*$:
\begin{equation}
  \hat{J}(\mathbf{Q},u)\triangleq g(\mathbf{Q},u)+\sum_{m\in\mathcal{M}}\mathbb{E}\left[\hat{V}_m(Q_m'))\right].\label{eqn:sub_state_action_func}
\end{equation}
Note that $\hat{J}(\mathbf{Q},u)$ is related to the R.H.S. of \eqref{eqn:hatmu}.
Along the lines of the structural analysis for the optimal policy in Sections IV and V, we can show that the proposed deterministic low complexity suboptimal policy has \textcolor{black}{similar} structural properties to those in Theorems~\ref{theorem:theorem1} and \ref{theorem:theorem2}.
\textcolor{black}{This similarity may be one key reason for the good performance of the proposed suboptimal solution, which will be shown in the numerical section.}

\begin{theorem}[Structural Properties of Suboptimal Policy $\hat{\mu}^*$]
    Under any randomized base unichain policy $\hat{\mu}$, the structural properties of the corresponding deterministic policy $\hat{\mu}^*$ are as follows.
%\begin{enumerate}
  %\item

1) In the uniform case, $\hat{\mu}^*$ has a switch structure, i.e., for all $m\in\mathcal{M}$, we have
\begin{equation}\label{eqn:subswitch}
  \hat{\mu}^*(\mathbf{Q})=u, \text{if}~Q_u\geq \hat{s}_u(\mathbf{Q}_{-u}),
\end{equation}
where the switch curve for content $u$ is given by
\begin{equation*}%\label{eqn:subcurve}
\hat{s}_u(\mathbf{Q}_{-u})\triangleq\begin{cases}\min\hat{\mathcal{S}}_u(\mathbf{Q}_{-u}),  & \text{if}~\hat{\mathcal{S}}_u(\mathbf{Q}_{-u})\neq\emptyset \\
            \infty,  &\text{otherwise}
  \end{cases}
  \end{equation*}
with $\hat{\mathcal{S}}_u(\mathbf{Q}_{-u})\triangleq\{Q_u| \hat{J}(\mathbf{Q},u)\leq \hat{J}(\mathbf{Q},v)~\forall v, v\neq u\}$. Here, $\mathbf{Q}_{-u}$ is defined in Theorem~\ref{theorem:theorem1}.
%  \item

2) In the nonuniform case, $\hat{\mu}^*$ has a partial switch structure, i.e., for all $u\in\mathcal{M}$ and $k\in\mathcal{K}$, we have
\begin{equation}\label{eqn:subswitch2}
~~~~~~~~ \hat{\mu}^*(\mathbf{Q})=u,~ \hbox{\parbox[t]{.25\textwidth}{if $Q_{u,k}\geq \hat{s}_{u,k}(\mathbf{Q}_{-u,-k})$ and$~~~~$   ~~~~~~~~~~~~~~~~condition (a) or (b) holds,}}
\end{equation}
where condition (a) is $k<k^\dag(k,Q_u)$, condition (b) is $k> k^\dag(k,Q_u)$ and $\hat{s}_{u,k}(\mathbf{Q}_{-u,-k})>0$, and the switch curve for content-user pair $(u,k)$ is given by
\begin{equation*}
\hat{s}_{u,k}(\mathbf{Q}_{-u,-k})\triangleq\begin{cases}\min\hat{\mathcal{S}}_{u,k}(\mathbf{Q}_{-u,-k}),  & \text{if}~\hat{\mathcal{S}}_{u,k}(\mathbf{Q}_{-u,-k})\neq\emptyset \\
            \infty,  &\text{otherwise}
  \end{cases}
  \end{equation*}
with $\hat{\mathcal{S}}_{u,k}(\mathbf{Q}_{-u,-k})\triangleq\{Q_{u,k}| \hat{J}(\mathbf{Q},u)\leq \hat{J}(\mathbf{Q},v)~\forall v, v\neq u\}$.
% \textcolor{black}{Moreover, $\hat{s}_{u,k}(\mathbf{Q}_{-u,-k})$ is non-increasing with $Q_{u,i}$ for all $i<k$.}
Here, $k^\dag(k,Q_u)$ and $\mathbf{Q}_{-u,-k}$ are defined in Theorem~\ref{theorem:theorem2}.
%\end{enumerate}
\label{theorem:strhatmu}
\end{theorem}
\begin{proof}
  Please see Appendix K.
\end{proof}

Similarly, Theorem~\ref{theorem:strhatmu}  indicates the following results.

 1) In the uniform case, for all $\mathbf{Q}\in\bm{\mathcal{Q}}$, we have
\begin{equation}\label{eqn:indicate_uni_sub}
  \hat{\mu}^*(\mathbf{Q})=u~\Rightarrow~\hat{\mu}^*(\mathbf{Q}+\mathbf{e}_u)=u.
\end{equation}

2) In the nonuniform case, for all $\mathbf{Q}\in\bm{\mathcal{Q}}$, we have
\begin{equation}\label{eqn:indicate_nonuni_sub}
  \hat{\mu}^*(\mathbf{Q})=u~\text{and}~\mathbf{Q}+\mathbf{E}_{u,k}\trianglerighteq\mathbf{Q}~\Rightarrow~\hat{\mu}^*(\mathbf{Q}+\mathbf{E}_{u,k})=u.
\end{equation}
\subsection{Structured Suboptimal Algorithm}
By making use of the relationship between $\hat{\mu}$ and $\hat{\mu}^*$ and the structural properties of $\hat{\mu}^*$ in Theorem~\ref{theorem:strhatmu}, we can develop a low complexity algorithm to obtain $\hat{\mu}^*$ in \eqref{eqn:hatmu}, which is summarized in Algorithm~\ref{alg:ssa}. We refer to Algorithm~\ref{alg:ssa} as the structured suboptimal algorithm (SSA).
\begin{algorithm}[!h]
% \small
\caption{Structured Suboptimal Algorithm}
\label{alg:ssa}
\begin{algorithmic}[1]
\State Given a randomized base unichain policy $\hat{\mu}$, compute the per-content value function $\{\hat{V}_m(Q_m)\}$ for all $m\in\mathcal{M}$ by solving  the corresponding linear system of equations in \eqref{eqn:perrandombellman}.\footnotemark \label{code:hatmu}
\State Obtain the proposed policy $\hat{\mu}^*$, where for each $\mathbf{Q}\in\bm{\mathcal{Q}}$, $\hat{\mu}^*(\mathbf{Q})$ is such that:\label{code:ssa}
\Statex\textbf{if}~$\exists u\in\mathcal{M}$, such that $\hat{\mu}^*(\mathbf{Q}-\mathbf{e}_u)=u$ (uniform case),\Statex or $\exists u\in\mathcal{M}$ and $k\in\mathcal{K}$, such that $\hat{\mu}^*(\mathbf{Q}-\mathbf{E}_{u,k})=u$ and \Statex$\mathbf{Q}\trianglerighteq\mathbf{Q}-\mathbf{E}_{u,k}$ (nonuniform case),~\textbf{then}
$$\hat{\mu}^*(\mathbf{Q})=u.$$
\Statex\textbf{else}
\Statex \hspace{25mm}Compute  $\hat{\mu}^*(\mathbf{Q})$ by \eqref{eqn:hatmu}.
\Statex\textbf{endif}
\end{algorithmic}
\end{algorithm}
\footnotetext{The solution to \eqref{eqn:perrandombellman} can be obtained directly using Gaussian elimination or iteratively using the relative value iteration method \cite{bertsekas}.}

\begin{figure*}[!t]
%%\vspace{-0.1cm}
\begin{minipage}[t]{0.245\linewidth}
\centering
%%\vspace{-0.1cm}
        % \includegraphics[height=4.3cm, width=4.4cm]{uniform_vs_wpwf_system_cost_new.eps}
        \includegraphics[scale=.3]{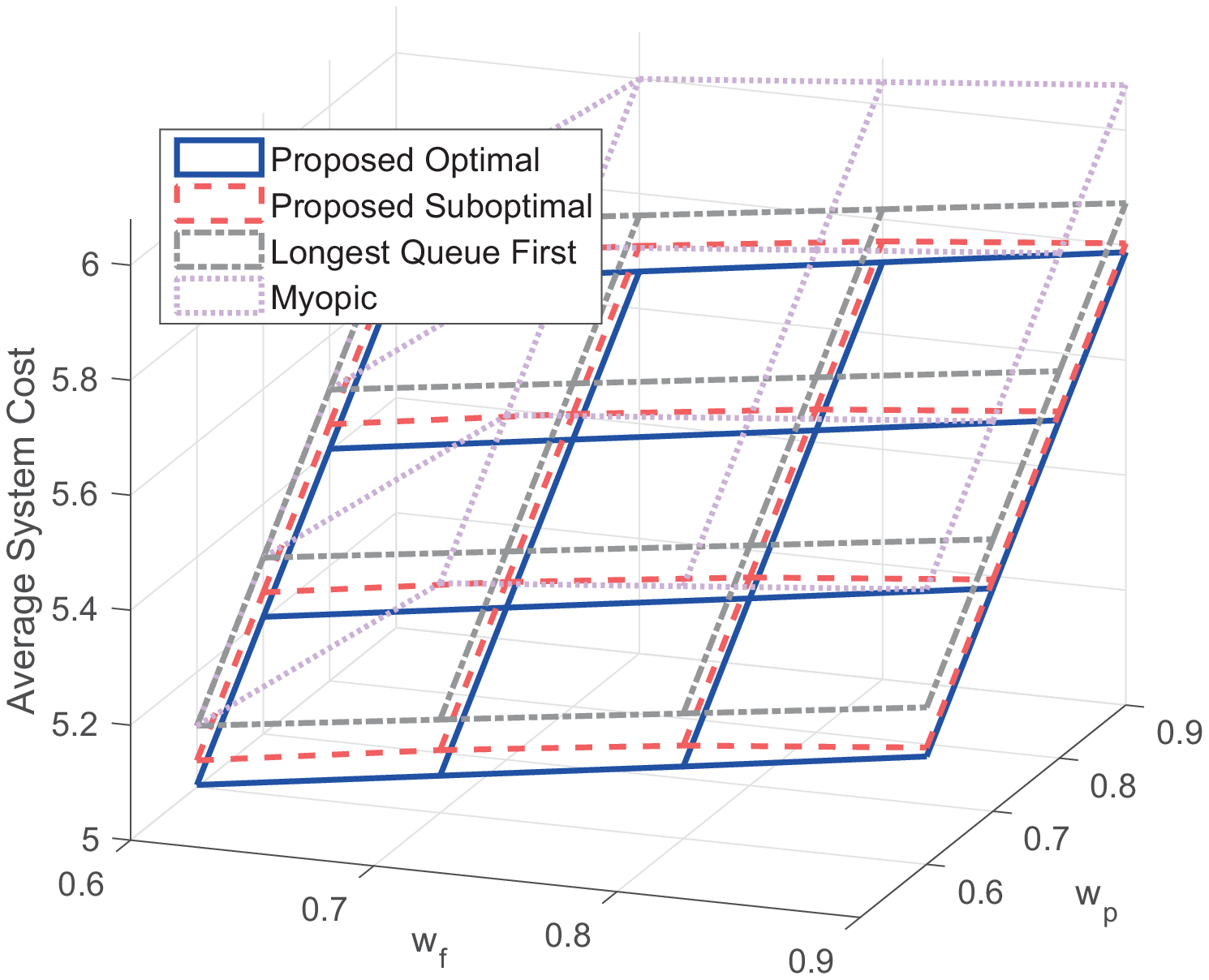}
%\includegraphics[scale=.3]{asymStruc3D.eps}
%        %\vspace{-0.1cm}
\subcaption{\small{Average system cost}}\label{fig:uniform_wpwf_system}
%%\vspace{-0.1cm}
\end{minipage}%
\begin{minipage}[t]{.245\linewidth}
\centering
%%\vspace{-0.1cm}
% \includegraphics[height=4.3cm, width=4.4cm]{uniform_vs_wpwf_delay_cost_new.eps}
        \includegraphics[scale=.3]{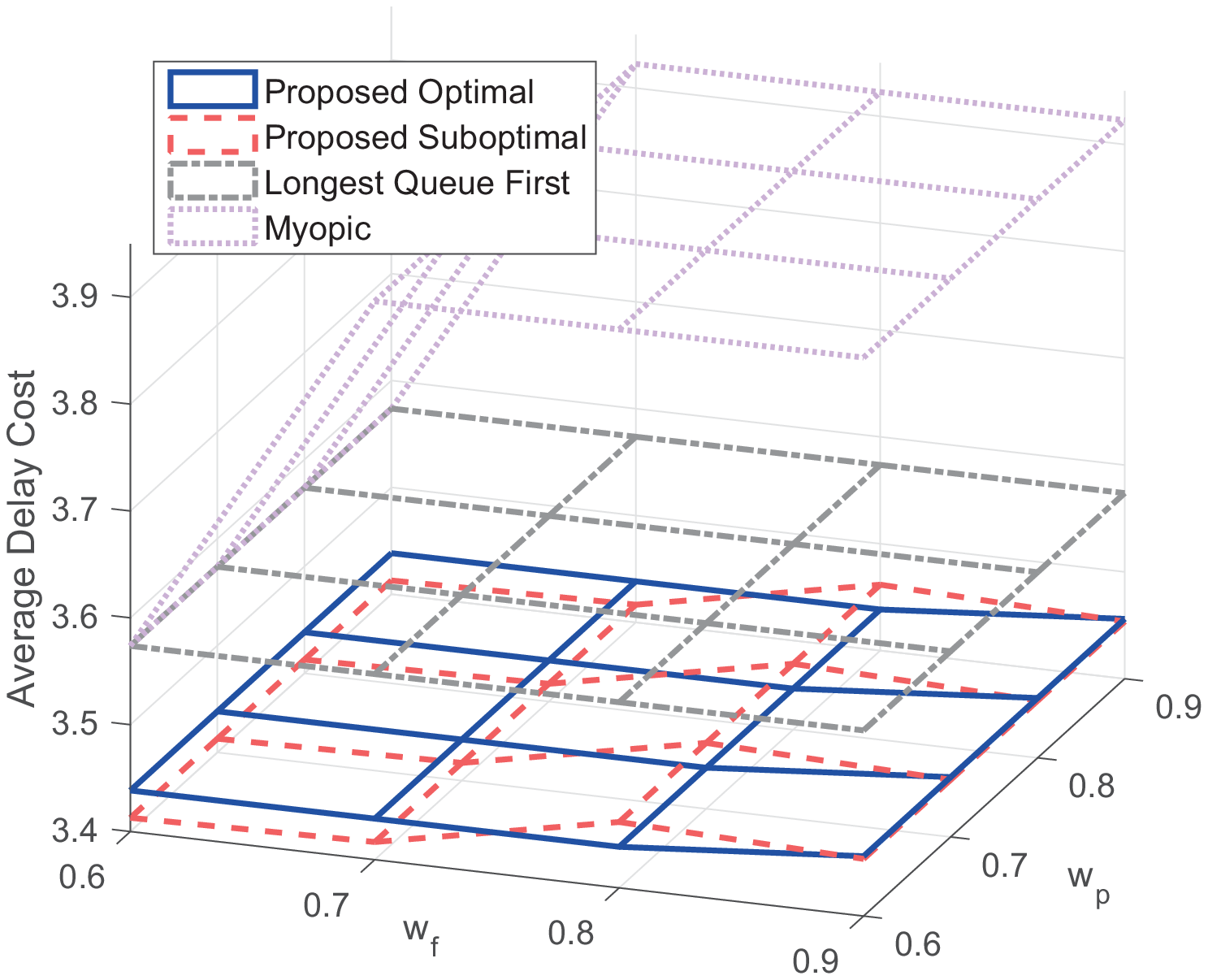}
%        %\vspace{-0.1cm}
\subcaption{\small{Average delay cost}.}\label{fig:uniform_wpwf_delay}
%%\vspace{-0.1cm}
\end{minipage}
\begin{minipage}[t]{.245\linewidth}
\centering
%%\vspace{-0.1cm}
% \includegraphics[height=4.3cm, width=4.4cm]{uniform_vs_wpwf_fetch_cost_new.eps}
        \includegraphics[scale=.3]{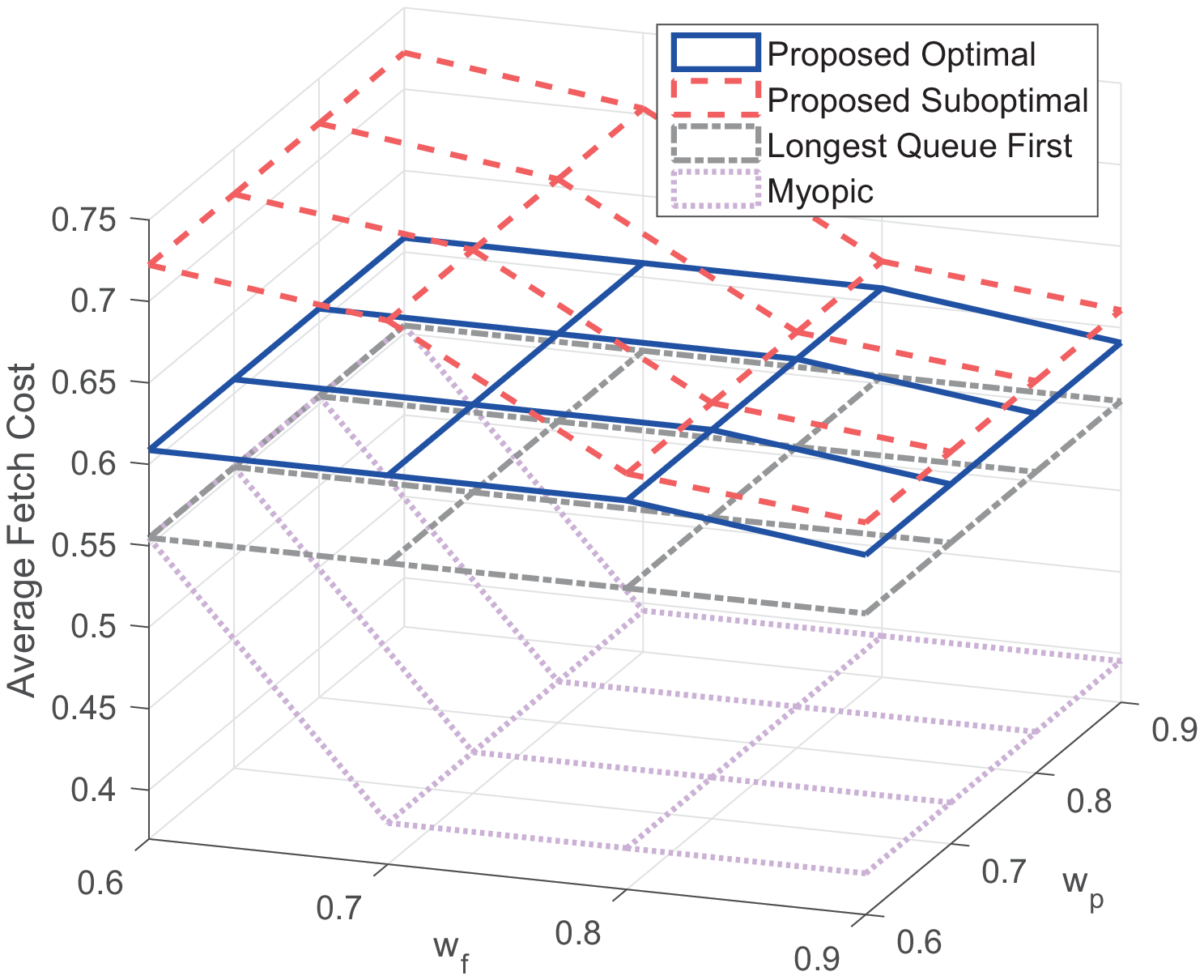}%        %\vspace{-0.1cm}
\subcaption{\small{Average fetching cost}.}\label{fig:uniform_wpwf_fetch}
%%\vspace{-0.1cm}
\end{minipage}
\begin{minipage}[t]{.245\linewidth}
\centering
%%\vspace{-0.1cm}
% \includegraphics[height=4.3cm, width=4.4cm]{uniform_vs_wpwf_power_cost_new.eps}
        \includegraphics[scale=.3]{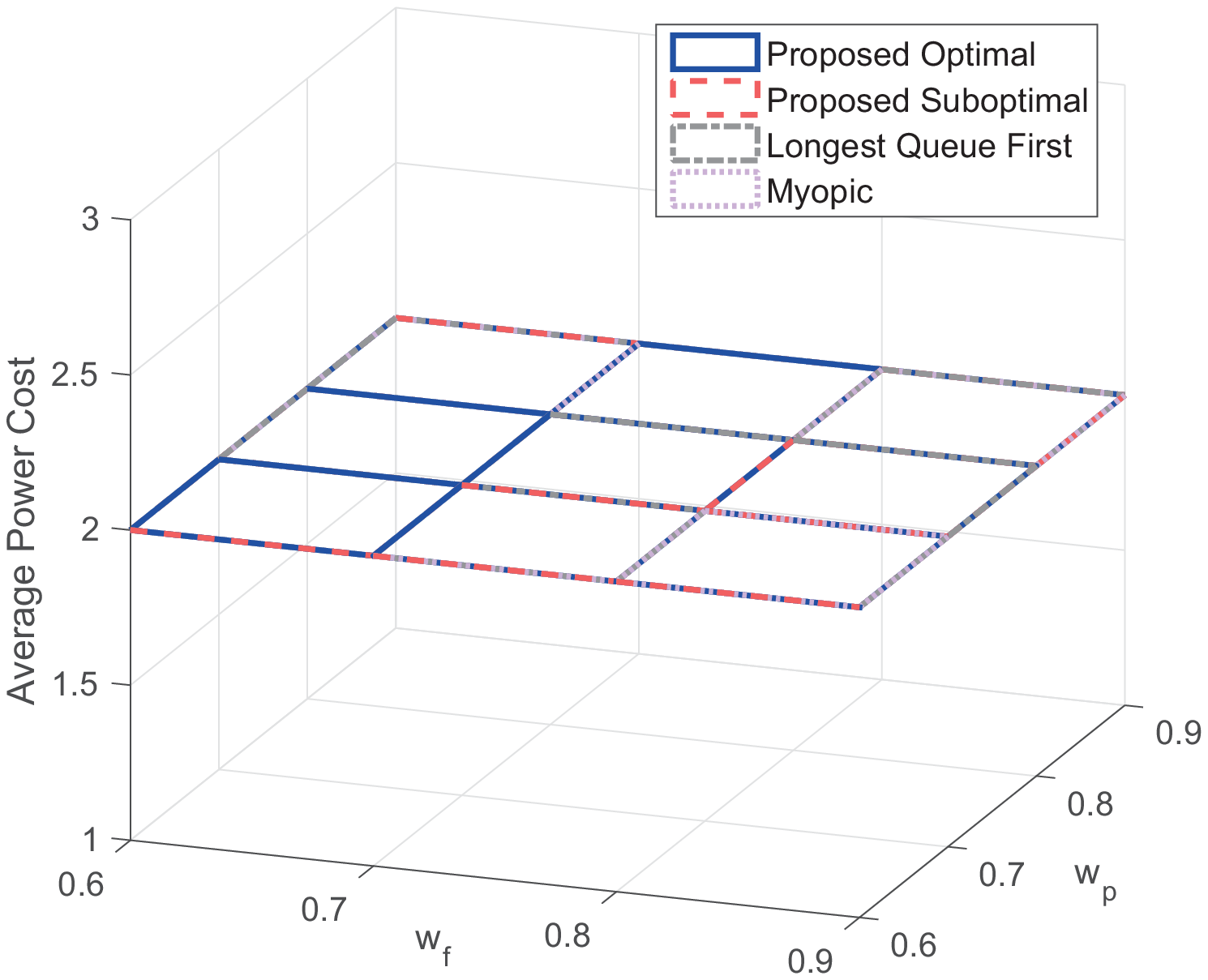}%%\vspace{-0.1cm}
%        %\vspace{-0.1cm}
\subcaption{\small{Average power cost}.}\label{fig:uniform_wpwf_power}
\end{minipage}
\caption{\small{Average costs versus $w_p$ and $w_f$ in the uniform case at $M=3$, \textcolor{black}{$|\mathcal{C}|=2$, $K=2$, and $\alpha=0.75$. }}}\label{fig:uniform_wpwf}
%\vspace{-0.2cm}
\end{figure*}

%for nonuniform case
\begin{figure*}[!t]
\begin{minipage}[t]{0.245\linewidth}
\centering
%%\vspace{-0.1cm}
        % \includegraphics[height=4.3cm, width=4.4cm]{nonuniform_vs_wpwf_system_cost_new.eps}
        \includegraphics[scale=.3]{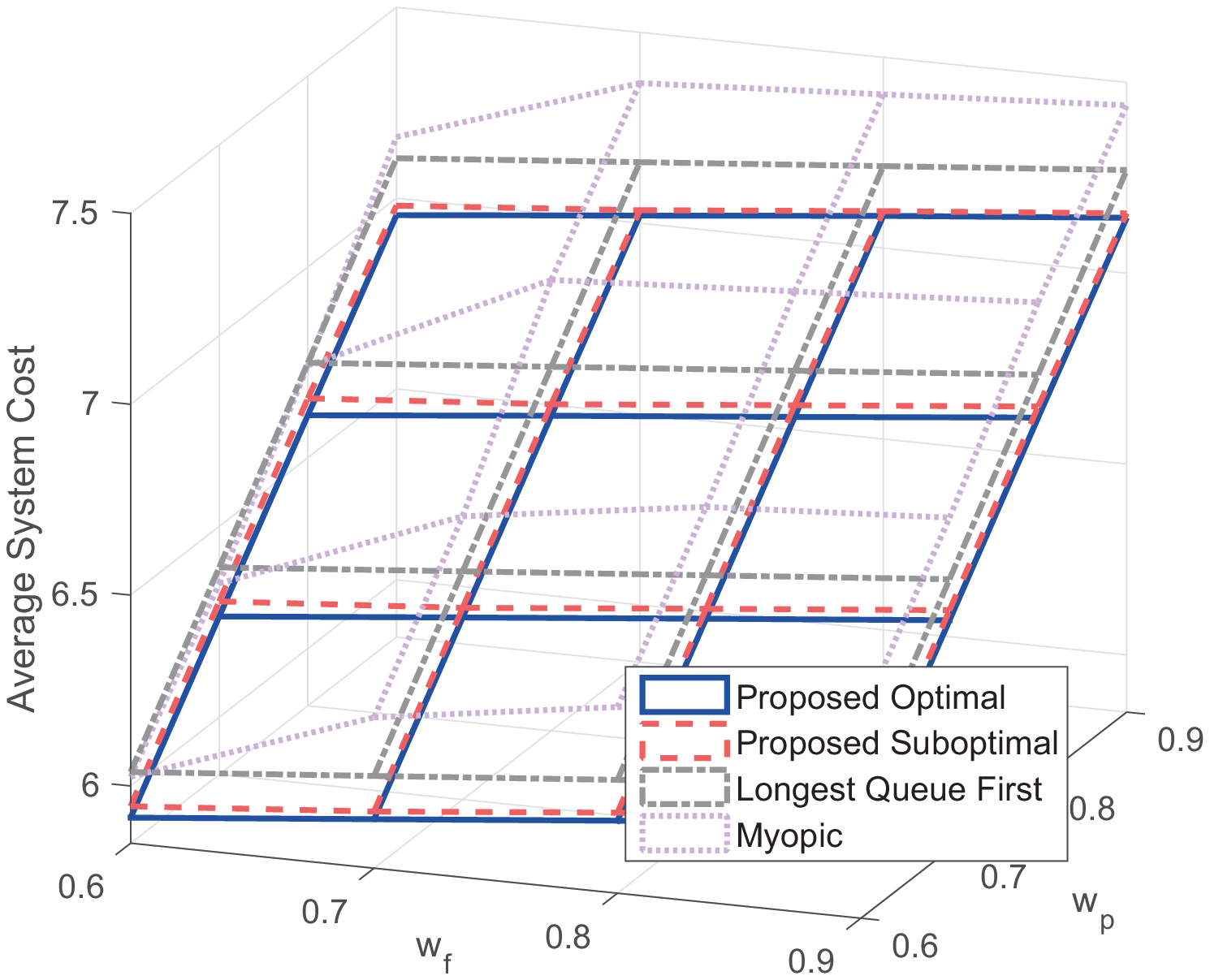}
%\includegraphics[scale=.3]{asymStruc3D.eps}
%        %\vspace{-0.1cm}
\subcaption{\small{Average system cost}}\label{fig:nonuniform_wpwf_system}
%%\vspace{-0.1cm}
\end{minipage}%
\begin{minipage}[t]{.245\linewidth}
\centering
%%\vspace{-0.1cm}
% \includegraphics[height=4.3cm, width=4.4cm]{nonuniform_vs_wpwf_delay_cost_new.eps}
        \includegraphics[scale=.3]{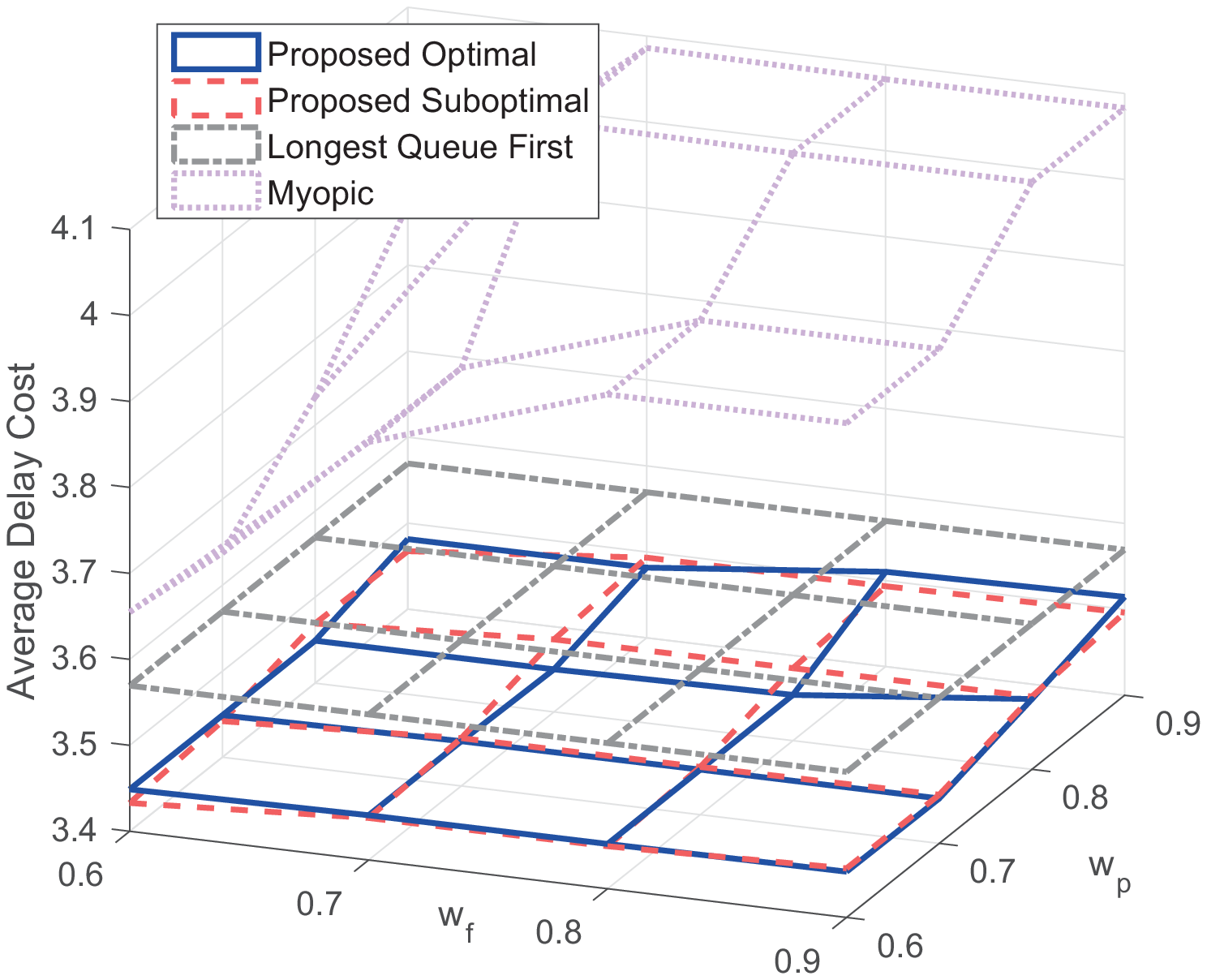}
%        %\vspace{-0.1cm}
\subcaption{\small{Average delay cost}.}\label{fig:nonuniform_wpwf_delay}
%%\vspace{-0.1cm}
\end{minipage}
\begin{minipage}[t]{.245\linewidth}
\centering
%%\vspace{-0.1cm}
% \includegraphics[height=4.3cm, width=4.4cm]{nonuniform_vs_wpwf_fetch_cost_new.eps}
        \includegraphics[scale=.3]{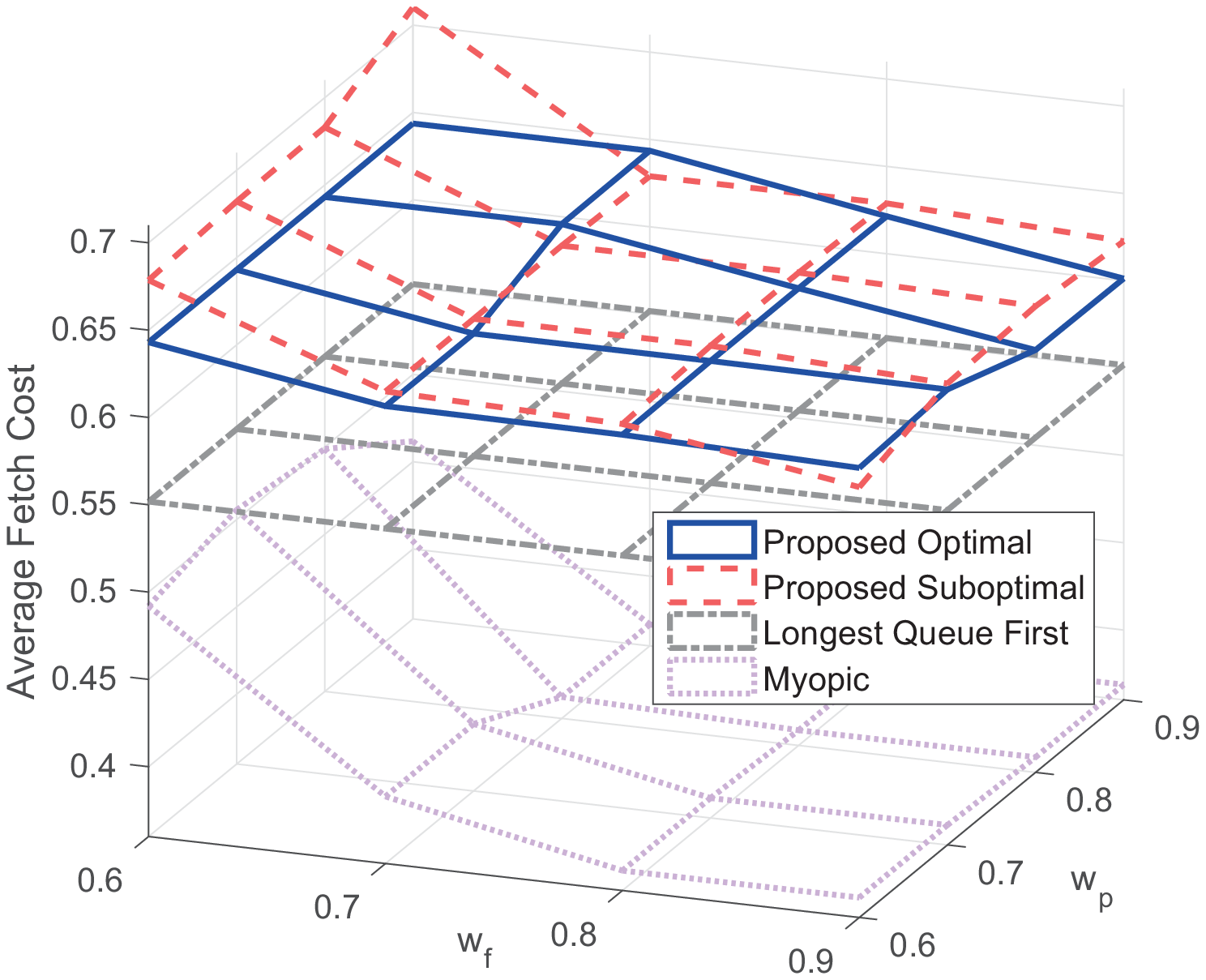}%        %\vspace{-0.1cm}
\subcaption{\small{Average fetching cost}.}\label{fig:nonuniform_wpwf_fetch}
%%\vspace{-0.1cm}
\end{minipage}
\begin{minipage}[t]{.245\linewidth}
\centering
%%\vspace{-0.1cm}
% \includegraphics[height=4.3cm, width=4.4cm]{nonuniform_vs_wpwf_power_cost_new.eps}
        \includegraphics[scale=.3]{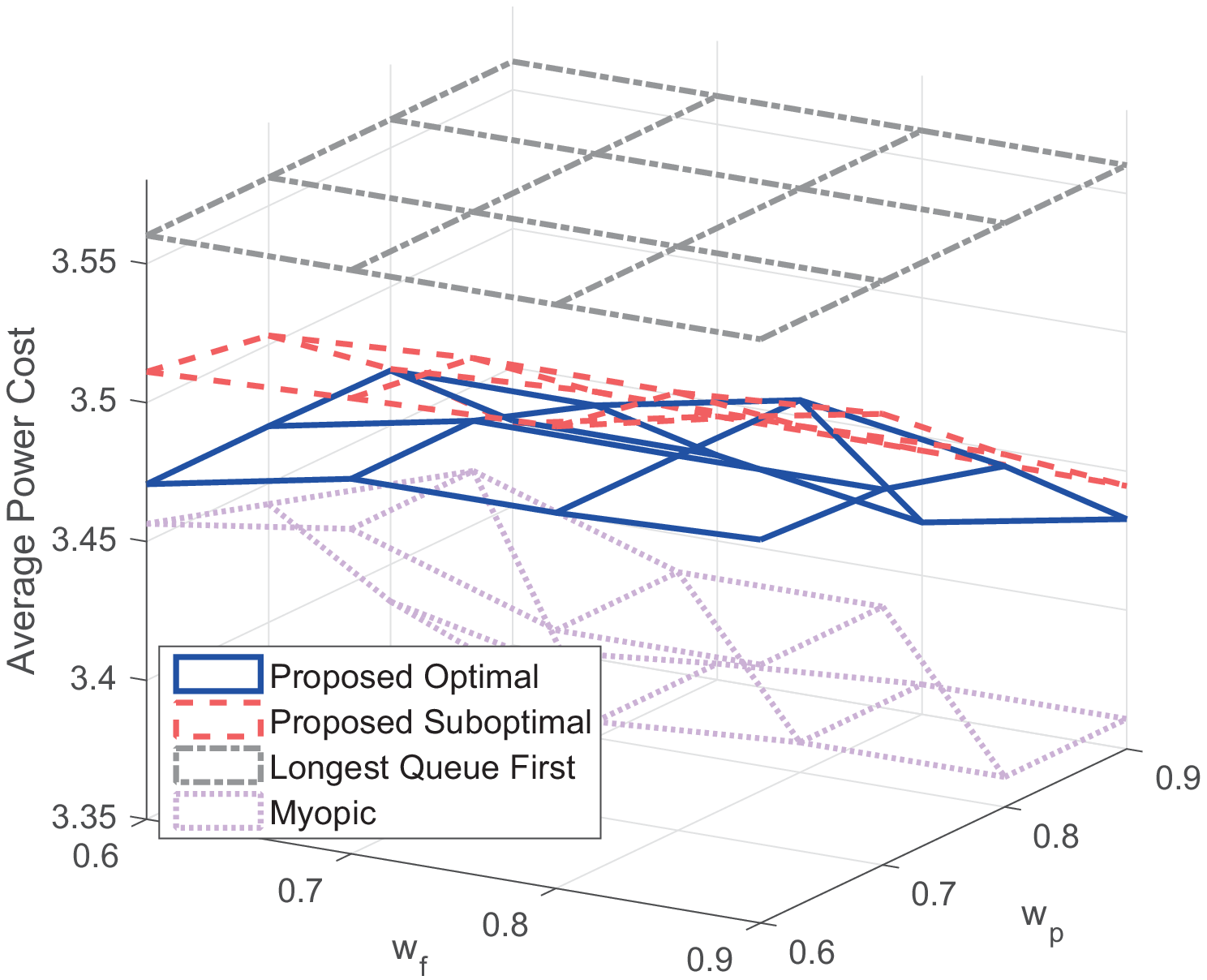}%%\vspace{-0.1cm}
%        %\vspace{-0.1cm}
\subcaption{\small{Average power cost}.}\label{fig:nonuniform_wpwf_power}
\end{minipage}
\caption{\small{Average costs versus $w_p$ and $w_f$ in the nonuniform case at $M=3$,  \textcolor{black}{$|\mathcal{C}|=2$, $K=2$, and $\alpha=0.75$. }}}\label{fig:nonuniform_wpwf}
%\vspace{-0.3cm}
\end{figure*}
Now we show that SSA has a significantly lower computational complexity than the two optimal algorithms proposed in Section VI, i.e., SRVIA and SPIA.
First, we compare the complexity of SSA and SRVIA. If we compute $\{\hat{V}_m(Q_m)\}$ in Step~\ref{code:hatmu} of SSA using the relative value iteration method, then the numbers of iterations required for Step~\ref{code:hatmu} of SSA and SRVIA are comparable.
However, as illustrated in Remark~\ref{remark:sub}, for each iteration, the number of value functions required to be updated in Step~\ref{code:hatmu} of SSA is much smaller than that in SRVIA. In addition, the number of optimizations required to be solved in Step~\ref{code:hatmu} of SSA is comparable to that in each iteration of SRVIA.
Thus, SSA has a much lower computational complexity than SRVIA.
Next, we compare the complexity of SSA and SPIA. SSA is similar to one iteration of SPIA.
As illustrated in Remark~\ref{remark:sub}, the number of value functions required to be updated in Step~\ref{code:hatmu} of SSA is much smaller than that in each iteraion of the policy evaluation step of SRVIA. In addition, the number of optimizations required to be solved in Step~\ref{code:hatmu} of SSA is comparable to that in each iteration of the structured policy update step of SPIA.
Thus, SSA has a much lower computational complexity than SPIA.

\section{Numerical results and discussion}\label{sec:simulation}
In this section, we evaluate the performance of the proposed optimal and suboptimal solutions via numerical examples.
In the simulations, we consider that in each slot, each user requests one content, which is content $m$ with probability $P_m$.
We assume that $\{P_m\}$ follows a (normalized) Zipf distribution with parameter $\alpha$\cite{zipf}.
\textcolor{black}{We consider that each content is of the same size, and the BS stores the most popular contents.
We set $c(m)=3$ for all $m$. In addition, for all $m$, in the uniform case, we set $p(m,k)=2$  for all $k$, and in the nonuniform case, we set $p(m,k)=2$ for $k=1,\cdots,K/2$ and  $p(m,2)=4$ for $k=K/2+1,\cdots,K$.
}

First, we compare the average costs of the proposed optimal and suboptimal policies with \textcolor{black}{three} baseline policies, i.e., a randomized base policy in Definition~\ref{definition:definition2}, the longest-queue-first policy in \cite{batch}\textcolor{black}{, and a myopic policy \cite{powell2007approximate}}. In particular, in each slot, the randomized base policy chooses one content randomly for multicasting according to the distribution $\{P_m\}$ on $\mathcal{M}$ and the longest-queue-first policy schedules the content with the longest request queue for multicasting.
\begin{answer}
In each slot, the myopic policy chooses the multicast scheduling action that minimizes a cost function $C(\mathbf{Q},u) $, i.e.,  $u(t)=\arg\min_{u\in\mathcal{M}}C(\mathbf{Q},u)$, where $$C(\mathbf{Q},u) \triangleq w_ff(u) + w_p p(\mathbf{Q},u) - d(\mathbf{Q},u), \forall\mathbf{Q}\in\bm{\mathcal{Q}}, u\in\mathcal{M},$$
 with $d(\mathbf{Q},u)\triangleq Q_u$ in the uniform case and $d(\mathbf{Q},u)\triangleq\sum_k Q_{u,k}$ in the nonuniform case.
% For the myopic policy, in each slot, we choose the multicast scheduling action that minimizes the cost function, i.e., $u(t)=\arg\min_{u\in\mathcal{M}}C(\mathbf{Q},u)$.
This policy determines the scheduling action myopically, without accurately considering the impact of the action on the future costs.
Note that, this myopic policy can also be treated as an approximate solution to the considered MDP through approximating $V(\mathbf{Q})$  in the Bellman equation with $\sum_mQ_m$ (uniform case) or $\sum_{m,k}Q_{m,k}$ (nonuniform case).
\end{answer}

Fig.~\ref{fig:uniform_wpwf} and Fig.~\ref{fig:nonuniform_wpwf} illustrate the average system, delay, power and fetching costs versus the weights of the power and fetching costs (i.e., $w_p$ and $w_f$) in the uniform and nonuniform cases, respectively.
It can be seen that the average system costs of the proposed optimal and suboptimal policies are very close to each other and are lower than
\textcolor{black}{those of the longest-queue-first policy and the myopic policy}. \textcolor{black}{The reason is that the proposed two policies can make foresighted decisions by better utilizing system state information and balancing the current cost and the futures costs}.  Moreover, we can observe that for the optimal and suboptimal policies, in the uniform case, the average delay cost increases with $w_f$ and does not change with $w_p$, and the average fetching cost decreases with $w_f$; in the nonuniform case, the average delay cost increases with $w_f$ and $w_p$, the average power cost decreases with $w_p$, and the average fetching cost decreases with $w_f$.  This reveals the tradeoff among the delay, power, and fetching costs of the optimal and suboptimal policies.

% \begin{answer}
% In Fig.~\ref{fig:cost_vs_zipf}, Fig.~\ref{fig:cost_vs_user_uniform} and Fig.~\ref{fig:cost_vs_user_nonuniform}, we investigate the impacts of the Zipf parameters and the number of users on the performance of the proposed suboptimal policy and the three baseline polices.
% \end{answer}
Fig.~\ref{fig:cost_vs_zipf} illustrates the average system cost versus the Zipf parameter $\alpha$ in the uniform and nonuniform cases.
The $\alpha$ parameter determines the ``peakiness'' of the content popularity distribution, i.e., a large $\alpha$ indicates that a small amount of contents account for the majority of content requests.
It can be seen that with the increase of $\alpha$, the average system cost of the proposed suboptimal policy decreases and the performance gains over the \textcolor{black}{three} baseline policies increase.
\textcolor{black}{This indicates that the proposed suboptimal policy can utilize caching more effectively as the content popularity distribution gets steeper.}

\begin{answer}
Fig.~\ref{fig:cost_vs_user_uniform} and Fig.~\ref{fig:cost_vs_user_nonuniform} illustrate the average system cost and the average system cost per user versus the number of users $K$ in the uniform and nonuniform cases, respectively.
We can observe that when the average request arrival rate increases (as $K$ increases), the average system costs per user of all policies decrease.
This reveals the benefit of the multicast transmission.
\end{answer}
\begin{figure}[t]
\begin{minipage}[h]{.5\linewidth}
\centering
%%\vspace{-0.1cm}
        % \includegraphics[height=4.4cm, width=4.3cm]{uniform_vs_zipf_system.eps}
        \includegraphics[scale=.33]{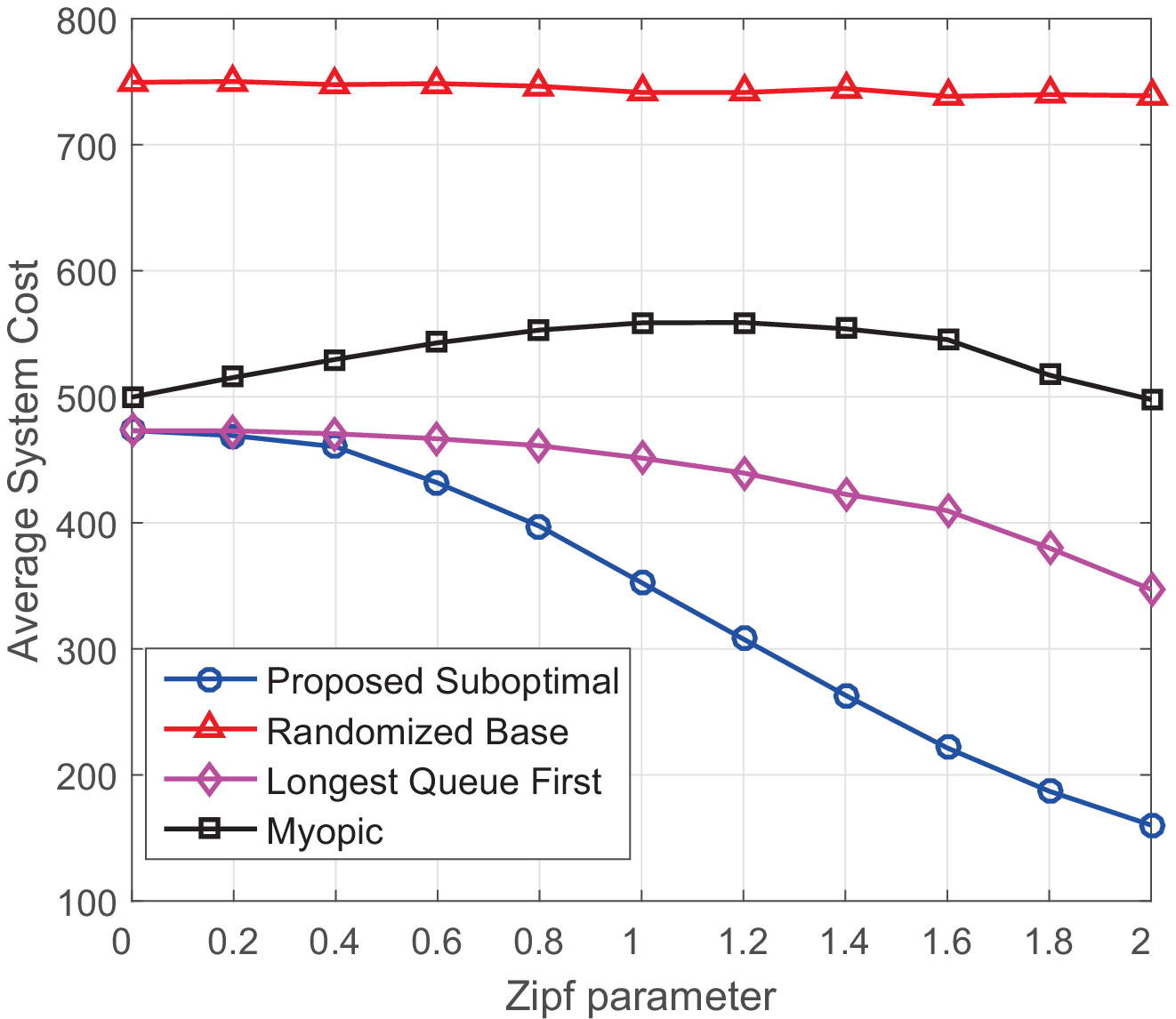}
        %\includegraphics[scale=.33]{3q.eps}
        %\vspace{-0.1cm}
\subcaption{\small{Uniform case.}}\label{fig:uniform_vs_zipf_system}
%\vspace{-0.1cm}
\end{minipage}%
\begin{minipage}[h]{.5\linewidth}
\centering
%%\vspace{-0.1cm}
% \includegraphics[height=4.4cm, width=4.3cm]{uniform_vs_zipf_system.eps}
        \includegraphics[scale=.33]{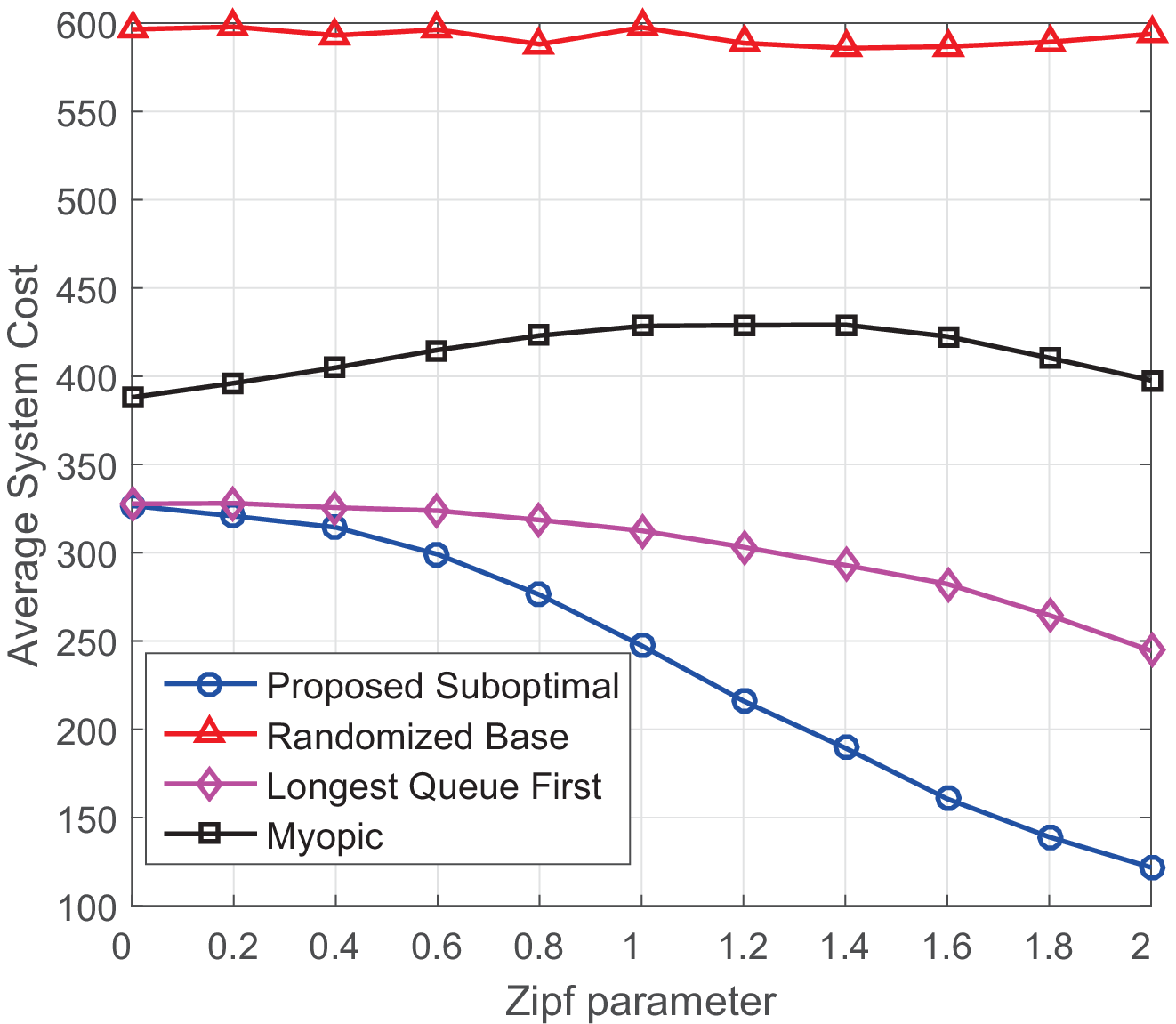}
        %\includegraphics[scale=0.33]{2q.eps}
        %\vspace{-0.1cm}
\subcaption{\small{Nonuniform case.}}\label{fig:nonuniform_vs_zipf_system}
%%\vspace{-0.1cm}
\end{minipage}
%\vspace{-0.05cm}
\caption{\small{Average system cost versus the Zipf parameter $\alpha$ in the uniform and nonuniform cases at \textcolor{black}{$M=30$, $|\mathcal{C}|=13$, $K=30$, and $w_p=w_f=5$.}}} \label{fig:cost_vs_zipf}
%\vspace{-0.2cm}
\end{figure}

\begin{figure}[t]
\begin{minipage}[h]{.5\linewidth}
\centering
%%\vspace{-0.1cm}
        % \includegraphics[height=4.4cm, width=4.3cm]{uniform_vs_zipf_system.eps}
        \includegraphics[scale=.33]{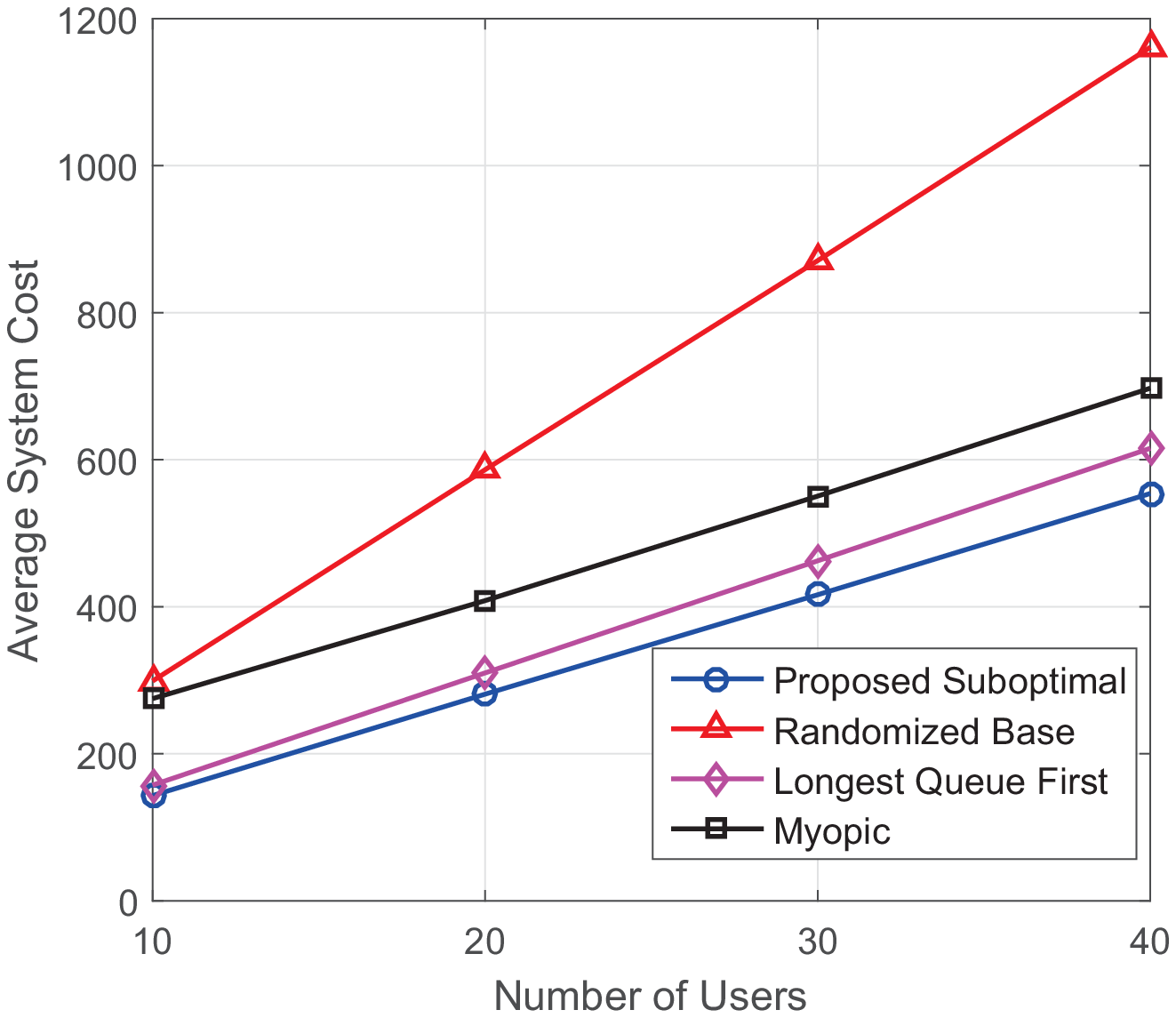}
        %\includegraphics[scale=.33]{3q.eps}
        %\vspace{-0.1cm}
\subcaption{\small{Average system cost.}}\label{fig:uniform_vs_user}
%\vspace{-0.1cm}
\end{minipage}%
\begin{minipage}[h]{.5\linewidth}
\centering
%%\vspace{-0.1cm}
% \includegraphics[height=4.4cm, width=4.3cm]{uniform_vs_zipf_system.eps}
        \includegraphics[scale=.33]{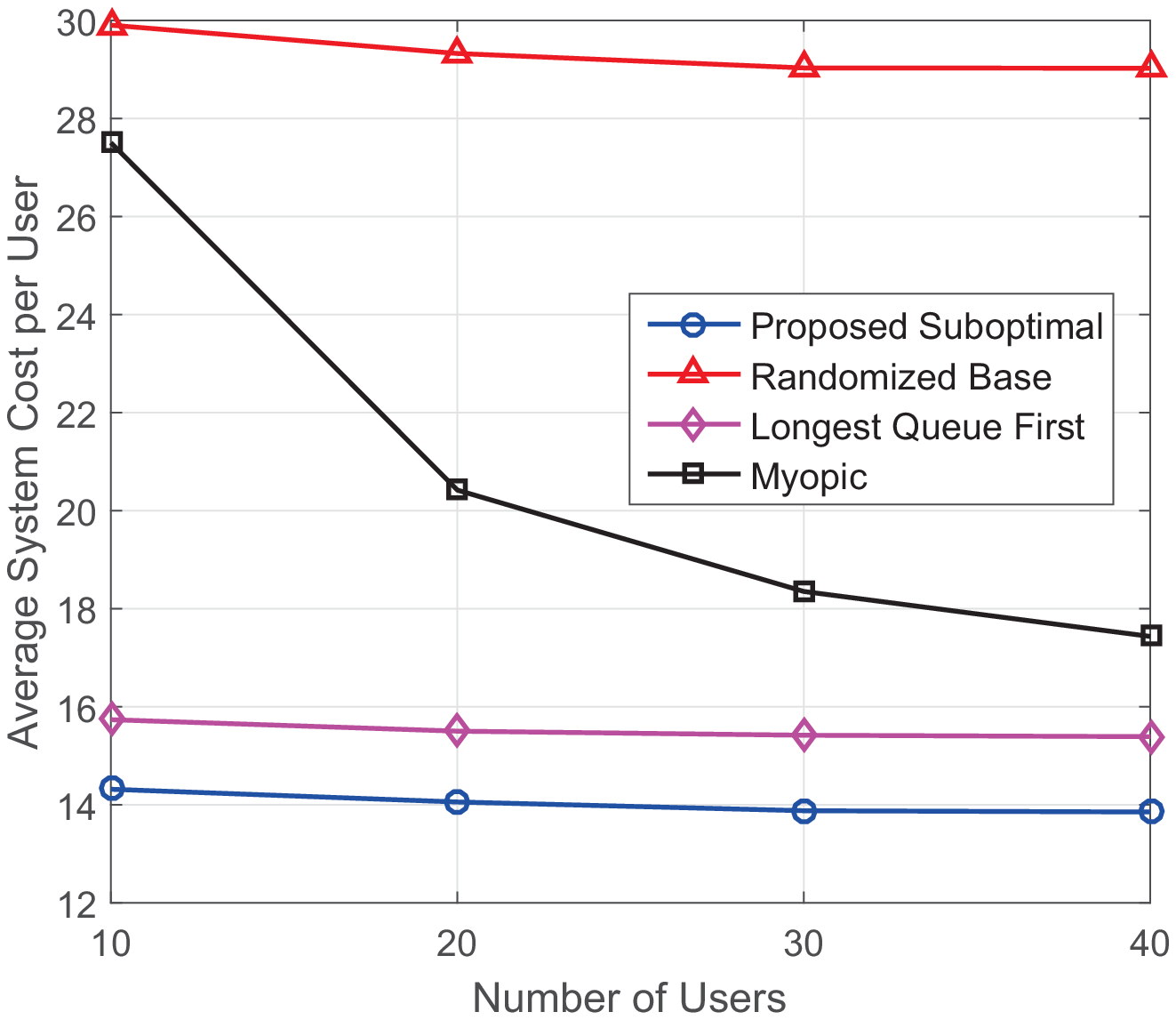}
        %\includegraphics[scale=0.33]{2q.eps}
        %\vspace{-0.1cm}
\subcaption{\small{Average system cost per user.}}\label{fig:uniform_vs_user_per}
%%\vspace{-0.1cm}
\end{minipage}
%\vspace{-0.05cm}
\caption{\small{Average system cost and system cost per user versus number of users $K$ in the uniform cases at \textcolor{black}{$M=30$, $|\mathcal{C}|=13$, $K=30$, $w_p=w_f=5$ and $\alpha=0.75$.}}} \label{fig:cost_vs_user_uniform}
%\vspace{-0.2cm}
\end{figure}

\begin{figure}[t]
\begin{minipage}[h]{.5\linewidth}
\centering
%%\vspace{-0.1cm}
        % \includegraphics[height=4.4cm, width=4.3cm]{uniform_vs_zipf_system.eps}
        \includegraphics[scale=.33]{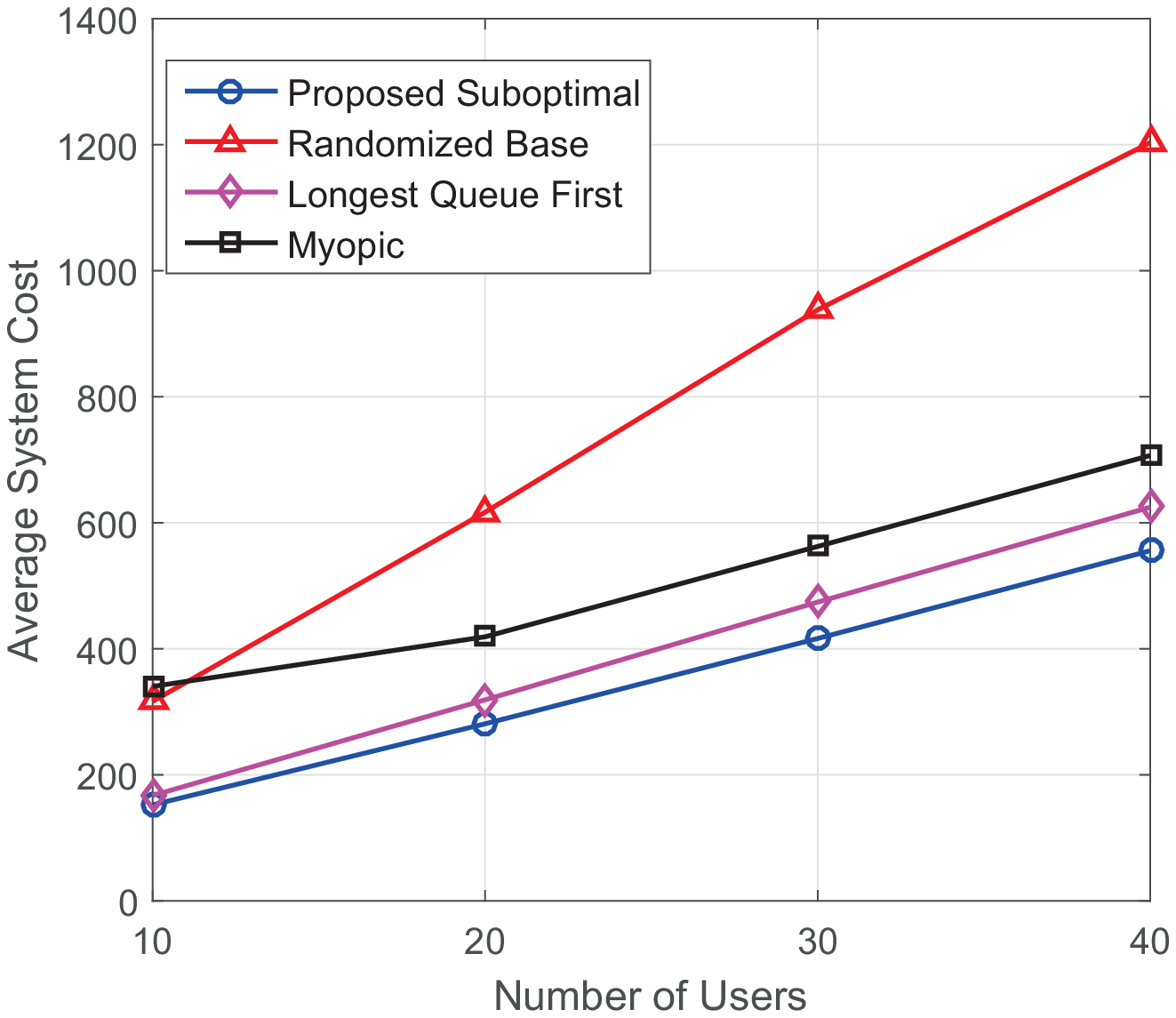}
        %\includegraphics[scale=.33]{3q.eps}
        %\vspace{-0.1cm}
\subcaption{\small{Average system cost.}}\label{fig:nonuniform_vs_user}
%\vspace{-0.1cm}
\end{minipage}%
\begin{minipage}[h]{.5\linewidth}
\centering
%%\vspace{-0.1cm}
% \includegraphics[height=4.4cm, width=4.3cm]{uniform_vs_zipf_system.eps}
        \includegraphics[scale=.33]{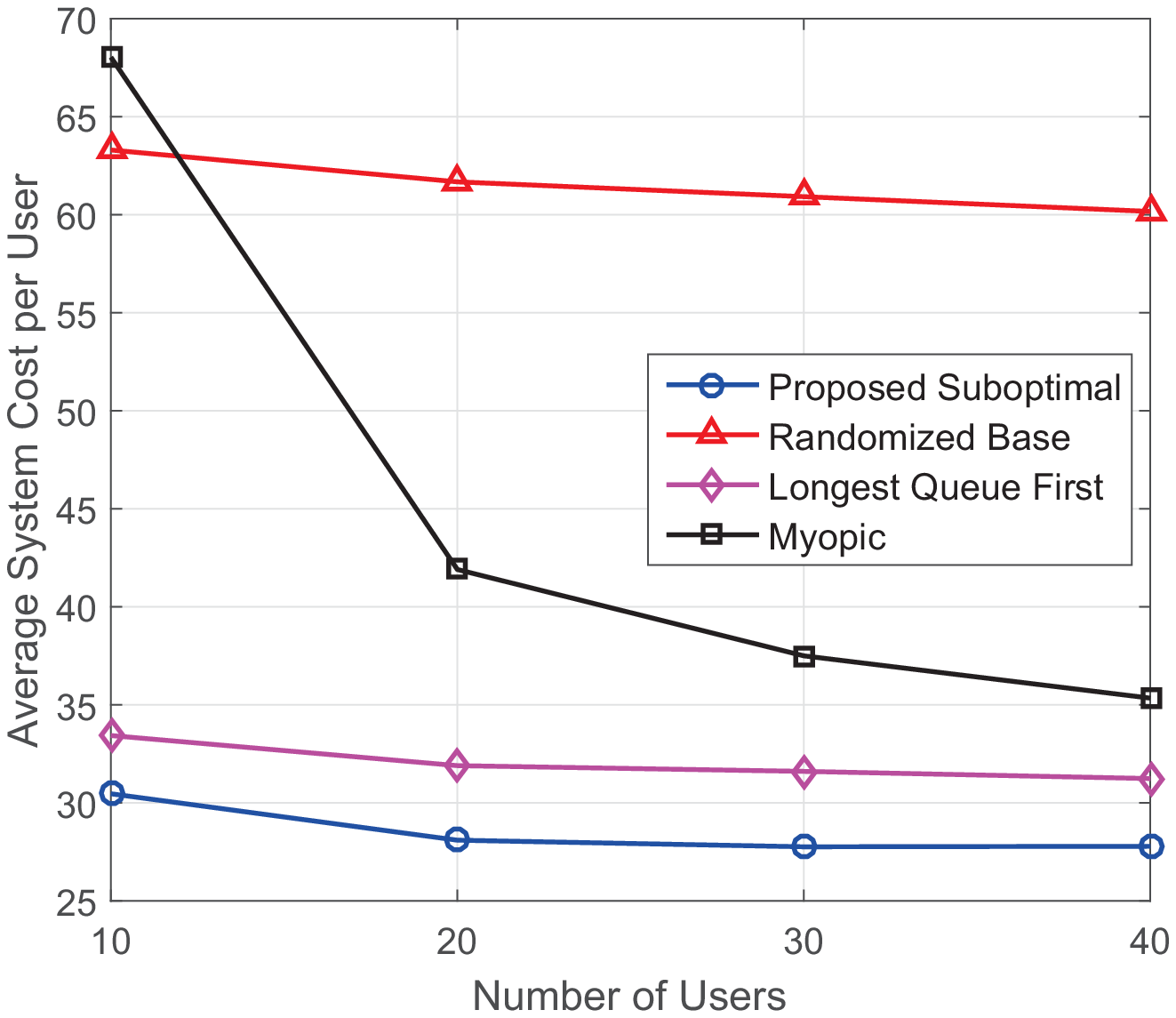}
        %\includegraphics[scale=0.33]{2q.eps}
        %\vspace{-0.1cm}
\subcaption{\small{Average system cost per user.}}\label{fig:nonuniform_vs_user_per}
%%\vspace{-0.1cm}
\end{minipage}
%\vspace{-0.05cm}
\caption{\small{Average system cost and system cost per user versus number of users $K$ in the nonuniform cases at \textcolor{black}{$M=30$, $|\mathcal{C}|=13$, $K=30$, $w_p=w_f=5$ and $\alpha=0.75$.}}} \label{fig:cost_vs_user_nonuniform}
%\vspace{-0.2cm}
\end{figure}

Next, we compare the computational complexity of the two standard optimal algorithms (RVIA and PIA), the two proposed low-complexity optimal algorithms (SRVIA and SPIA), and the proposed low-complexity suboptimal algorithm (SSA) in Table~\ref{table:complexity} for the uniform and nonuniform cases.
It can be seen that SRVIA and SPIA have much lower computational complexity than RVIA and PIA, respectively (with reductions of over 25\% in computation time).
\textcolor{black}
{Note that, the computation times of the four optimal algorithms and the computational reductions of the proposed SRVIA and SPIA have the same order of growth.
Therefore, although the proposed low-complexity optimal algorithms suffer from the curse of dimensionality, their computational complexity reductions are remarkable.}
%With the increase of the state space size ($|\bm{\mathcal{Q}}|=\Pi_m{N_m}$ in the uniform case and $|\bm{\mathcal{Q}}|=\Pi_m\Pi_k (N_{m,k}+1)$ in the nonuniform case),
Moreover, we can observe that SSA has a significantly lower computational complexity than all the \textcolor{black}{optimal} algorithms.
These verify the discussions in Sections VI and VII.

\begin{table}[!htbp]
\color{black}
\centering
% \begin{adjustbox}{max width=0.49\textwidth}
\begin{tabular}{|c|c|c|c|c|c|c|}
\hline
 & $M$ & RVIA & SRVIA & PIA & SPIA & SSA \\ \hline
\multirow{3}{*}{Uniform} & 2   &  0.0138 & 0.0085  &  0.0119 & 0.0077 & 0.0016\\ \cline{2-7}
                     & 3  &  0.577  & 0.428  &  0.614 & 0.476 & 0.0063\\ \cline{2-7}
                     & 4  &  23.42  & 17.31  &  20.48 & 16.33 & 0.0828\\ \hline
 \multirow{3}{*}{Nonuniform} & 2  &  0.380 & 0.295  &  0.419 & 0.337 & 0.0076\\ \cline{2-7}
                     & 3  &  15.55  &  12.37	& 18.45	 &  14.41	 & 0.305  \\ \cline{2-7}
                     & 4  &  2315.8 &  1783.2   & 2473.5 &	1878.8	& 32.57 \\ \hline
\end{tabular}
\caption{\textcolor{black}{\small{Average Matlab computation time (sec) for different algorithms in the uniform and nonuniform cases.
$|\mathcal{C}|=1$,  $w_p=w_f=1$,  $K=2$, $\alpha=0.75$, $N_m=10$ for all $m$ and $N_{m,k}=4$ for all $m,k$.}}}
\label{table:complexity}
% \end{adjustbox}

\end{table}

\begin{answer}
\section{Extension to Markov-Modulated Request Arrivals}\label{sec:markov}
In this section, we extent the structural analysis for the i.i.d. request arrivals to the Markov-modulated request arrivals.
Specifically, we assume that for each $m$ and $k$, the request arrival $\{A_{m,k}(t)\}$ evolves according to an ergodic finite-state  Markov chain with the transition probability $\Pr[A_{m,k}(t+1)|A_{m,k}(t)]$.
In this case, the system state consists of the request queue state $\mathbf{Q}$ and the request arrival state $\mathbf{A}$.
% The multicast scheduling action is made based on the observation of the request state $\mathbf{Q}(t)$ and the request arrival state $\mathbf{A}(t)$.
We define the stationary multicast scheduling policy $\tilde{\mu}$ as a mapping from system state space $\bm{\mathcal{Q}}\times\bm{\mathcal{A}}$ to the multicast scheduling action space $\mathcal{M}$ and formulate the corresponding system cost minimization problem. Similar to Lemma~\ref{lemma:bellman}, we have the following Bellman equation:
\begin{equation*}
  % \noindent
  % $$
  \tilde{\theta}+\tilde{V}(\mathbf{Q},\mathbf{A})=\min_{u\in\mathcal{M}}\left\{g(\mathbf{Q},u)+\mathbb{E}\left[\tilde{V}(\mathbf{Q}',\mathbf{A}')\right]\right\},~\forall \mathbf{Q},\mathbf{A},% \label{eqn:bellmanMarkov}
   % $$
  \end{equation*}

 \noindent where the expectation is taken over the distributions of $\mathbf{A}$ and $\mathbf{A}'$, and $\mathbf{Q}'$ is defined in Lemma~\ref{lemma:bellman}.
 Following the analysis in Sections IV and V, we can show that the optimal policy $\tilde{\mu}^*$ for the Markov-modulated request arrival model possesses  similar structural properties to the optimal policy $\mu^*$ for the i.i.d. request arrival model.
 \begin{theorem}[Structural Properties of Optimal Policy $\tilde{\mu}^*$]
For Markov-modulated request arrivals, the structural properties of the optimal policy $\tilde{\mu}^*$ are as follows.
%\begin{enumerate}
  %\item

1) In the uniform case, $\tilde{\mu}^*$ has a switch structure, i.e., for all $m\in\mathcal{M}$, we have
\begin{equation}\label{eqn:markov_uniform}
  \tilde{\mu}^*(\mathbf{Q},\mathbf{A})=u, \text{if}~Q_u\geq \tilde{s}_u(\mathbf{Q}_{-u},\mathbf{A}).
\end{equation}
% where the switch curve for content $u$ is given by
% \begin{equation*}%\label{eqn:subcurve}
% \tilde{s}_u(\mathbf{Q}_{-u},\mathbf{A})\triangleq\begin{cases}\min\tilde{\mathcal{S}}_u(\mathbf{Q}_{-u},\mathbf{A}),  & \text{if}~\tilde{\mathcal{S}}_u(\mathbf{Q}_{-u},\mathbf{A})\neq\emptyset \\
%             \infty,  &\text{otherwise}
%   \end{cases}
%   \end{equation*}
% with $\tilde{\mathcal{S}}_u(\mathbf{Q}_{-u},\mathbf{A})\triangleq\{Q_u| \hat{J}(\mathbf{Q},\mathbf{A},u)\leq \hat{J}(\mathbf{Q},\mathbf{A},v)~\forall v, v\neq u\}$. Here, $\mathbf{Q}_{-u}$ is defined in Theorem~\ref{theorem:theorem1}.
% %  \item

2) In the nonuniform case, $\tilde{\mu}^*$ has a partial switch structure, i.e., for all $u\in\mathcal{M}$ and $k\in\mathcal{K}$, we have
\begin{equation}\label{eqn:markov_non}
~~~~~~~~ \tilde{\mu}^*(\mathbf{Q},\mathbf{A})=u,~ \hbox{\parbox[t]{.25\textwidth}{if $Q_{u,k}\geq \tilde{s}_{u,k}(\mathbf{Q}_{-u,-k},\mathbf{A})$ and condition (a) or (b) holds,}}
\end{equation}
where condition (a) is $k<k^\dag(k,Q_u)$, condition (b) is $k> k^\dag(k,Q_u)$ and $\tilde{s}_{u,k}(\mathbf{Q}_{-u,-k},\mathbf{A})>0$.
% Moreover, $\tilde{s}_{u,k}(\mathbf{Q}_{-u,-k},\mathbf{A})$ is non-increasing with $Q_{u,i}$ for all $i<k$.

\noindent The switch curves $\tilde{s}_u(\mathbf{Q}_{-u},\mathbf{A})$ and $\tilde{s}_{u,k}(\mathbf{Q}_{-u,-k},\mathbf{A})$ in \eqref{eqn:markov_uniform} and \eqref{eqn:markov_non}  are defined in a similar manner to the switch curves in Theorems~\ref{theorem:theorem1} and \ref{theorem:theorem2}, respectively.
% and the switch curve for content-user pair $(u,k)$ is given by
% \begin{equation*}
% \tilde{s}_{u,k}(\mathbf{Q}_{-u,-k})\triangleq\begin{cases}\min\tilde{\mathcal{S}}_{u,k}(\mathbf{Q}_{-u,-k}),  & \text{if}~\hat{\mathcal{S}}_{u,k}(\mathbf{Q}_{-u,-k})\neq\emptyset \\
%             \infty,  &\text{otherwise}
%   \end{cases}
%   \end{equation*}
% with $\tilde{\mathcal{S}}_{u,k}(\mathbf{Q}_{-u,-k})\triangleq\{Q_{u,k}| \hat{J}(\mathbf{Q},u)\leq \tilde{J}(\mathbf{Q},v)~\forall v, v\neq u\}$.

% Here, $k^\dag(k,Q_u)$ and $\mathbf{Q}_{-u,-k}$ are defined in Theorem~\ref{theorem:theorem2}.
%\end{enumerate}
\label{theorem:markov}
\end{theorem}

Similarly, Theorem~\ref{theorem:markov} implies the following results.

 1) In the uniform case, for all $\mathbf{Q}$ and $\mathbf{A}$, we have
\begin{equation}\label{eqn:indicate_uni_markov}
  \tilde{\mu}^*(\mathbf{Q},\mathbf{A})=u~\Rightarrow~\tilde{\mu}^*(\mathbf{Q}+\mathbf{e}_u,\mathbf{A})=u.
\end{equation}

2) In the nonuniform case, for all $\mathbf{Q}$ and $\mathbf{A}$, we have
\begin{align}\label{eqn:indicate_nonuni_markov}
  \tilde{\mu}^*(\mathbf{Q},\mathbf{A})=u,~\mathbf{Q}+\mathbf{E}_{u,k}\trianglerighteq\mathbf{Q}~\Rightarrow~\tilde{\mu}^*(\mathbf{Q}+\mathbf{E}_{u,k},\mathbf{A})=u.
\end{align}
As in Sections VI and VII, the structural properties in \eqref{eqn:indicate_uni_markov} and \eqref{eqn:indicate_nonuni_markov} can also be utilized to design low-complexity
optimal and suboptimal algorithms.
\end{answer}

%\subsection{Computational Complexity}
\section{Conclusion}
In this paper, we consider the optimal dynamic multicast scheduling to jointly minimize the average delay, power, and  fetching costs for cache-enabled content-centric wireless networks.
We formulate this stochastic optimization problem as an infinite horizon average cost MDP.
We show that the optimal policy has a switch structure in the uniform case and a partial switch structure in the nonuniform case. Moreover, in the uniform case with two contents, we show that the switch curve is monotonically non-decreasing.
Based on these structural results, we propose two low-complexity optimal algorithms.
\textcolor{black}{Motivated by the switch structures of the optimal policy,} to further reduce the complexity, we also propose a low-complexity suboptimal policy, which has \textcolor{black}{similar structural properties to} the optimal policy, and develop a low-complexity algorithm to compute this policy.
\textcolor{black}{These analytical results hold for both  i.i.d. request arrival and  Markov-modulated request arrival models.}
%The optimality properties obtained in this paper can provide design insights for multicast scheduling in practical cache-enabled content-centric wireless networks.

\appendices
\section*{Appendix A: Proof of Lemma~\ref{lemma:bellman}}\label{app:bellman}
By Proposition 4.2.5 in \cite{bertsekas}, the Weak Accessibly (WA) condition holds for unichain policies. Thus, by Proposition 4.2.3 and Proposition 4.2.1 in \cite{bertsekas}, the optimal system cost of the MDP in Problem~\ref{problem:originalproblem} is the same for all initial states and the solution $(\theta,V(\mathbf{Q}))$ to the following Bellman equation exists.
\begin{align}\label{eqn:bellman2}
  \theta+V(\mathbf{Q})=\min_{u\in\mathcal{M}}\Bigg\{g(\mathbf{Q},u)+\sum_{\mathbf{Q}'\in\bm{\mathcal{Q}}}\Pr[\mathbf{Q}'|&\mathbf{Q},u]V(\mathbf{Q}')\Bigg\}\nonumber\\
  &\forall \mathbf{Q}\in\bm{\mathcal{Q}}.
\end{align}
The transition probability is given by
\begin{align}\label{eqn:tranb}
  &\Pr[\mathbf{Q}'|\mathbf{Q},u]\\
\triangleq&\Pr[\mathbf{Q}(t+1)=\mathbf{Q}'|\mathbf{Q}(t)=\mathbf{Q},u(t)=u]\nonumber\\
=&\mathbb{E}\left[\Pr\left[\mathbf{Q}(t+1)=\mathbf{Q}'|\mathbf{Q}(t)=\mathbf{Q},u(t)=u,\mathbf{A}(t)=\mathbf{A}\right]\right]\nonumber,
\end{align}
\begin{align*}
\text{where }&\Pr\left[\mathbf{Q}(t+1)=\mathbf{Q}'|\mathbf{Q}(t)=\mathbf{Q},u(t)=u,\mathbf{A}(t)=\mathbf{A}\right]\nonumber\\
&=\left\{		
   \begin{array}{ll}
       1, & \hbox{if $\mathbf{Q}'$ satisfies \eqref{eqn:queue-ho} or \eqref{eqn:queue-he}} \\
       0, & \hbox{otherwise}
   \end{array}
 \right..
\end{align*}
By substituting \eqref{eqn:tranb} into \eqref{eqn:bellman2}, we have \eqref{eqn:bellman},  which completes the proof.

\section*{Appendix B: Proof of Lemma~\ref{lemma:propertyV1}}\label{app:propertyV1}
We prove Lemma \ref{lemma:propertyV1} using RVIA and induction.

First, we introduce RVIA\cite[Chapter 4.3]{bertsekas}. For each state $\mathbf{Q}\in\bm{\mathcal{Q}}$, let $V_n(\mathbf{Q})$ be the value function in the $n$th iteration, where $n=0,1,\cdots$.
Define
\begin{align}
&J_{n+1}(\mathbf{Q},u_n)\triangleq g(\mathbf{Q},u_n)+\mathbb{E}[V_n(\mathbf{Q}')],\label{eqn:jn}
\end{align}
where $g(\mathbf{Q},u_n)=\sum_mQ_m+w_pp(u_n)+w_ff(u_n)$ and $\mathbf{Q}'=(Q'_m)_{m\in\mathcal{M}}$ with $Q'_m=\min\{\mathbf{1}(u_n\neq m)Q_m+A_m,N_m\}$ in the uniform case;
$g(\mathbf{Q},u_n)=\sum_{m,k}Q_{m,k}+w_pp(u_n,k^\ddag(\mathbf{Q},u_n))+w_ff(u_n)$ with $k^\ddag(\mathbf{Q},u_n)=\max\{k|Q_{u_n,k}>0\}$, $\mathbf{Q}'=(Q'_{m,k})_{m\in\mathcal{M},k\in\mathcal{K}}$ and $Q'_{m,k}=\min\{\mathbf{1}(u_n\neq m)Q_{m,k}+A_{m,k},N_{m,k}\}$ in the nonuniform case.

Note that $J_{n+1}(\mathbf{Q},u_n)$ is related to the R.H.S of the Bellman equation in \eqref{eqn:bellman}. We refer to $J_{n+1}(\mathbf{Q},u_n)$ as the state-action cost function in the $n$th iteration.
For each $\mathbf{Q}$, RVIA calculates $V_{n+1}(\mathbf{Q})$ according to
\begin{equation}\label{eqn:RVIA}
  V_{n+1}(\mathbf{Q})=\min_{u_n} J_{n+1}(\mathbf{Q},u_n)-\min_{u_n} J_{n+1}(\mathbf{Q}^\S,u_n),~\forall n
\end{equation}
 where $J_{n+1}(\mathbf{Q},u_n)$ is given by \eqref{eqn:jn} and $\mathbf{Q}^\S\in \bm{\mathcal{Q}}$ is some fixed state. Under any initialization of $V_0(\mathbf{Q})$, the generated sequence $\{V_n(\mathbf{Q})\}$ converges to $V(\mathbf{Q})$\cite[Proposition 4.3.2]{bertsekas}, i.e.,
 \begin{equation}
   \lim_{n\to\infty}V_n(\mathbf{Q})=V(\mathbf{Q}),~\forall \mathbf{Q}\in\bm{\mathcal{Q}},\label{eqn:converge}
 \end{equation}
 where $V(\mathbf{Q})$ satisfies the Bellman equation in \eqref{eqn:bellman}.
Let $\mu^*_n(\mathbf{Q})$ denote the control that attains the minimum of the first term in \eqref{eqn:RVIA} in the $n$th iteration for all $\mathbf{Q}$, i.e.,
\begin{equation}
  \mu^*_n(\mathbf{Q})=\arg\min_{u_n}J_{n+1}(\mathbf{Q},u_n),~~\forall \mathbf{Q}\in\bm{\mathcal{Q}}.\label{eqn:optimaln}
\end{equation}
We refer to $\mu^*_n$ as the optimal policy for the $n$th iteration.

Next, we prove Lemma~\ref{lemma:propertyV1} through mathematical induction using RVIA.
Denote $\mathbf{Q}^1\triangleq(Q_m^1)_{m\in\mathcal{M}}$ and $\mathbf{Q}^2\triangleq(Q_m^2)_{m\in\mathcal{M}}$.
To prove Lemma~\ref{lemma:propertyV1}, it is equivalent to show that for any $\mathbf{Q}^1,\mathbf{Q}^2\in \bm{\mathcal{Q}}$ such that $\mathbf{Q}^2\succeq\mathbf{Q}^1$,
\begin{equation}
V_n(\mathbf{Q}^2)\geq V_n(\mathbf{Q}^1),\label{eqn:vn}
\end{equation}
holds for all $n=0,1,\cdots$.
First, we initialize $V_0(\mathbf{Q})=0$ for all $\mathbf{Q}\in\bm{\mathcal{Q}}$. Thus, we have $V_0(\mathbf{Q}^1)=V_0(\mathbf{Q}^2)=0$, i.e., \eqref{eqn:vn} holds for $n=0$.
Assume that \eqref{eqn:vn} holds for  some \textcolor{black}{$n\geq 0$}. We will prove that \eqref{eqn:vn} also holds for $n+1$. By \eqref{eqn:RVIA}, we have
\begin{align}
&V_{n+1}(\mathbf{Q}^1)=J_{n+1}\left(\mathbf{Q}^1,\mu^*_n(\mathbf{Q}^1)\right)-\min_{u_n} J_{n+1}(\mathbf{Q}^\S,u_n)\nonumber\\
&\overset{(a)}{\leq}J_{n+1}\left(\mathbf{Q}^1,\mu^*_n(\mathbf{Q}^2)\right)-\min_{u_n} J_{n+1}(\mathbf{Q}^\S,u_n)\nonumber\\
&\overset{(b)}{=}\mathbb{E}[V_n(\mathbf{Q}^{1'})]+\sum_mQ_m^1+w_pp(\mu^*_n(\mathbf{Q}^2))+w_ff(\mu^*_n(\mathbf{Q}^2))\nonumber\\&~~~-\min_{u_n} J_{n+1}(\mathbf{Q}^\S,u_n),
\label{eqn:vn1q1}
\end{align}
where  $(a)$ follows from the optimality of $\mu_n^*(\mathbf{Q}^1)$ for $\mathbf{Q}^1$ in the $n$th iteration, $(b)$ directly follows from \eqref{eqn:jn} and $\mathbf{Q}^{1'}=(Q_m^{1'})_{m\in\mathcal{M}}$ with $Q_m^{1'}=\min\{\mathbbm{1}(\mu^*_n(\mathbf{Q}^2)\neq m)Q_m^1+A_m,N_m\}$. By \eqref{eqn:jn} and \eqref{eqn:RVIA}, we also have
\begin{align}
&V_{n+1}(\mathbf{Q}^2)=J_{n+1}\left(\mathbf{Q}^2,\mu^*_n(\mathbf{Q}^2)\right)-\min_{u_n} J_{n+1}(\mathbf{Q}^\S,u_n)\nonumber\\
&=\mathbb{E}[V_n(\mathbf{Q}^{2'})]+\sum_mQ_m^2+w_pp(\mu^*_n(\mathbf{Q}^2))+w_ff(\mu^*_n(\mathbf{Q}^2))\nonumber\\&~~~-\min_{u_n} J_{n+1}(\mathbf{Q}^\S,u_n),
\label{eqn:vn1q2}
\end{align}
where $\mathbf{Q}^{2'}=(Q_m^{2'})_{m\in\mathcal{M}}$ with $Q_m^{2'}=\min\{\mathbbm{1}(\mu^*_n(\mathbf{Q}^2)\neq m)Q_m^2+A_m,N_m\}$.
Then, we compare \eqref{eqn:vn1q1} and \eqref{eqn:vn1q2} term by term.
Due to $\mathbf{Q}^2\succeq\mathbf{Q}^1$, we have $\sum_mQ_m^2\geq \sum_mQ_m^1$ and $\mathbf{Q}^{2'}\succeq\mathbf{Q}^{1'}$, implying that $\mathbf{E}[V_n(\mathbf{Q}^{2'})]\geq \mathbf{E}[V_n(\mathbf{Q}^{1'})]$ by the induction hypothesis.
Thus, we have $V_{n+1}(\mathbf{Q}^2)\geq V_{n+1}(\mathbf{Q}^1)$, i.e., \eqref{eqn:vn} holds for $n+1$. Therefore, by induction, we can show that \eqref{eqn:vn} holds for any $n$.  By taking limits on both sides of \eqref{eqn:vn} and by \eqref{eqn:converge}, we complete the proof of Lemma~\ref{lemma:propertyV1}.

\section*{Appendix C: Proof of Lemma~\ref{lemma:propertyJ1}}\label{app:propertyJ1}
By \eqref{eqn:state_action_func}, we have
\begin{align}\label{eqn:prooflemma3}
 &J(\mathbf{Q},u)-J(\mathbf{Q},v)-J(\mathbf{Q}+\mathbf{e}_u,u)+J(\mathbf{Q}+\mathbf{e}_u,v)\nonumber\\
=&\mathbb{E}[V(\mathbf{Q}^{1'})]+g(\mathbf{Q},u)-\mathbb{E}[V(\mathbf{Q}^{2'})]-g(\mathbf{Q},v)\nonumber\\
-&\mathbb{E}[V(\mathbf{Q}^{3'})]-g(\mathbf{Q}+\mathbf{e}_u,u)+\mathbb{E}[V(\mathbf{Q}^{4'})]+g(\mathbf{Q}+\mathbf{e}_u,v)\nonumber\\
\overset{(c)}{=}&\mathbb{E}[V(\mathbf{Q}^{1'})]-\mathbb{E}[V(\mathbf{Q}^{2'})]-\mathbb{E}[V(\mathbf{Q}^{3'})]+\mathbb{E}[V(\mathbf{Q}^{4'})],
\end{align}
%\begin{eqnarray}
% &&J(\mathbf{Q},u)-J(\mathbf{Q},v)-\left(J(\mathbf{Q}+\mathbf{e}_u,u)-J(\mathbf{Q}+\mathbf{e}_u,v)\right)\nonumber\\
%&&=\mathbb{E}[V(\mathbf{Q}^{1'})]+g(\mathbf{Q},u)-\mathbb{E}[V(\mathbf{Q}^{2'})]-g(\mathbf{Q},v)\nonumber\\
%&&-\left(\mathbb{E}[V(\mathbf{Q}^{3'})]+g(\mathbf{Q}+\mathbf{e}_u,u)-\mathbb{E}[V(\mathbf{Q}^{4'})]-g(\mathbf{Q}+\mathbf{e}_u,v)\right)\nonumber\\
%&&=\mathbb{E}[V(\mathbf{Q}^{1'})]-\mathbb{E}[V(\mathbf{Q}^{3'})]+\mathbb{E}[V(\mathbf{Q}^{4'})]-\mathbb{E}[V(\mathbf{Q}^{2'})],
%\end{eqnarray}
where $\mathbf{Q}=(Q_m)_{m\in\mathcal{M}}$, $\mathbf{Q}^{i'}=(Q^{i'}_m)_{m\in\mathcal{M}}$, $i=1,2,3,4$ with
\begin{subequations}
\begin{align}
  &Q^{1'}_m=\min\{\mathbbm{1}(u\neq m)Q_m+A_m,N_m\},\label{eqn:q1prime}\\
  &Q^{2'}_m=\min\{\mathbbm{1}(v\neq m)Q_m+A_m,N_m\},\label{eqn:q2prime}\\
  &Q^{3'}_m=\left\{
                       \begin{array}{ll}
                         \min\{A_u,N_u\}& \hbox{if $m=u$} \\
                         \min\{Q_m+A_m,N_m\}&\hbox{otherwise}
                       \end{array}
                     \right.,\label{eqn:q3prime}\\
%&\mathbf{Q}^{3'}_m=\begin{cases}
%  \min\{A_u,N_u\}  ~\hbox{if $m=u$} \\
%  \min\{\mathbbm{1}(u\neq m)Q_m+A_m,N_m\}~ \hbox{otherwise}
%\end{cases}
  &Q^{4'}_m=\left\{
                       \begin{array}{ll}
                         \min\{Q_u+1+A_u,N_u\} & \hbox{if $m=u$} \\
                         \min\{\mathbbm{1}(v\neq m)Q_m+A_m,N_m\} & \hbox{otherwise}
                       \end{array}
                     \right.,\label{eqn:q4prime}
  %                    &Q^{4'}_m=\begin{cases}\min\{Q_u+1+A_u,N_u\},  & \text{if}~m=u \\
  %           \min\{\mathbbm{1}(v\neq m)Q_m+A_m,N_m\},  &\text{otherwise}
  % \end{cases}
\end{align}
\end{subequations}
and (c) is due to
%\begin{align}
%  &g(\mathbf{Q},u)-g(\mathbf{Q},v)-g(\mathbf{Q}+\mathbf{e}_u,u)+g(\mathbf{Q}+\mathbf{e}_u,v)\nonumber\\
%  =&\Big(\sum_mQ_m+w_pp(u)+w_ff(u)\Big)-\Big(\sum_mQ_m+w_pp(v)\nonumber\\&+w_ff(v)\Big)
%  -\Big(\sum_mQ_m+1+w_pp(u)+w_ff(u)\Big)\nonumber\\
%  &+\Big(\sum_mQ_m+1+w_pp(v)+w_ff(v)\Big)\nonumber\\
%  =&0.
%\end{align}
% \begin{align}
%   &g(\mathbf{Q},u)-g(\mathbf{Q},v)-g(\mathbf{Q}+\mathbf{e}_u,u)+g(\mathbf{Q}+\mathbf{e}_u,v)\nonumber\\
%   =&\Big(\sum_mQ_m+w_pp(u)+w_ff(u)\Big)\nonumber\\
%   &-\Big(\sum_mQ_m+w_pp(v)+w_ff(v)\Big)\nonumber\\
%   &-\Big(\sum_mQ_m+1+w_pp(u)+w_ff(u)\Big)\nonumber\\
%   &+\Big(\sum_mQ_m+1+w_pp(v)+w_ff(v)\Big)=0
% \end{align}
\begin{align}
  &g(\mathbf{Q},u)-g(\mathbf{Q},v)-g(\mathbf{Q}+\mathbf{e}_u,u)+g(\mathbf{Q}+\mathbf{e}_u,v)\nonumber\\
  =&\Big(\sum_mQ_m+w_pp(u)+w_ff(u)\Big)
  -\Big(\sum_mQ_m+w_pp(v)\nonumber\\&+w_ff(v)\Big)-\Big(\sum_mQ_m+1+w_pp(u)+w_ff(u)\Big)\nonumber\\
  &+\Big(\sum_mQ_m+1+w_pp(v)+w_ff(v)\Big)=0.
\end{align}
To prove Lemma~\ref{lemma:propertyJ1}, it remains to show that the R.H.S. of \eqref{eqn:prooflemma3} is nonnegative.
By comparing \eqref{eqn:q1prime} with \eqref{eqn:q3prime}, we can see that $Q^{1'}_m=Q^{3'}_m$ for all $m$,  i.e., $\mathbf{Q}^{1'}=\mathbf{Q}^{3'}$. Thus, we have $\mathbb{E}[V(\mathbf{Q}^{1'})]=\mathbb{E}[V(\mathbf{Q}^{3'})]$.
By comparing \eqref{eqn:q2prime} with \eqref{eqn:q4prime}, we can see that $Q^{4'}_u\geq Q^{2'}_u$ and $Q^{4'}_m=Q^{2'}_m$ for all $m\neq u$, i.e., $\mathbf{Q}^{4'}\succeq\mathbf{Q}^{2'}$. Thus, by Lemma~\ref{lemma:propertyV1}, we have $\mathbb{E}[V(\mathbf{Q}^{4'})]\geq\mathbb{E}[V(\mathbf{Q}^{2'})]$.
Therefore, by \eqref{eqn:prooflemma3}, we have $J(\mathbf{Q}+\mathbf{e}_u,u)-J(\mathbf{Q}+\mathbf{e}_u,v)\leq J(\mathbf{Q},u)-J(\mathbf{Q},v)$. We complete the proof of Lemma~\ref{lemma:propertyJ1}.
\section*{Appendix D: Proof of Theorem~\ref{theorem:theorem1}}\label{app:theorem1}
Consider content $u\in\mathcal{M}$ and state $\mathbf{Q}=(Q_m)_{m\in\mathcal{M}}$ where $Q_{u}=s_u(\mathbf{Q}_{-u})$.
Note that, if  $s_u(\mathbf{Q}_{-u})=\infty$, \eqref{eqn:switch} always holds.
Therefore, in the following, we only consider that $s_u(\mathbf{Q}_{-u})<\infty$.
According to  the definition of $s_u(\mathbf{Q}_{-u})$ in Theorem~\ref{theorem:theorem1}, we can see that $J(\mathbf{Q},u)\leq J(\mathbf{Q},v)$ for all $v\in\mathcal{M}$ and $v\neq u$. Thus, it is optimal to multicast content $u$ for state $\mathbf{Q}$, i.e., $\mu^*(\mathbf{Q})=u$.
Consider another state $\mathbf{Q}'=(Q_{m}')_{m\in\mathcal{M}}$ where $Q_{u}'\geq Q_{u}$ and $Q_{m}'=Q_{m}$ for all $m\neq u$.
To prove Theorem~\ref{theorem:theorem1}, it is equivalent to show that $\mu^*(\mathbf{Q}')=u$.
By Lemma~\ref{lemma:propertyJ1}, for all $v\in\mathcal{M}$ and $v\neq u$, we have
\begin{equation}
  J(\mathbf{Q}',u)-J(\mathbf{Q}',v)\leq J(\mathbf{Q},u)-J(\mathbf{Q},v)\leq 0.
\end{equation}
Thus, it is optimal to multicast content $u$ for $\mathbf{Q}'$.
We complete the proof of Theorem~\ref{theorem:theorem1}.

\section*{Appendix E: Proof of Lemma~\ref{lemma:propertyofswitch}}\label{app:propertyofswitch}
To prove the monotonically non-decreasing property of $s_2(Q_1)$ with respect to $Q_1$, it is equivalent to show that, if $\mu^*\left(\mathbf{Q}+\mathbf{e}_1\right)=2$, then $\mu^*\left(\mathbf{Q}\right)=2$. This is sufficient to  show
\begin{equation}\label{eqn:prooflemma4}
  J\left(\mathbf{Q},2\right)-J\left(\mathbf{Q},1\right)\leq J\left(\mathbf{Q}+\mathbf{e}_1,2\right)-J\left(\mathbf{Q}+\mathbf{e}_1,1\right),
\end{equation}
where $\mathbf{Q}=(Q_1,Q_2)$ and $\mathbf{e}_1=(1,0)$.

By \eqref{eqn:state_action_func}, we have
\begin{align}\label{eqn:prooflemma42}
  &J\left(\mathbf{Q},2\right)-J\left(\mathbf{Q},1\right)-J\left(\mathbf{Q}+\mathbf{e}_1,2\right)+J\left(\mathbf{Q}+\mathbf{e}_1,1\right)\nonumber\\
=&\mathbb{E}[V(\mathbf{Q}^{1'})]+g(\mathbf{Q},2)-\mathbb{E}[V(\mathbf{Q}^{2'})]-g(\mathbf{Q},1)\nonumber\\
-&\mathbb{E}[V(\mathbf{Q}^{3'})]-g(\mathbf{Q}+\mathbf{e}_1,2)+\mathbb{E}[V(\mathbf{Q}^{4'})]+g(\mathbf{Q}+\mathbf{e}_1,1)\nonumber\\
\overset{(d)}{=}&\mathbb{E}[V(\mathbf{Q}^{1'})]-\mathbb{E}[V(\mathbf{Q}^{2'})]-\mathbb{E}[V(\mathbf{Q}^{3'})]+\mathbb{E}[V(\mathbf{Q}^{4'})],
\end{align}
where
\begin{subequations}
\begin{align}
&\mathbf{Q}^{1'}=(\min\{Q_1+A_1,N_1\},\min\{A_2,N_2\}),\label{eqn:q1prime2}\\
&\mathbf{Q}^{2'}=(\min\{A_1,N_1\},\min\{Q_2+A_2,N_2\}),\label{eqn:q2prime2}\\
&\mathbf{Q}^{3'}=(\min\{Q_1+A_1+1,N_1\},\min\{A_2,N_2\}),\label{eqn:q3prime2}\\
&\mathbf{Q}^{4'}=(\min\{A_1,N_1\},\min\{Q_2+A_2,N_2\}),\label{eqn:q4prime2}
\end{align}
\end{subequations}
and (d) is due to
%\begin{align}
%  &g(\mathbf{Q},2)-g(\mathbf{Q},1)-g(\mathbf{Q}+\mathbf{e}_1,2)+g(\mathbf{Q}+\mathbf{e}_1,1)\nonumber\\
%  =&\big(Q_1+Q_2+w_pp(2)+w_ff(2)\big)-\big(Q_1+Q_2+w_pp(1)\nonumber\\&+w_ff(1)\big)
%  -\big(Q_1+Q_2+1+w_pp(2)+w_ff(2)\big)\nonumber\\
%  &+\big(Q_1+Q_2+1+w_pp(1)+w_ff(1)\big)\nonumber\\
%  =&0.
%\end{align}
% \begin{align}
%   &g(\mathbf{Q},2)-g(\mathbf{Q},1)-g(\mathbf{Q}+\mathbf{e}_1,2)+g(\mathbf{Q}+\mathbf{e}_1,1)\nonumber\\
%   =&\big(Q_1+Q_2+w_pp(2)+w_ff(2)\big)\nonumber\\
%   &-\big(Q_1+Q_2+w_pp(1)+w_ff(1)\big)\nonumber\\
%   &-\big(Q_1+Q_2+1+w_pp(2)+w_ff(2)\big)\nonumber\\
%   &+\big(Q_1+Q_2+1+w_pp(1)+w_ff(1)\big)=0
% \end{align}
\begin{align}
  &g(\mathbf{Q},2)-g(\mathbf{Q},1)-g(\mathbf{Q}+\mathbf{e}_1,2)+g(\mathbf{Q}+\mathbf{e}_1,1)\nonumber\\
  =&\big(Q_1+Q_2+w_pp(2)+w_ff(2)\big)-\big(Q_1+Q_2+w_pp(1)\nonumber\\
  &+w_ff(1)\big)-\big(Q_1+Q_2+1+w_pp(2)+w_ff(2)\big)\nonumber\\
  &+\big(Q_1+Q_2+1+w_pp(1)+w_ff(1)\big)=0.
\end{align}
To prove Lemma~\ref{lemma:propertyofswitch}, it remains to show that the R.H.S. of \eqref{eqn:prooflemma42} is nonpositive.
By comparing \eqref{eqn:q1prime2} with \eqref{eqn:q3prime2}, we have  $\mathbf{Q}^{3'}\succeq\mathbf{Q}^{1'}$, implying that $\mathbb{E}[V(\mathbf{Q}^{3'})]\geq\mathbb{E}[V(\mathbf{Q}^{1'})]$ by Lemma~\ref{lemma:propertyV1}.
By comparing \eqref{eqn:q2prime2} with \eqref{eqn:q3prime2}, we have $\mathbf{Q}^{4'}=\mathbf{Q}^{2'}$, implying that $\mathbb{E}[V(\mathbf{Q}^{4'})]=\mathbb{E}[V(\mathbf{Q}^{2'})]$.
Thus, by \eqref{eqn:prooflemma42}, we can show that \eqref{eqn:prooflemma4} holds.

Similarly, we can show that the following inequality holds:
\begin{equation}\label{eqn:prooflemma4Q2}
  J\left(\mathbf{Q},1\right)-J\left(\mathbf{Q},2\right)\leq J\left(\mathbf{Q}+\mathbf{e}_2,1\right)-J\left(\mathbf{Q}+\mathbf{e}_2,2\right),
\end{equation}
where $\mathbf{Q}=(Q_1,Q_2)$ and $\mathbf{e}_2=(0,1)$.
Thus, if $\mu^*\left(\mathbf{Q}+\mathbf{e}_2\right)=1$, then $\mu^*\left(\mathbf{Q}\right)=1$. This implies the monotonically non-decreasing property of $s_1(Q_2)$ with respect to $Q_2$.
We complete the proof of Lemma~\ref{lemma:propertyofswitch}.

\section*{Appendix F: Proof of Proposition~\ref{proposition:num}}\label{app:num}
% For the uniform case with two contents, the size of the state space is $|\bm{\mathcal{Q}}|=(N_1+1)(N_2+1)$. By Definition~\ref{definition:definition1}, in each $\mathbf{Q}\in\bm{\mathcal{Q}}$, there are two possible actions. Therefore, the number of the policies in Definition~\ref{definition:definition1} is $2^{(N_1+1)(N_2+1)}$.

Let $Z(N_1,N_2)$ denote  the number of the policies with monotonically non-decreasing curves.
By Theorem~\ref{theorem:theorem1}, either  $s_2(Q_1)$ or $s_1(Q_2)$ is sufficient to characterize the optimal policy.
Hence, we have
\begin{align}
  Z(N_1,N_2)&=\sum_{a_{N_1}=0}^{N_2+1}\sum_{a_{N_1-1}=0}^{a_{N_1}}\cdots\sum_{a_1=0}^{a_2}\sum_{a_0=0}^{a_1}1 \label{eqn:num1}\\
&=\sum_{b_{N_2}=0}^{N_1+1}\sum_{b_{N_2-1}=0}^{b_{N_2}}\cdots\sum_{b_1=0}^{b_2}\sum_{b_0=0}^{b_1}1,
\label{eqn:num2}
\end{align}
where \eqref{eqn:num1} and \eqref{eqn:num2} are the number of all possible $s_2(Q_1)$ and $s_1(Q_2)$, respectively.
In the following, we shall show that
\begin{equation}
  Z(N_1,N_2)= {N_1+N_2+2 \choose N_1+1}
\label{eqn:zfunc}
\end{equation}
holds for any positive integers $N_1,N_2$.

We use induction on $n=N_1+N_2\geq 2$.
If $n=2$, then $N_1=N_2=1$  and we have $Z(1,1)=\sum_{a_1=0}^2\sum_{a_0=0}^{a_1}1=6={4\choose 2}$.
Assume \eqref{eqn:zfunc} holds for any positive integers $N_1,N_2$ with $N_1+N_2=n\geq 2$. Now consider $Z(N_1,N_2)$ with $N_1+N_2=n+1$.
If $N_1=1$, then by \eqref{eqn:num1}, we have $Z(1,N_2)=\sum_{a_1=0}^{N_2+1}\sum_{a_0=0}^{a_1}1={N_2+3\choose 2}$. %=\frac{(N_2+2)(N_2+3)}{2}
If $N_2=1$, then by \eqref{eqn:num2}, we have $Z(N_1,1)=\sum_{b_1=0}^{N_1+1}\sum_{b_0=0}^{b_1}1={N_1+3\choose 2}$. %=\frac{(N_1+2)(N_1+3)}{2}
If $N_1,N_2>1$, then by \eqref{eqn:num1} and the induction hypothesis,
we have $Z(N_1,N_2)=Z(N_1-1,N_2)+Z(N_1,N_2-1)={N_1+N_2+1\choose N_1}+{N_1+N_2+1\choose N_1+1}={N_1+N_2+2 \choose N_1+1}$. Thus, \eqref{eqn:zfunc} holds whenever $N_1+N_2=n+1$.
Therefore, by induction, \eqref{eqn:zfunc} holds for any positive integers $N_1,N_2$. We complete the proof of Proposition~\ref{proposition:num}.

\section*{Appendix G: Proof of Lemma~\ref{lemma:propertyV2}}\label{app:propertyV2}
We prove Lemma~\ref{lemma:propertyV2} through mathematical induction using the RVIA in Appendix B.
Denote $\mathbf{Q}^1\triangleq(Q_{m,k}^1)_{m\in\mathcal{M},k\in\mathcal{K}}$ and $\mathbf{Q}^2\triangleq(Q_{m,k}^2)_{m\in\mathcal{M},k\in\mathcal{K}}$.
To prove Lemma~\ref{lemma:propertyV2}, by \eqref{eqn:converge}, it is equivalent to show that for any $\mathbf{Q}^1,\mathbf{Q}^2\in \bm{\mathcal{Q}}$ such that $\mathbf{Q}^2\trianglerighteq\mathbf{Q}^1$,
\begin{equation}
V_n(\mathbf{Q}^2)\geq V_n(\mathbf{Q}^1),\label{eqn:vn2}
\end{equation}
holds for all $n=0,1,\cdots$.
We initialize $V_0(\mathbf{Q})=0$ for all $\mathbf{Q}\in\bm{\mathcal{Q}}$. Thus, we have $V_0(\mathbf{Q}^1)=V_0(\mathbf{Q}^2)=0$, i.e., \eqref{eqn:vn2} holds for $n=0$.
Assume that \eqref{eqn:vn2} holds for some \textcolor{black}{$n\geq 0$}. We will prove that \eqref{eqn:vn2} also holds for $n+1$.
By \eqref{eqn:RVIA}, we have
\begin{align}
&V_{n+1}(\mathbf{Q}^1)%=J_{n+1}\left(\mathbf{Q}^1,\mu^*_n(\mathbf{Q}^1)\right)-\min_{u_n} J_{n+1}(\mathbf{Q}^\S,u_n)\nonumber\\
\overset{(e)}{\leq}J_{n+1}\left(\mathbf{Q}^1,\mu^*_n(\mathbf{Q}^2)\right)-\min_{u_n} J_{n+1}(\mathbf{Q}^\S,u_n)\nonumber\\
&\overset{(f)}{=}\mathbb{E}[V_n(\mathbf{Q}^{1'})]+\sum_{m,k}Q^1_{m,k}+w_pp(\mu^*_n(\mathbf{Q}^2),k^\ddag(\mathbf{Q^1},\mu^*_n(\mathbf{Q}^2)))\nonumber\\&~~~+w_ff(\mu^*_n(\mathbf{Q}^2))-\min_{u_n} J_{n+1}(\mathbf{Q}^\S,u_n),
\label{eqn:vn2q1}
\end{align}
where  $(e)$ follows from the optimality of $\mu_n^*(\mathbf{Q}^1)$ for $\mathbf{Q}^1$ in the $n$th iteration, $(f)$ directly follows from \eqref{eqn:jn}, $k^\ddag(\mathbf{Q^1},\mu^*_n(\mathbf{Q}^2))=\max\left\{k|\mathbf{Q^1_{\mu^*_n(\mathbf{Q}^2),k}>0}\right\}$ and $\mathbf{Q}^{1'}=(Q_{m,k}^{1'})_{m\in\mathcal{M},k\in\mathcal{K}}$ with $Q_{m,k}^{1'}=\min\{\mathbbm{1}(\mu^*_n(\mathbf{Q}^2)\neq m)Q_{m,k}^1+A_{m,k},N_{m,k}\}$.
By \eqref{eqn:jn} and \eqref{eqn:RVIA}, we also have
\begin{align}
&V_{n+1}(\mathbf{Q}^2)=J_{n+1}\left(\mathbf{Q}^2,\mu^*_n(\mathbf{Q}^2)\right)-\min_{u_n} J_{n+1}(\mathbf{Q}^\S,u_n)\nonumber\\
&=\mathbb{E}[V_n(\mathbf{Q}^{2'})]+\sum_{m,k}Q^2_{m,k}+w_pp(\mu^*_n(\mathbf{Q}^2),k^\ddag(\mathbf{Q^2},\mu^*_n(\mathbf{Q}^2)))\nonumber\\&~~~+w_ff(\mu^*_n(\mathbf{Q}^2))-\min_{u_n} J_{n+1}(\mathbf{Q}^\S,u_n),
\label{eqn:vn2q2}
\end{align}
where $k^\ddag(\mathbf{Q^2},\mu^*_n(\mathbf{Q}^2))=\max\left\{k|\mathbf{Q^2_{\mu^*_n(\mathbf{Q}^2),k}>0}\right\}$ and $\mathbf{Q}^{2'}=(Q_{m,k}^{2'})_{m\in\mathcal{M},k\in\mathcal{K}}$ with $Q_{m,k}^{2'}=\min\{\mathbbm{1}(\mu^*_n(\mathbf{Q}^2)\neq m)Q_{m,k}^2+A_{m,k},N_{m,k}\}$.

Next, we compare \eqref{eqn:vn2q1} and \eqref{eqn:vn2q2} term by term.
Due to $\mathbf{Q}^2\trianglerighteq\mathbf{Q}^1$, we have $\mathbf{Q}^{2'}\trianglerighteq\mathbf{Q}^{1'}$. Thus, by the induction hypothesis, we have $\mathbb{E}[V_n(\mathbf{Q}^{2'})]\geq \mathbb{E}[V_n(\mathbf{Q}^{1'})]$. Due to $\mathbf{Q}^2\trianglerighteq\mathbf{Q}^1$, we have $\sum_{m,k}Q^2_{m,k}\geq \sum_{m,k}Q^1_{m,k}$ and $k^\ddag(\mathbf{Q^2},\mu^*_n(\mathbf{Q}^2))=k^\ddag(\mathbf{Q^1},\mu^*_n(\mathbf{Q}^2))$, implying that $p(\mu^*_n(\mathbf{Q}^2),k^\ddag(\mathbf{Q^2},\mu^*_n(\mathbf{Q}^2)))=p(\mu^*_n(\mathbf{Q}^2),k^\ddag(\mathbf{Q^1},\mu^*_n(\mathbf{Q}^2)))$.
Thus, we have $V_{n+1}(\mathbf{Q}^2)\geq V_{n+1}(\mathbf{Q}^1)$, i.e., \eqref{eqn:vn2} holds for $n+1$. Therefore, by induction, we can show that \eqref{eqn:vn2} holds for any $n$.  By taking limits on both sides of \eqref{eqn:vn2} and by \eqref{eqn:converge}, we complete the proof of Lemma~\ref{lemma:propertyV2}.
\section*{Appendix H: Proof of Lemma~\ref{lemma:propertyJ2}}\label{app:propertyJ2}
By \eqref{eqn:state_action_func}, we have
%\begin{eqnarray}
% &&~J(\mathbf{Q},u)-J(\mathbf{Q},v)-\left(J(\mathbf{Q}+\mathbf{E}_{u,k},u)-J(\mathbf{Q}+\mathbf{E}_{u,k},v)\right)\nonumber\\
%&&=\mathbb{E}[V(\mathbf{Q}^{1'})]+g(\mathbf{Q},u)-\mathbb{E}[V(\mathbf{Q}^{2'})]-g(\mathbf{Q},v)\nonumber\\
%&&~-\left(\mathbb{E}[V(\mathbf{Q}^{3'})]+g(\mathbf{Q}+\mathbf{e}_u,u)-\mathbb{E}[V(\mathbf{Q}^{4'})]-g(\mathbf{Q}+\mathbf{e}_u,v)\right)\nonumber\\
%&&=\mathbb{E}[V(\mathbf{Q}^{1'})]-\mathbb{E}[V(\mathbf{Q}^{3'})]+\mathbb{E}[V(\mathbf{Q}^{4'})]-\mathbb{E}[V(\mathbf{Q}^{2'})]\nonumber\\
%&&~+w_p\left( p(u,k^\ddag(\mathbf{Q},u))-p(u,k^\ddag(\mathbf{Q}+\mathbf{E}_{u,k},u))\right)\nonumber\\
%&&~-w_p\left(p(v,k^\ddag(\mathbf{Q},v))-p(v,k^\ddag(\mathbf{Q}+\mathbf{E}_{u,k},v))\right),
%\end{eqnarray}
\begin{align}\label{eqn:prooflemma6}
 &J(\mathbf{Q},u)-J(\mathbf{Q},v)-J(\mathbf{Q}+\mathbf{E}_{u,k},u)+J(\mathbf{Q}+\mathbf{E}_{u,k},v)\nonumber\\
=&\mathbb{E}[V(\mathbf{Q}^{1'})]+g(\mathbf{Q},u)-\mathbb{E}[V(\mathbf{Q}^{2'})]-g(\mathbf{Q},v)-\mathbb{E}[V(\mathbf{Q}^{3'})]\nonumber\\
&-g(\mathbf{Q}+\mathbf{E}_{u,k},u))+\mathbb{E}[V(\mathbf{Q}^{4'})]+g(\mathbf{Q}+\mathbf{E}_{u,k},v)\nonumber\\
=&\mathbb{E}[V(\mathbf{Q}^{1'})]-\mathbb{E}[V(\mathbf{Q}^{3'})]+\mathbb{E}[V(\mathbf{Q}^{4'})]-\mathbb{E}[V(\mathbf{Q}^{2'})]\nonumber\\
&+w_p\left( p(u,k^\ddag(\mathbf{Q},u))-p(u,k^\ddag(\mathbf{Q}+\mathbf{E}_{u,k},u))\right)\nonumber\\
&-w_p\left(p(v,k^\ddag(\mathbf{Q},v))-p(v,k^\ddag(\mathbf{Q}+\mathbf{E}_{u,k},v))\right),
\end{align}
where $\mathbf{Q}=(Q_{m,i})_{m\in\mathcal{M},i\in\mathcal{K}}$ and $\mathbf{Q}^{j'}=(Q^{j'}_{m,i})_{m\in\mathcal{M},i\in\mathcal{K}}$, $j=1,2,3,4$ with
\begin{subequations}
\begin{align}
  &Q^{1'}_{m,i}=\min\{\mathbbm{1}(u\neq m)Q_{m,i}+A_{m,i},N_{m,i}\},\label{eqn:q1prime3}\\
  &Q^{2'}_{m,i}=\min\{\mathbbm{1}(v\neq m)Q_{m,i}+A_{m,i},N_{m,i}\},\label{eqn:q2prime3}\\
  &Q^{3'}_{m,i}=\left\{
                       \begin{array}{ll}
                         \min\{A_{u,k},N_{u,k}\}& \hbox{if $m=u, k=i$} \\
                         \min\{Q_{m,i}+A_{m,i},N_{m,i}\}&\hbox{otherwise}
                       \end{array}
                     \right.,\label{eqn:q3prime3}\\
  &Q^{4'}_{m,i}=\left\{
                       \begin{array}{ll}
                         \min\{Q_{u,k}+1+A_{u,k},N_{u,k}\}&\hbox{if $m=u, k=i$} \\
                         \min\{\mathbbm{1}(v\neq m)Q_{m,i}+A_{m,i},N_{m,i}\} & \hbox{otherwise}
                       \end{array}
                     \right.\label{eqn:q4prime3}
%&\mathbf{Q}^{3'}_{m,i}=\begin{cases}
%                         \min\{A_{u,k},N_{u,k}\}& \hbox{if $m=u, k=i$} \\
%                         \min\{\mathbbm{1}(u\neq m)Q_{m,i}+A_{m,i},N_{m,i}\}&\hbox{otherwise}
%\end{cases},\\
%&\mathbf{Q}^{4'}_{m,i}=\begin{cases}
%                         \min\{Q_{u,k}+1+A_{u,k},N_{u,k}\}& \hbox{if $m=u, k=i$} \\
%                         \min\{\mathbbm{1}(v\neq m)Q_{m,i}+A_{m,i},N_{m,i}\}&\hbox{otherwise}
%\end{cases}.
\end{align}
\end{subequations}
To prove Lemma~\ref{lemma:propertyJ2}, it remains to show that the R.H.S. of \eqref{eqn:prooflemma6} is nonnegative.
By comparing \eqref{eqn:q1prime3} with \eqref{eqn:q3prime3}, we can see that $Q^{1'}_{m,i}=Q^{3'}_{m,i}$ for all $m,i$,  i.e., $\mathbf{Q}^{1'}=\mathbf{Q}^{3'}$. Thus, we have $\mathbb{E}[V(\mathbf{Q}^{1'})]=\mathbb{E}[V(\mathbf{Q}^{3'})]$.
By comparing \eqref{eqn:q2prime3} with \eqref{eqn:q4prime3}, we can see that $Q^{4'}_{u,k}\geq Q^{2'}_{u,k}\geq Q^{2}_{u,k}>0$ and $Q^{4'}_{m,i}=Q^{2'}_{m,i}$ for all $m\neq u, i\neq k$, i.e., $\mathbf{Q}^{4'}\trianglerighteq\mathbf{Q}^{2'}$. Thus, by Lemma~\ref{lemma:propertyV2}, we have $\mathbb{E}[V(\mathbf{Q}^{4'})]\geq\mathbb{E}[V(\mathbf{Q}^{2'})]$.
In addition, due to $\mathbf{Q}+\mathbf{E}_{u,k}\trianglerighteq\mathbf{Q}$, we have $k\leq\max\{k|Q_{u,k}>0\}$. Then, we have $k^\ddag(\mathbf{Q}+\mathbf{E}_{u,k},u)=k^\ddag(\mathbf{Q},u)$ and $k^\ddag(\mathbf{Q}+\mathbf{E}_{u,k},v)=k^\ddag(\mathbf{Q},v)$, implying that $p(u,k^\ddag(\mathbf{Q},u))=p(u,k^\ddag(\mathbf{Q}+\mathbf{E}_{u,k},u))$ and $p(u,k^\ddag(\mathbf{Q},v))=p(u,k^\ddag(\mathbf{Q}+\mathbf{E}_{u,k},v))$.
Therefore, by \eqref{eqn:prooflemma6}, we have $J(\mathbf{Q}+\mathbf{E}_{u,k},u)-J(\mathbf{Q}+\mathbf{E}_{u,k},v)\leq J(\mathbf{Q},u)-J(\mathbf{Q},v)$.
We complete the proof of Lemma~\ref{lemma:propertyJ2}.
\section*{Appendix I: Proof of Theorem~\ref{theorem:theorem2}}\label{app:theorem2}
Consider content $u\in\mathcal{M}$, user $k\in\mathcal{K}$ and state $\mathbf{Q}=(Q_{m,i})_{m\in\mathcal{M},i\in\mathcal{K}}$ where $Q_{u,k}=s_{u,k}(\mathbf{Q}_{-u,-k})$.
Note that, if $s_{u,k}(\mathbf{Q}_{-u,-k})=\infty$, \eqref{eqn:switch2} always holds.
Therefore, in the following, we only consider that $s_{u,k}(\mathbf{Q}_{-u,-k})<\infty$.
According to the definition of $s_{u,k}(\mathbf{Q}_{-u,-k})$ in Theorem~\ref{theorem:theorem2}, we can see that $J(\mathbf{Q},u)\leq J(\mathbf{Q},v)$ for all $v\in\mathcal{M}, v\neq u$.
Thus, it is optimal to multicast content $u$ for state $\mathbf{Q}$, $\mu^*(\mathbf{Q})=u$. Consider another state $\mathbf{Q}'=(Q_{m,i}')_{m\in\mathcal{M},i\in\mathcal{K}}$ where $Q_{u,k}'\geq Q_{u,k}$ and $Q_{m,i}'=Q_{m,i}$ for all $(m,i)\neq (u,k)$. To prove Theorem~\ref{theorem:theorem2}, it is equivalent to show that it is also optimal to multicast content $u$ for state $\mathbf{Q}'$, i.e.,
\begin{equation}
  J(\mathbf{Q}',u)\leq J(\mathbf{Q}',v), \forall v\in\mathcal{M}, v\neq u.
\end{equation}
According to the relationship between $k$ and $k^\dag(k,\mathbf{Q}_u)$ as well as the value of $s_{u,k}(\mathbf{Q}_{-u,-k})$, we have the following three cases.

(1) If $k<k^\dag(k,\mathbf{Q}_u)$, i.e., condition (a) holds, we have $k<\max\{k|Q_{u,k}>0\}$. By Lemma~\ref{lemma:propertyJ2}, for any $v\in\mathcal{M}$ and $v\neq u$, we have
\begin{equation}\label{eqn:proofthm2}
  J(\mathbf{Q}',u)-J(\mathbf{Q}',v)\leq J(\mathbf{Q},u)-J(\mathbf{Q},v)\leq 0.
\end{equation}
Thus, it is optimal to multicast content $u$ for state $\mathbf{Q}'$.

(2) If $k>k^\dag(k,\mathbf{Q}_u)$ and $s_{u,k}(\mathbf{Q}_{-u,-k})>0$, i.e. condition (b) holds, we have $k=\max\{k|Q_{u,k}>0\}$. By Lemma~\ref{lemma:propertyJ2}, for any $v\in\mathcal{M}$ and $v\neq u$, \eqref{eqn:proofthm2} also holds. Thus, it is optimal to multicast content $u$ for state $\mathbf{Q}'$.

(3) If $k>k^\dag(k,\mathbf{Q}_u)$ and $s_{u,k}(\mathbf{Q}_{-u,-k})=0$, implying that $k>\max\{k|Q_{u,k}>0\}$, then Lemma~\ref{lemma:propertyJ2} does not apply and it is unknown whether \eqref{eqn:proofthm2} holds. Therefore, it is unclear whether it is optimal to multicast content $u$ for state $\mathbf{Q}'$.

We complete the proof of Theorem~\ref{theorem:theorem2}.
\section*{Appendix J: Proof of Lemma~\ref{lemma:decomposition}}\label{app:decomposition}
Following the proof in \cite{harvest}, we shall prove the additive property w.r.t. the value function.
%. Then, we have the per-content fixed point equation in terms of $(\hat{\theta}_m,\{\hat{V}_m(Q_m)\})$, given in \eqref{eqn:perrandombellman}. Under $\hat{\mu}$, the induced Markov chain %has a single recurrent class. Therefore, the solutions to \eqref{eqn:randombellman} and \eqref{eqn:perrandombellman} exist, respectively.
First, we have $g(\mathbf{Q},u)=\sum_{m\in\mathcal{M}}g_m(Q_m,u)$. Second, by the relationship between the joint distribution and the marginal distribution, we have $\sum_{\mathbf{Q}\in\bm{\mathcal{Q}}}\Pr[\mathbf{Q}'|\mathbf{Q},u]=\sum_{Q_m'\in\mathcal{Q}_m}\Pr[Q_m'|\mathbf{Q},u]=\sum_{Q_m'\in\mathcal{Q}_m}\Pr[Q_m'|Q_m,u]$.
Therefore, by substituting $\hat{\theta}=\sum_{m\in\mathcal{M}}\hat{\theta}_m$ and $\hat{V}(\mathbf{Q})=\sum_{m\in\mathcal{M}}\hat{V}_m(Q_m)$ into \eqref{eqn:randombellman}, we can see that the equality holds, which completes the proof.

\section*{Appendix K: Proof of Theorem~\ref{theorem:strhatmu}}\label{app:theorem3}
We first prove the structural property of $\hat{\mu}^*$ for the uniform case.
First, we show that for all $m\in\mathcal{M}$, the per-content value function $\hat{V}_m(Q_m)$ satisfies
    \begin{equation}
  \hat{V}_m(Q_m^2)\geq \hat{V}_m(Q_m^1),\label{eqn:propertyvm}
    \end{equation}
for any $Q_m^1$, $Q_m^2\in\mathcal{Q}_m$ such that $Q_m^2\geq Q_m^1$.
For each $Q_m\in\mathcal{Q}_m$, in the $n$-th iteration, RVIA updates $\hat{V}_m^{n+1}(Q_m)$ according to
\begin{align}
  &\hat{V}_m^{n+1}(Q_m)=\mathbb{E}^{\hat{\mu}}\left[g_m(Q_m,u)\right]\nonumber\\
    &+\sum_{Q_m'\in\mathcal{Q}_m}\mathbb{E}^{\hat{\mu}}\left[\Pr[Q_m'|Q_m,u]\right]\hat{V}_m(Q_m')-\hat{V}_m^{n+1}(Q_m^\S),
\end{align}
where $Q_m^\S\in\mathcal{Q}_m$ is some fixed state.
Following the proof of Lemma~\ref{lemma:propertyV1}, we can show that for any $Q_m^1$, $Q_m^2\in\mathcal{Q}_m$ such that $Q_m^2\geq Q_m^1$, we have $\hat{V}_m^n(Q_m^2)\geq \hat{V}_m^n(Q_m^1)$ for all $n=0,1,\cdots$. Therefore, we can show that \eqref{eqn:propertyvm} holds through mathematical induction using RVIA.
Then, following the proof of Lemma~\ref{lemma:propertyJ1}, we can show that for any $u,v\in\mathcal{M}$ and $u\neq v$, $\hat{J}(\mathbf{Q},u)-\hat{J}(\mathbf{Q},v)$ is monotonically non-increasing with $Q_u$, i.e.,
\begin{equation}
  \hat{J}(\mathbf{Q}+\mathbf{e}_u,u)-\hat{J}(\mathbf{Q}+\mathbf{e}_u,v)\leq \hat{J}(\mathbf{Q},u)-\hat{J}(\mathbf{Q},v).\label{eqn:subpropertyJ1}
\end{equation}
Finally, following the proof of Theorem~\ref{theorem:theorem1}, we can show that $\hat{\mu}^*$ has a switch structure in the uniform case. We complete the proof of Part 1) in Theorem~\ref{theorem:strhatmu}.

Next, we prove the structural property of $\hat{\mu}^*$ for the nonuniform case. The procedure is similar to that for the uniform case.
First, similar to the proofs for \eqref{eqn:propertyvm} and Lemma~\ref{lemma:propertyV2}, we can show that for all $m\in\mathcal{M}$, the per-content value function $\hat{V}_m(Q_m)$ satisfies
    \begin{equation}
  \hat{V}_m(Q_m^2)\geq \hat{V}_m(Q_m^1),\label{eqn:propertyvm2}
    \end{equation}
  for any $Q_m^1$, $Q_m^2\in\mathcal{Q}_m$ such that $Q_m^2\trianglerighteq Q_m^1$.
Then, following the proof of Lemma~\ref{lemma:propertyJ2}, we can show that for any $u,v\in\mathcal{M}$ and $u\neq v$, $k\leq \max\{k|Q_{u,k}>0\}$,
\begin{equation}
  \hat{J}(\mathbf{Q}+\mathbf{E}_{u,k},u)-\hat{J}(\mathbf{Q}+\mathbf{E}_{u,k},v)\leq \hat{J}(\mathbf{Q},u)-\hat{J}(\mathbf{Q},v). \label{eqn:subpropertyJ2}
\end{equation}
Finally, following the proof of Theorem~\ref{theorem:theorem2}, we can show that $\hat{\mu}^*$ has a partial switch structure in the nonuniform case.
We complete the proof of Part 2) in Theorem~\ref{theorem:strhatmu}.

\bibliographystyle{IEEEtran}
\bibliography{IEEEabrv,caching}

\begin{thebibliography}{10}
\providecommand{\url}[1]{#1}
\csname url@samestyle\endcsname
\providecommand{\newblock}{\relax}
\providecommand{\bibinfo}[2]{#2}
\providecommand{\BIBentrySTDinterwordspacing}{\spaceskip=0pt\relax}
\providecommand{\BIBentryALTinterwordstretchfactor}{4}
\providecommand{\BIBentryALTinterwordspacing}{\spaceskip=\fontdimen2\font plus
\BIBentryALTinterwordstretchfactor\fontdimen3\font minus
  \fontdimen4\font\relax}
\providecommand{\BIBforeignlanguage}[2]{{%
\expandafter\ifx\csname l@#1\endcsname\relax
\typeout{** WARNING: IEEEtran.bst: No hyphenation pattern has been}%
\typeout{** loaded for the language `#1'. Using the pattern for}%
\typeout{** the default language instead.}%
\else
\language=\csname l@#1\endcsname
\fi
#2}}
\providecommand{\BIBdecl}{\relax}
\BIBdecl

\bibitem{Cisco}
Cisco, ``Cisco visual networking index: Global mobile data traffic forecast
  update, 2014-2019,'' \emph{White Paper, February}, 2015.

\bibitem{liuhui}
H.~Liu, Z.~Chen, X.~Tian, X.~Wang, and M.~Tao, ``On content-centric wireless
  delivery networks,'' \emph{{IEEE} Wireless Commun. Mag.}, vol.~21, no.~6, pp.
  118--125, 2014.

\bibitem{femto}
K.~Shanmugam, N.~Golrezaei, A.~Dimakis, A.~Molisch, and G.~Caire,
  ``Femtocaching: Wireless content delivery through distributed caching
  helpers,'' \emph{{IEEE} Trans. Inf. Theory}, vol.~59, no.~12, Dec 2013.

\bibitem{TassiulasTCOM}
K.~Poularakis, G.~Iosifidis, and L.~Tassiulas, ``Approximation algorithms for
  mobile data caching in small cell networks,'' \emph{{IEEE} Trans. Commun.},
  vol.~62, no.~10, pp. 3665--3677, Oct 2014.

\bibitem{mimo}
A.~Liu and V.~Lau, ``Exploiting base station caching in {MIMO} cellular
  networks: Opportunistic cooperation for video streaming,'' \emph{{IEEE}
  Trans. Signal Process.}, vol.~63, no.~1, pp. 57--69, Jan 2015.

\bibitem{Bastug2015}
E.~Bastug, M.~Bennis, M.~Kountouris, and M.~Debbah, ``Cache-enabled small cell
  networks: modeling and tradeoffs,'' \emph{EURASIP Journal of Wireless
  Communications and Networking}, vol. 2015, p.~1, 2015.

\bibitem{embms}
D.~Lecompte and F.~Gabin, ``Evolved multimedia broadcast/multicast service
  ({eMBMS}) in {LTE}-advanced: overview and {Rel}-11 enhancements,''
  \emph{{IEEE} Commun. Mag.}, vol.~50, no.~11, pp. 68--74, 2012.

\bibitem{HouTON}
I.-H. Hou, ``Broadcasting delay-constrained traffic over unreliable wireless
  links with network coding,'' \emph{{IEEE/ACM} Trans. Netw.}, vol.~23, no.~3,
  pp. 728--740, June 2015.

\bibitem{multicast_MIT}
K.~S. Kim, C.~ping Li, and E.~Modiano, ``Scheduling multicast traffic with
  deadlines in wireless networks,'' in \emph{Proc. IEEE INFOCOM}, April 2014.

\bibitem{multicast_capacity}
S.~Zhou and L.~Ying, ``On delay constrained multicast capacity of large-scale
  mobile ad-hoc networks,'' in \emph{Proc. IEEE INFOCOM}, March 2010.

\bibitem{7249208}
U.~Niesen and M.~Maddah-Ali, ``Coded caching for delay-sensitive content,'' in
  \emph{Proc. IEEE ICC}, June 2015.

\bibitem{TWC16}
K.~Poularakis, G.~Iosifidis, V.~Sourlas, and L.~Tassiulas, ``Exploiting caching
  and multicast for {5G} wireless networks,'' \emph{IEEE Trans. Wireless
  Commun.}, vol.~PP, no.~99, pp. 1--1, 2016.

\bibitem{ton}
N.~Abedini and S.~Shakkottai, ``Content caching and scheduling in wireless
  networks with elastic and inelastic traffic,'' \emph{{IEEE/ACM} Trans.
  Netw.}, vol.~22, no.~3, pp. 864--874, June 2014.

\bibitem{bertsekas}
D.~P. Bertsekas, \emph{Dynamic programming and optimal control, 3rd edition,
  volume II}, 2011.

\bibitem{cdc}
R.~Gummadi, ``Optimal control of a broadcasting server,'' in \emph{Proc. 48th
  IEEE Conf. Decision Control (CDC/CCC)}, Dec. 2009.

\bibitem{batch}
C.~H. Xia, G.~Michailidis, N.~Bambos, and P.~W. Glynn, ``Optimal control of
  parallel queues with batch service,'' \emph{Probability in the Engineering
  and Informational Sciences}, vol.~16, no.~03, pp. 289--307, 2002.

\bibitem{switch}
G.~Koole, ``Assigning a single server to inhomogeneous queues with switching
  costs,'' \emph{Theoretical Computer Science}, vol. 182, no.~1, pp. 203--216,
  1997.

\bibitem{MPIA}
D.~V. Djonin and V.~Krishnamurthy, ``{MIMO} transmission control in fading
  channels--a constrained {Markov} decision process formulation with monotone
  randomized policies,'' \emph{{IEEE} Trans. Signal Process.}, vol.~55, no.~10,
  pp. 5069--5083, 2007.

\bibitem{OR}
A.~H. Elwany, N.~Z. Gebraeel, and L.~M. Maillart, ``Structured replacement
  policies for components with complex degradation processes and dedicated
  sensors,'' \emph{Oper. Res.}, vol.~59, no.~3, pp. 684--695, 2011.

\bibitem{powell2007approximate}
W.~B. Powell, \emph{Approximate Dynamic Programming: Solving the curses of
  dimensionality}.\hskip 1em plus 0.5em minus 0.4em\relax John Wiley \& Sons,
  2007, vol. 703.

\bibitem{factoredMDP}
C.~Guestrin, D.~Koller, R.~Parr, and S.~Venkataraman, ``Efficient solution
  algorithms for factored mdps,'' \emph{Journal of Artificial Intelligence
  Research}, pp. 399--468, 2003.

\bibitem{Neely}
D.~Bethanabhotla, G.~Caire, and M.~Neely, ``Adaptive video streaming for
  wireless networks with multiple users and helpers,'' \emph{{IEEE} Trans.
  Commun.}, vol.~63, no.~1, pp. 268--285, Jan 2015.

\bibitem{ITC}
C.~Fricker, P.~Robert, and J.~Roberts, ``A versatile and accurate approximation
  for lru cache performance,'' in \emph{Proc. ITC}, 2012.

\bibitem{puterman}
M.~L. Puterman, \emph{Markov decision processes: discrete stochastic dynamic
  programming}.\hskip 1em plus 0.5em minus 0.4em\relax John Wiley \& Sons,
  2009, vol. 414.

\bibitem{Koole}
G.~Koole, \emph{Monotonicity in Markov reward and decision chains: Theory and
  applications}.\hskip 1em plus 0.5em minus 0.4em\relax Now Publishers Inc,
  2007.

\bibitem{complexity}
M.~L. Littman, T.~L. Dean, and L.~P. Kaelbling, ``On the complexity of solving
  markov decision problems,'' in \emph{Proceedings of the Eleventh conference
  on Uncertainty in artificial intelligence}, 1995, pp. 394--402.

\bibitem{harvest}
Y.~Cui, V.~Lau, and Y.~Wu, ``Delay-aware {BS} discontinuous transmission
  control and user scheduling for energy harvesting downlink coordinated {MIMO}
  systems,'' \emph{{IEEE} Trans. Signal Process.}, vol.~60, no.~7, pp.
  3786--3795, July 2012.

\bibitem{zipf}
L.~Breslau, P.~Cao, L.~Fan, G.~Phillips, and S.~Shenker, ``Web caching and
  {Zipf-like} distributions: evidence and implications,'' in \emph{Proc. IEEE
  INFOCOM}, March 1999.

\end{thebibliography}

\end{document}